\documentclass[letterpaper,11pt]{article}

\usepackage{fullpage}
\usepackage{authblk}

\usepackage{appendix}
\usepackage{graphicx}
\usepackage{dcolumn}
\usepackage{multirow}
\usepackage{amsmath}
\usepackage{amsthm}
\usepackage{amssymb}
\usepackage{appendix}

\newcommand{\bi}{\begin{itemize}}
\newcommand{\ei}{\end{itemize}}

\newcommand{\be}{\begin{enumerate}}
\newcommand{\ee}{\end{enumerate}}

\newcommand{\beq}{\begin{equation}}
\newcommand{\eeq}{\end{equation}}

\newcommand{\beqnr}{\begin{eqnarray}}
\newcommand{\eeqnr}{\end{eqnarray}}

\newtheorem{assumption}{ASSUMPTION}
\newtheorem{lemma}{Lemma}
\newtheorem{proposition}{Proposition}
\newtheorem{corollary}{Corollary}
\newtheorem{theorem}{Theorem}
\newtheorem{definition}{Definition}

\DeclareMathOperator*{\argmin}{arg\,min}

\begin{document}


\title{The effect of round-off error on long memory processes}


\author[1]{Gabriele La Spada\thanks{Corresponding author. 
E-mail: gla@princeton.edu.}} 
\affil{\it Department of Economics, Princeton University, Princeton, NJ 08544-1021, USA}

\author[2, 3, 4]{Fabrizio Lillo}
\affil[2]{\it Scuola Normale Superiore di Pisa, Piazza dei Cavalieri 7, 56126 Pisa, Italy}
\affil[3]{\it Dipartimento di Fisica, Universit\`a di Palermo, Viale delle Scienze, I-90128, Palermo, Italy}
\affil[4]{\it Santa Fe Institute, 1399 Hyde Park Road, Santa Fe, NM 87501, USA}

\maketitle

\begin{abstract}
We study how the round-off (or discretization) error changes the statistical properties of a Gaussian long memory process. We show that the autocovariance and the spectral density of the discretized process are asymptotically rescaled by a factor smaller than one, and we compute exactly this scaling factor. Consequently, we find that the discretized process is also long memory with the same Hurst exponent as the original process. We consider the properties of two estimators of the Hurst exponent, namely the local Whittle (LW) estimator and the Detrended Fluctuation Analysis (DFA). By using analytical considerations  and numerical simulations we show that, in presence of round-off error, both estimators are severely negatively biased in finite samples. Under regularity conditions we prove that the LW estimator applied to discretized processes is consistent and asymptotically normal. Moreover, we compute the asymptotic properties of the DFA for a generic (i.e. non Gaussian) long memory process and we apply the result to discretized  processes.
\end{abstract}


\newpage

\section{Introduction}
\label{intro}

``All economic data is discrete'' (Engle and Russell (2004)). Round-off errors
occur whenever a real valued process is observed on a grid of discrete values.
A special case of round-off error is obtained by taking the sign of a real valued process. Round-off errors change the properties of the original stochastic process. Few exact results exist on the effect of round-off error on stochastic processes. Delattre and Jacod (1997), for example, proved a central limit theorem on a Brownian motion with round-off error sampled at discrete times. In this paper we study analytically and numerically the effect of round-off error on long memory processes. We will use the term {\it discretized process} to refer to the process with round-off error and the term {\it discretization} to refer to the rounding procedure.

Round-off errors can be either due to a limit in the resolution of the observing device or to the fact that an underlying real valued process can manifest itself only as a discrete valued process. When seen as a resolution effect, round-off error can also be considered as a special case of measurement error. One recent strand of econometric literature considers the problem of estimation of a process that can be observed only with some noise due to the measurement process. The typical modeling approach is to consider the measurement error as an additive white noise process uncorrelated with the unobserved latent process (see, for example, Hansen and Lunde (2010)). Round-off error can be considered as a different form of measurement error, which is a deterministic function of the underlying latent process and it is not uncorrelated with it. 

An important example of round-off error due to the fact that the process manifests itself only as a discrete valued process is the dynamics of asset prices. Despite the fact that the price of an asset is a real number, transaction prices (but also quote prices) can assume only values which are multiple of a minimum value called tick size. 
For transaction by transaction data the tick over price ratio can be large leading to a price dynamics that cannot be even approximated as a real valued process. In the market microstructure literature the round-off error is one of the main sources of disturbance in the estimation of integrated volatility. Several papers have considered how tick size affects the diffusion dynamics of price on short time scales. For example, Gottlieb and Kalay (1985) and Harris (1990) developed microstructure models to investigate the effect of price discretization on return variance and serial correlation. More recent papers consider the effect of price discretization on phase portrait of returns (Szpiro (1998)), on price-dividend relation (Bali and Hite (1998)), and on integrated volatility (Rosenbaum (2009)). See also La Spada {\it et al.} (2011) and references therein for a recent review of the effect of tick size on the diffusion properties of financial asset prices.
Beside market microstructure, 
discretized processes emerge naturally in discrete choice models, binned data, computer vision, detectors, and digital signal processing. 
Recently, the effect of discretization has also been studied in the framework of rational inattention (Saint-Paul (2011)).
As an example of a financial application of long memory processes when considering the sign of a variable, we mention the work of Lillo and Farmer (2004). There authors empirically show that the sign of the market order flow in a double auction financial market is a long memory process.

Long memory processes are ubiquitous in natural, social, and economic systems (see Beran (1994)).  However, a large part of the theory, statistics, and modeling of long memory processes is based on the assumption of normality of the distribution, and very often also on the hypothesis of linearity of the generating mechanism. A major step toward generalization has been the paper by Dittmann and Granger (2002) that investigated the properties of a process obtained from a non-linear transformation of a Gaussian long memory process. Their main result is that when the transformation can be written as a {\it finite} linear combination of Hermite polynomials, the transformed process has the same or a smaller Hurst exponent of the original process, depending on the Hermite rank of the transformation (see below for a more precise definition of these terms). However, their results ``need not hold for  transformations with infinite Hermite expansion" (see Dittmann and Granger (2002), Proposition 1). Moreover, recently, several papers have also studied the asymptotic properties of common estimators of the Hurst exponent when they are applied to non-Gaussian and non-linear time series (e.g., Dalla, Giraitis and Hidalgo (2006); see Section~\ref{estimation} for a detailed literature review). 
The present paper contributes in this direction by presenting several results for a specific, yet important, class of non-linear transformations of a Gaussian long memory process with an {\it infinite} Hermite expansion, namely the class of discretization transformations. 


In this paper we present  an in-depth analysis of the properties of a stochastic process obtained from the discretization of a large class of Gaussian long memory processes. We give the asymptotic behavior of the autocovariance and of the spectral density, by computing explicitly the leading term and the order of the second term in a series expansion. We find that the autocovariance and the autocorrelation are asymptotically rescaled by a factor smaller than one, and the Hurst exponent is the same for the continuous and the discretized process. The spectral density is also rescaled for small frequencies by the same scaling factor as the autocovariance. We find an explicit closed form of this scaling factor. It is worth noting that the decrease of the autocorrelation function holds for the discretization of any Gaussian weakly stationary process, either long-memory or short-memory. Our results are consistent with those in Hansen and Lunde (2010) on the estimation of the persistence and the autocorrelation function of a time series measured with error.

We then consider two classic methods to estimate the Hurst exponent, namely the local Whittle (LW) estimator suggested by K\"unsch (1987) and the Detrended Fluctuation Analysis (DFA) introduced by Peng {\it et al.} (1994). The LW estimator is a very popular semiparametric method for investigating the long memory properties of a stochastic process. It has been extensively studied in econometric theory and frequently used in empirical work in financial econometrics. However, most of the results in the literature do not hold for the discretized process (see discussion in Section~\ref{estimation}). Recently, Dalla, Giraitis and Hidalgo (2006) proved consistency and asymptotic normality of the LW estimator for a general class of nonlinear processes, but, as discussed in detail in Section~\ref{estimation}, we cannot directly apply their results to the class of processes considered in this paper. As we show below, under suitable regularity conditions we are able to extend the results in Dalla {\it et al.} (2006) to our setting and prove that the LW estimator is consistent and asymptotically normal for Gaussian long memory processes observed with roundoff error. Moreover, we also show that one of the main results of Dalla {\it et al.} (2006) can be generalized by relaxing some assumptions that are not necessary. 

The DFA (Peng {\it et al.} (1994)) is another very popular semiparametric method for the investigation of long memory properties of generic processes. It was introduced more than fifteen years ago to investigate physiological data, in particular the heartbeat signal. Since its introduction it has been applied to a large variety of systems, including physical, biological, economic, and technological data. In economics and finance it has been applied for example in Schmitt {\it et al.} (2000), Lillo and Farmer (2004), Di Matteo {\it et al.} (2005), Yamasaki {\it et al.} (2005), and Alfarano and Lux (2007). In a recent paper Bardet and Kammoun (2008) computed explicitly the  asymptotic properties of the DFA for the fGn and for a general class of Gaussian weakly stationary long memory processes. Here we generalize their results to a non-Gaussian generic long memory process and we applied them to the discretized process. As a byproduct we show that the order of the error of the root-mean-square fluctuation given by Theorem 4.2 of Bardet and Kammoun (2008)  is not correct for a generic Gaussian process.  Because of the cancellation of a term the theorem is correct for the fGn and the fARIMA process as claimed in Bardet and Kammoun (2008). By comparing the root-mean-square fluctuation of the discretized process to that of the continuous process we argue that the second-order term induces a negative bias in the estimation of the Hurst exponent. 


The paper is organized as follows. In Section~\ref{longmemo} we define the class of long memory processes we consider in the present paper. In Section~\ref{nullmodel} we derive some analytical results on the distribution, autocovariance, and spectral density of the discretized process. Section~\ref{estimation} presents analytical and numerical results on the estimation of the Hurst exponent obtained by using the LW estimator and the DFA.  Finally, Section~\ref{conclusion} concludes. 

%

\section{Long memory processes}
\label{longmemo}

In this paper we are interested in studying the effect of round-off error on a long memory process.
There are several possible definitions of a long memory process.  The very general definition we are using in the present paper is the following.
\begin{definition}\label{generalDef}
A discrete time weakly stationary stochastic process $\{X(t)\}_{t \in \mathbb{N}}$ is long memory if its autocovariance function $\gamma(k)$ behaves as 
\beq
\gamma(k)= k^{2 H - 2} L(k) \qquad for \;\; k\geq 1\,,\label{ACVgeneralDef}
\eeq
where $H\in(0.5, 1)$ and $L(k)$ is a slowly varying function  at infinity\footnote{$L(x)$ is a slowly varying function if $\lim_{x \to \infty} L(tx)/L(x) = 1$, $\forall t>0$ (see Embrechts {\it et al.} (1997)).  In the definitions above, and for the purposes of this paper, we are considering only positively correlated long-memory processes. Negatively correlated long-memory processes also exist, but the long-memory processes we will consider in the rest of the paper are all positively correlated.}.
\end{definition}
The parameter $H$ is called Hurst exponent or, sometimes, the self-similarity parameter. Under this definition of long memory the autocovariance function does not necessarily decay as a pure power-law. Consider the case $L(k)=\log k$, or any power of the logarithm function.
The autocorrelation function of a long memory process is not integrable on the positive real line and, as a consequence, the process does not have a typical time scale. 
 

%


The class of long memory processes defined above is quite large due to the arbitrariness of the slowly varying function $L(k)$. Some of the results we will present below hold for a more restricted class of long-memory processes characterized by the properties of the slowly varying function in Definition~\ref{generalDef}. 

Following Embrechts {\it et al.} (1997) we denote with $\mathcal{R}_0$ the set of slowly varying functions at infinity.
We introduce the following definition
\begin{definition}\label{wellbehavedDef}
We define the set of well behaved slowly varying function as
\beq
\mathcal{L}(K, I, b_i, \beta_i) \equiv \left\{ l \in \mathcal{R}_0:  \exists \, K>0 \; s.t. \;\;  l(k) =  \sum_{i=0}^{I} b_i k^{-\beta_i}  \quad \forall\, k > K\right\}
\eeq
with $I \leq \infty$, $b_0>0$, $b_i\neq0$ $\forall \, i\geq 1$,  $\beta_0=0$, and $\beta_{i} < \beta_{i+1}$ $\forall \, i \geq 0$. Moreover, if $I=\infty$, the parameters $b_i$ and $\beta_i$ are such that  the series $\sum_{i=0}^{\infty} b_i k^{-\beta_i}$  converges absolutely $\forall\, k > K$.
\end{definition}

Note that all the slowly varying functions which are analytic at infinity are well behaved with $I=\infty$. Hereafter, we abbreviate $\mathcal{L}(K, I, b_i, \beta_i)$ as $\mathcal{L}$.

We are interested here in the discretization of long memory processes  characterized by a Gaussian distribution.  A stationary Gaussian process is completely characterized by the mean (hereafter assumed to be zero), the variance $D$ and the autocorrelation function $\rho(k)=\gamma(k)/D$. Two classes of stationary Gaussian long memory processes are often considered in the literature. The first one is the fractional Gaussian noise (fGn) (see Mandelbrot and van Ness (1968)), characterized by the autocorrelation function
\begin{equation}
\rho(k) =\frac{1}{2} [(k+1)^{2H}+(k-1)^{2H} -2k^{2H}].
\label{eq:fgnACF}
\end{equation}
For large $k$ the asymptotic expansion of (\ref{eq:fgnACF}) is 
\begin{equation}
\rho(k)=\frac{H(2H-1) }{k^{2-2H}}\left(1 +\frac{(2H-2)(2H-3)}{12}\frac{1}{k^{2}}+~\ldots\right).\nonumber
\end{equation}

The second important example is the  fARIMA$(0,d,0)$ process\footnote{In the following we consider only fARIMA$(0,d,0)$. Therefore, when in the following we refer to fARIMA, we mean fARIMA$(0,d,0)$. For results on more general fARIMA see Section \ref{discussionH}.}, where $d=H-0.5$, whose autocorrelation is
\begin{equation}\label{eq:farimaACF}
\rho(k)=\frac{\Gamma(3/2-H) \Gamma(H+k-1/2)}{\Gamma(H-1/2) \Gamma(k+3/2-H)},
\end{equation}
where $\Gamma(\cdot)$ is the gamma function. The asymptotic expansion of (\ref{eq:farimaACF}) is 
\begin{equation}
\rho(k) = \frac{\Gamma(3/2-H)}{\Gamma(H-1/2)}\frac{1}{k^{2-2H}} \left(1-\frac{(4H^3-12H^2+11H-3)}{12}\frac{1}{k^2}+\ldots\right) \nonumber.
\end{equation}

Note that both for the fGn and for the fARIMA process the slowly varying function $L(k)$ is analytic at infinity and therefore $L \in \mathcal{L}$ with $I=\infty$. In the following we will present results for the discretization of generic stationary Gaussian long memory processes, and we will consider the fGn or the fARIMA as special cases.

\section{The discretized process}
\label{nullmodel}

Given a discrete-time real-valued process  $\{X(t)\}_{t \in \mathbb{N}}$ and  a grid of points $ j\delta$ with $j \in \mathbb{Z}$ and $\delta >0$, the discretized process at time $t$ is $X_d(t)=[X(t)/\delta] \delta$, where $[z]$ is the integer part of $z$. The parameter $\delta$ sets the level of round-off error. This type of discretization appears, for example, whenever a weakly stationary process is discretized through a binning procedure. 

The probability mass function of the discretized process is
\begin{equation}
p_d(x)=\sum_{n=-\infty}^{\infty} q_n \delta_D(x-n\delta)\,, \qquad \text{where} \qquad q_n=\int_{(n-1/2)\delta}^{(n+1/2)\delta}p(x)dx\label{eq:pmf}
\end{equation}
%
$p(x)$ is the probability density function of $X_t$ and $\delta_D(x)$ is the Dirac delta function.

It is useful to introduce the adimensional scaling variable
\begin{equation}
\chi=\frac{D}{\delta^2}
\end{equation}
Since $X(t)$ is Gaussian distributed, the variance $D_d$ of the discretized process can be calculated explicitly. For detailed analytical results on the distributional properties of $X_d(t)$ see Appendix \ref{App:DistrProp}. Left panel of Figure \ref{varFigureIncr} shows the ratio $D_d/D$ as a function of the scaling parameter $\chi$. It is worth noting that  this ratio is not monotonic. For small $\chi$ the ratio goes to zero because for $\delta \gg D$ essentially all the probability mass falls in the bin centered at zero. Finally, the parameter $\chi$ sets the fraction $q_0$ of points which in the discretized process have value zero.  It is direct to show that, if $X(t)$ is Gaussian, $q_0=erf\left[1/2\sqrt{2\chi}\right]$, where $erf[\cdot]$ is the error function. This is clearly a monotonically decreasing function. In the numerical examples below we will use $\chi=0.1$, $0.25$, and $0.5$ corresponding to $q_0=0.886$, $0.683$, and $0.521$, respectively. 

A different type of discretization that we will consider below is obtained by taking the sign of $X_t$. 
Assuming that the distribution function of $X_t$ is absolutely continuous with respect to the Lebesgue measure in a neighborhood of $0$,  so
that the event $X(t)=0$ has zero probability, this discretization leads
to $X_s(t)=sign(X(t))=\pm 1$ with probability 1. 
We discuss the sign transformation in more detail in Appendix~\ref{App:Sign}.



\begin{figure}[ht]
\begin{center}
\includegraphics[scale=0.28]{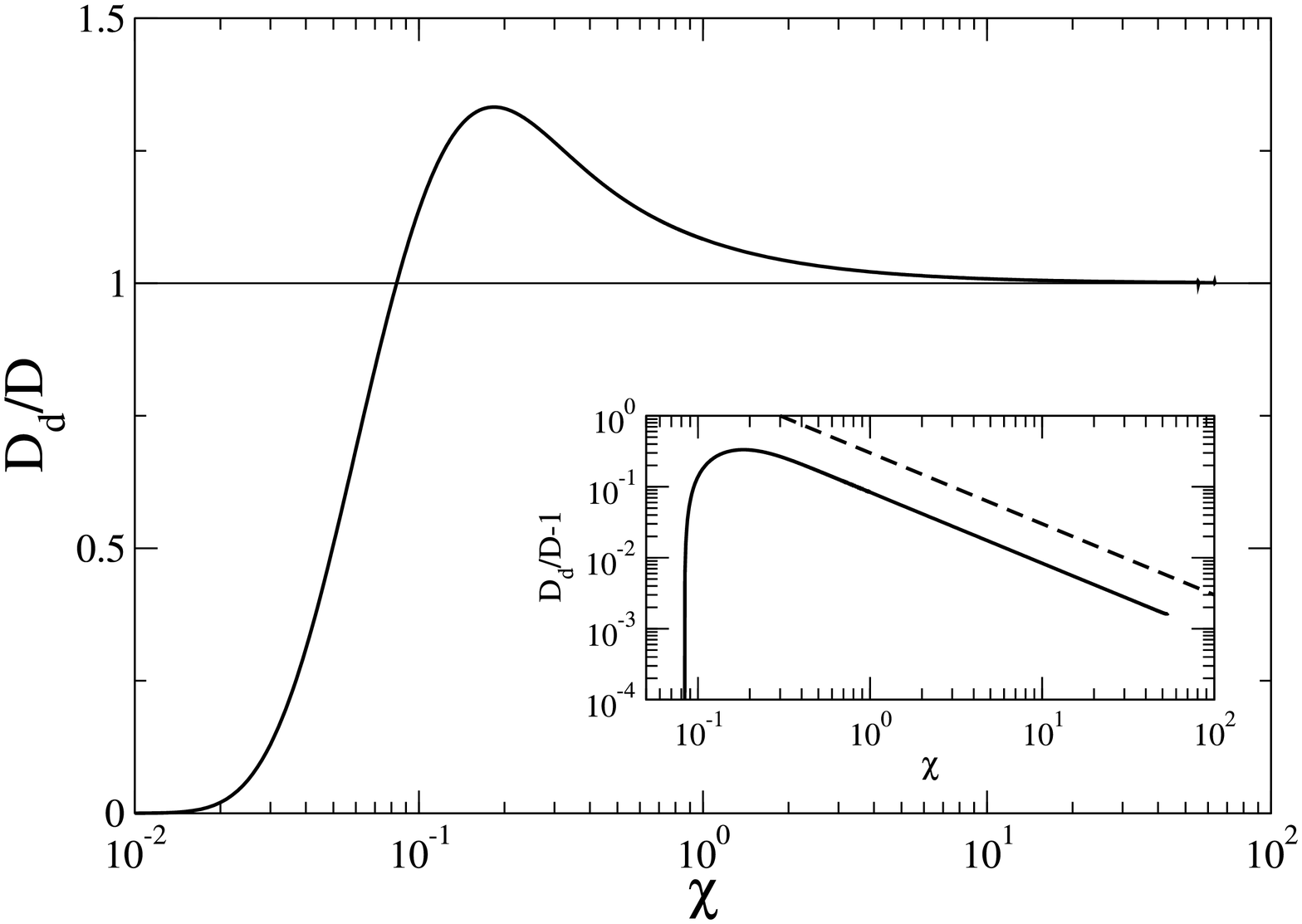}
\includegraphics[scale=0.28]{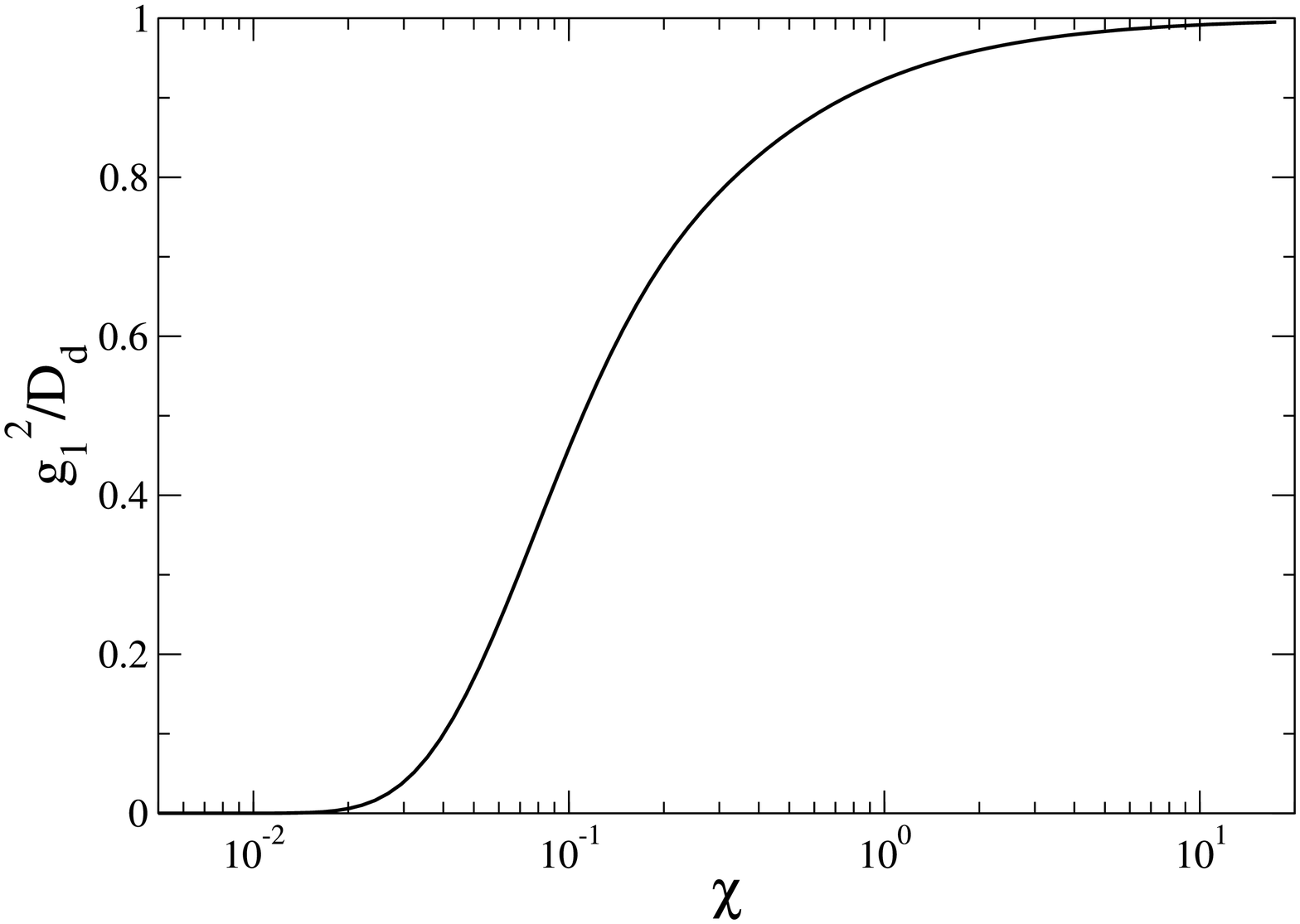}
\caption{Left panel. Ratio between the variance $D_d$ of the discretized process and the variance $D$ of the original Gaussian process as a function of the scaling parameter $\chi$. The inset shows the fractional change $D_d/D-1$ as a function of $\chi$ in a log-log scale. The dashed line shows the $\chi^{-1}$ behavior. Right panel. Ratio between the squared coefficient $g_1^2$ (see formula (\ref{gCoeff})) and the variance of the discretized process as a function of the scaling parameter $\chi$. This ratio is equal to the ratio between the autocorrelation function $\rho_d(k)$ of the discretized process and the autocorrelation function $\rho(k)$ for large values of lags $k$ (see Corollary~\ref{CorACF}).}
\label{varFigureIncr}
\end{center}
\end{figure}


\subsection{Autocovariance and autocorrelation function}\label{SubSec:ACV}

In order to study how the correlation properties change when a stationary Gaussian process is discretized we will make use of a general theory presented, for example, in Beran (1994) and Dittmann and Granger (2002). For the benefit of the reader we recapitulate here this approach. 

Without loss of generality, we consider the case of an underlying stationary Gaussian process of unit variance. 
The starting point is a series expansion of the bivariate Gaussian density function in Hermite polynomials. Hermite polynomials are an orthonormal polynomial system with Gaussian weight. Specifically, we use the same normalization of Hermite polynomials as in Dittmann and Granger (2002), i.e.,
\begin{equation}
\int_{-\infty}^{\infty} H_n(x) H_m(x) \frac{e^{-x^2/2}}{\sqrt{2\pi}} dx=\delta_{nm}
\end{equation}
where $\delta_{nm}$ is the Kronecker delta.
The expansion of the bivariate Gaussian density function $P(x,y)$ in Hermite polynomials is (see Barrett and Lampard (1955))
\begin{equation}
P(x,y)=P(x)P(y)\left[1+\sum_{j=1}^{\infty} \rho^j H_j(x) H_j(y)\right]
\end{equation} 
where $P(x)$ is the univariate Gaussian density function and $\rho$ is the correlation coefficient between variables $x$ and $y$.

Following Lemma 1 of Dittmann and Granger (2002), if we transform two Gaussian random variables $X$ and $Y$ with a nonlinear transformation $g(\cdot)$ that can be decomposed in Hermite polynomials
\begin{equation}
g(x)=g_0+\sum_{j=1}^\infty g_j H_j(x)
\label{Hermite}
\end{equation}
the linear covariance and the linear correlation of the transformed variables are
\begin{equation}
Cov[g(X)g(Y)]=\sum_{j=1}^\infty g_j^2 \rho^j  ~~~~~~~~~~~~~Cor[g(X)g(Y)]=\frac{\sum_{j=1}^\infty g_j^2 \rho^j}{\sum_{j=1}^\infty g_j^2}
\label{hexpansion}
\end{equation}
The proof is straightforward. The coefficients $g_j$ in (\ref{Hermite}) are
\begin{equation}\label{gCoeff}
g_j=\int_{-\infty}^{\infty} g(x) H_j(x) \frac{e^{-x^2/2}}{\sqrt{2\pi}} dx
\end{equation}
The smallest $j>0$ such that $g_j$ is non-vanishing is called the Hermite rank of the function $g(x)$. Note that the second equation in  (\ref{hexpansion}) Êimplies that any non-linear transformation of a bivariate Gaussian distribution decreases the correlation between the variables (the covariance can of course increase or decrease). When $X$ and $Y$ describes the same process at two different times the above equations can be used to compute the autocovariance and autocorrelation properties of the transformed process.

Dittmann and Granger (2002) used the above expansion to study  how the long memory properties change as a result of nonlinear transformations that can be written as finite Hermite expansions. (See Proposition 1 of Dittmann and Granger (2002).) They mention that this approach  cannot be used if the transformation has an infinite Hermite expansion.
As we will show below, the discretization can be expressed as an infinite sum of Hermite polynomials and therefore we cannot use directly their Proposition 1. In the following we compute explicitly the asymptotic behavior of the autocovariance function and of the spectral density of the discretized process in order to directly obtain its long memory properties.


By using the theory outlined above we compute the asymptotic behavior of the autocovariance $\gamma_d(k)$ and the autocorrelation $\rho_d(k)$ of the discretized process.


%
\begin{proposition}\label{Prop:ACVGen}
Let $\{X(t)\}_{t \in \mathbb{N}}$  be a stationary Gaussian process with autocovariance function given by Definition \ref{generalDef}. Then the autocovariance function of the discretized process $\{X_d(t)\}_{t \in \mathbb{N}}$ satisfies
\beq
\gamma_d(k)= \left(\frac{ \vartheta_2(0,e^{-1/2\chi}) }{\sqrt{2 \pi \chi}}\right)^2 k^{2H-2} L(k) \left(1 + O\left(k^{4H-4} L^2(k)\right)\right)\qquad as \;\; k\rightarrow \infty\,,\label{eq:ACVGen}
\eeq
where $\vartheta_a(u,q)$ is the elliptic theta function.
\end{proposition}%
The proof of this and of the other propositions in this section can be found in Appendix \ref{App:Proof1}.  According to Definition \ref{generalDef}, this proposition proves that the discretized process is a long memory process with the same Hurst exponent as the original process. Note that asymptotically the autocovariance of the discretized process is proportional to the autocovariance of the original process.

Moreover, as for the original process, the autocovariance function does not necessarily decay as a pure power-law. The second order corrections are nested inside the slowly varying function, and it may happen that the first order term in  $O(k^{4H-4}L^{2}(k))$ dominates the second order term in $k^{2H-2}L(k)$. 
To compute explicitly the second-order corrections we need to specify a functional form for the slowly varying function $L(k)$. To this end we consider long memory processes whose slowly varying function in the autocovariance is well behaved according to Definition \ref{wellbehavedDef}.

\begin{proposition}\label{Prop:ACV}
Let $\{X(t)\}_{t \in \mathbb{N}}$  be a stationary Gaussian process with autocovariance function given by Definition \ref{generalDef} and  $L \in \mathcal{L}$. Then, the autocovariance function of the discretized process $\{X_d(t)\}_{t \in \mathbb{N}}$ satisfies
\beq
\gamma_d(k)= \left(\frac{ \vartheta_2(0,e^{-1/2\chi}) }{\sqrt{2 \pi \chi}}\right)^2 b_0 \, k^{2H-2}  \left(1 + O\left(k^{-\min(4-4H, \beta_1)} \right)\right)\qquad as \;\; k\rightarrow \infty,
\eeq
where $\vartheta_a(u,q)$ is the elliptic theta function.
\end{proposition}

The behavior of the autocorrelation function is trivially obtained from the two above propositions. For example, for the more general case we have the following.


\begin{corollary}\label{CorACF}
Under the conditions of Proposition \ref{Prop:ACVGen}, the autocorrelation function $\rho_d(k)$ of the discretized process $\{X_d(t)\}_{t \in \mathbb{N}}$ satisfies, as $k\rightarrow \infty$, 
\beq
\rho_d(k) =\left(\frac{ \vartheta_2(0,e^{-1/2\chi}) }{\sqrt{2 \pi \chi}}\right)^2 \frac{k^{2H-2} L(k)}{D_d} \left(1 + O\left(k^{4H-4} L(k)^2\right)\right) =
\frac{g_1^2}{D_d}\rho(k)\left(1 + O\left(k^{4H-4} L(k)^2\right)\right)
\eeq
\end{corollary}
%

For large $k$, $\rho_d(k)/\rho(k) \sim g_1^2/D_d$. Hereafter, the notation $x_k \sim y_k$ means that $x_k/y_k \to 1$ as $k\to \infty$, unless specified otherwise. The right panel of Figure \ref{varFigureIncr} shows $g_1^2/D_d$ as a function of $\chi$. We observe that the more severe is the discretization, the larger is the reduction of the autocorrelation function. 

We numerically tested our propositions and the error made by considering only the leading term in the asymptotic expansion. We simulated a fGn with unit variance for which $\lim_{k\to\infty} L(k) =H(2H-1)$. Figure \ref{acv}  shows the sample autocovariance of the discretized process with $\chi=0.1$ for $H=0.7$ and $H=0.85$, and for different sample size, namely $n=2^{10}=1,024$ and $n=2^{14}=16,384$. Before discussing this figure, we remind that the sample autocovariance (and the autocorrelation) is a biased estimator for long memory time series. Hosking (1996) showed that, for a generic long memory process with an autocovariance asymptotically decaying as $\gamma(k)\sim\lambda k^{2H-2}$ with $0.5 < H <1$ and $\lambda > 0$,  the sample autocovariance $\hat\gamma(k)$ has an asymptotic bias
\begin{equation}\label{biasEq}
E[\hat\gamma(k)]-\gamma(k)\sim\frac{-\lambda n^{2H-2}}{H(2H-1)} \qquad as \;\; n\to\infty\,,
\end{equation}
where $n$ is the length of the sample time series. Since Hosking's theorem only requires that the process is long memory, we can apply it also to the discretized time series. Figure \ref{acv} shows a very good agreement between simulations and analytical results with bias correction. 



\begin{figure}[t]
\begin{center}
\includegraphics[scale=0.28]{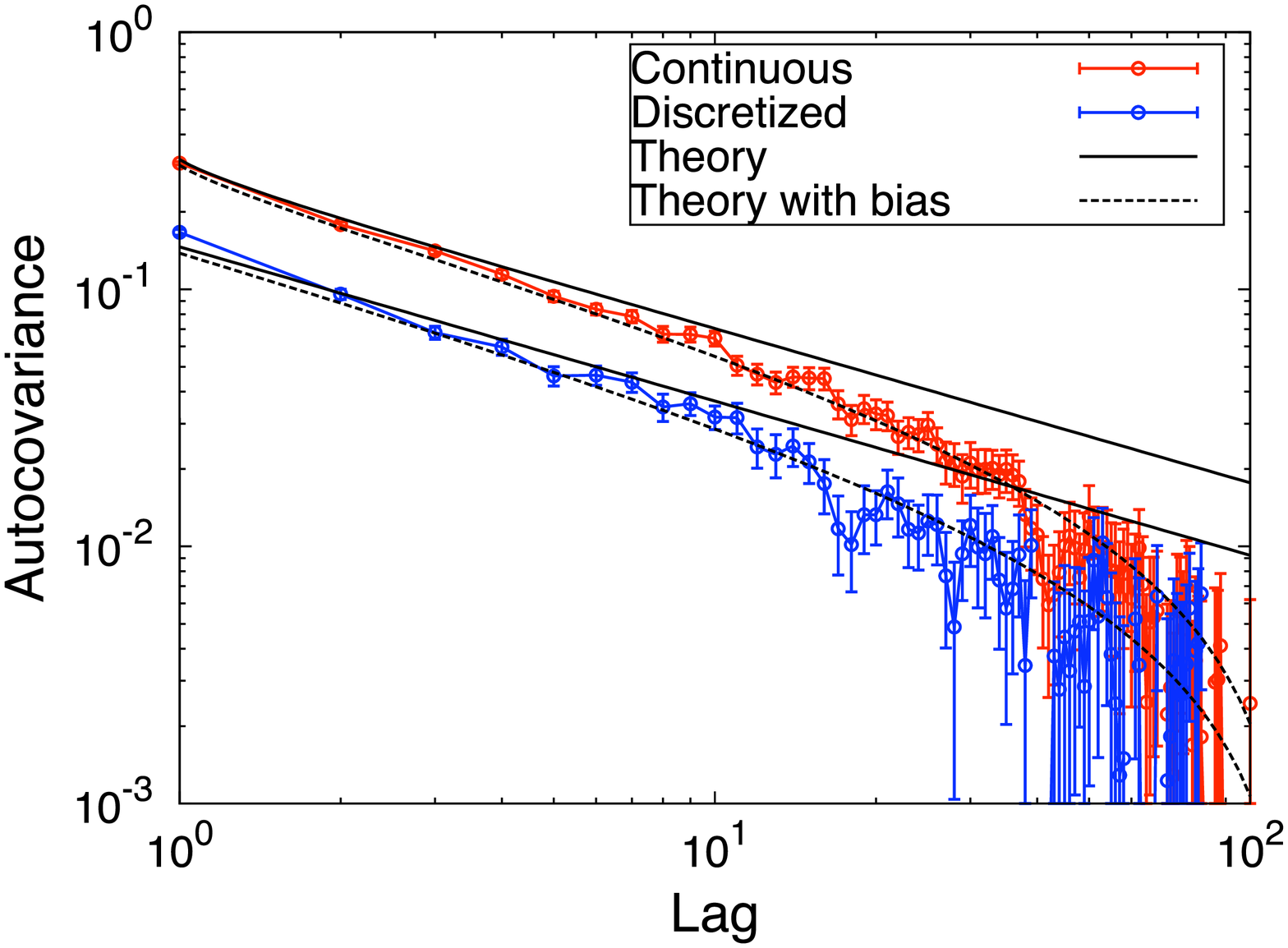}
\includegraphics[scale=0.28]{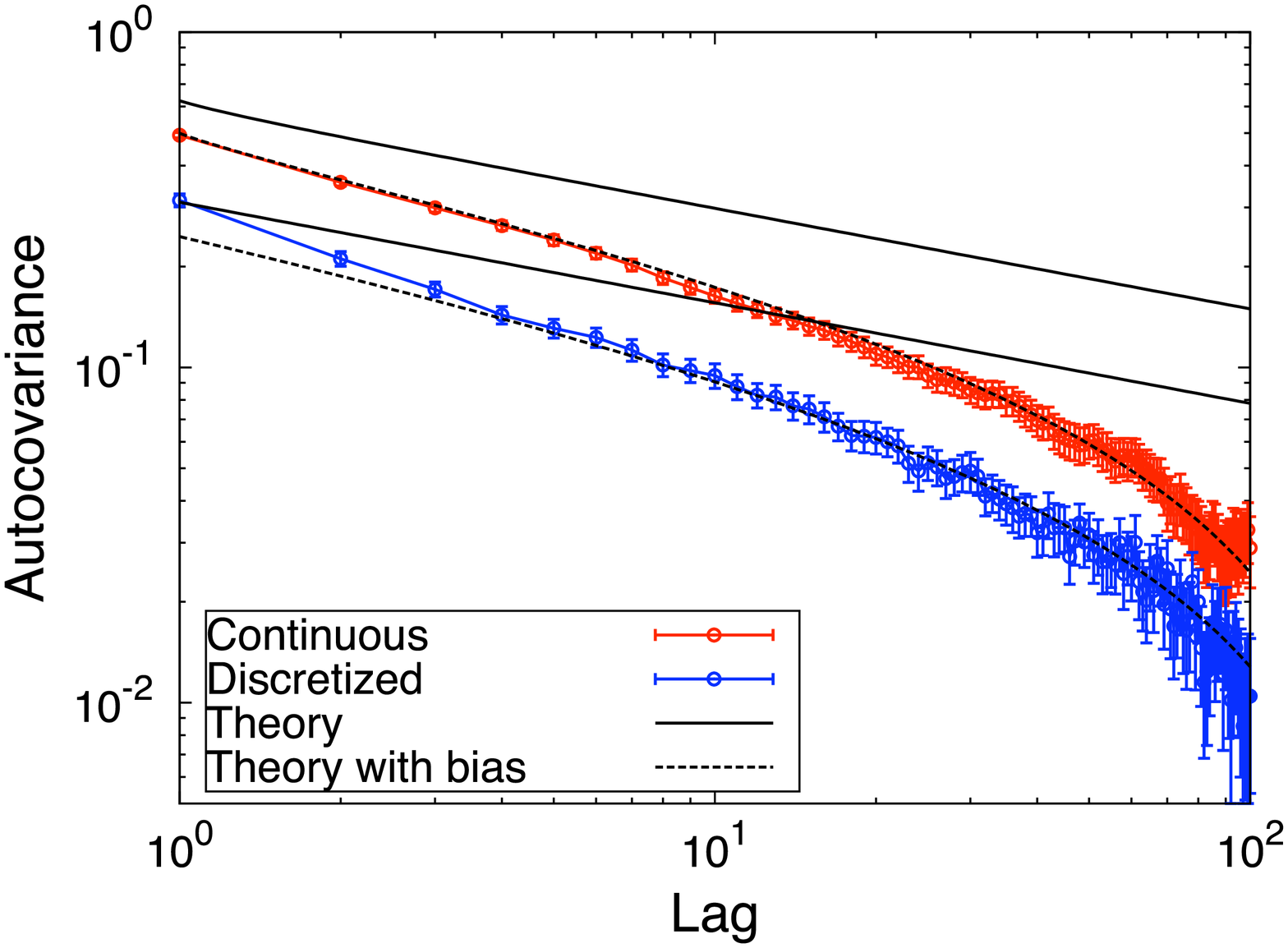}
\includegraphics[scale=0.28]{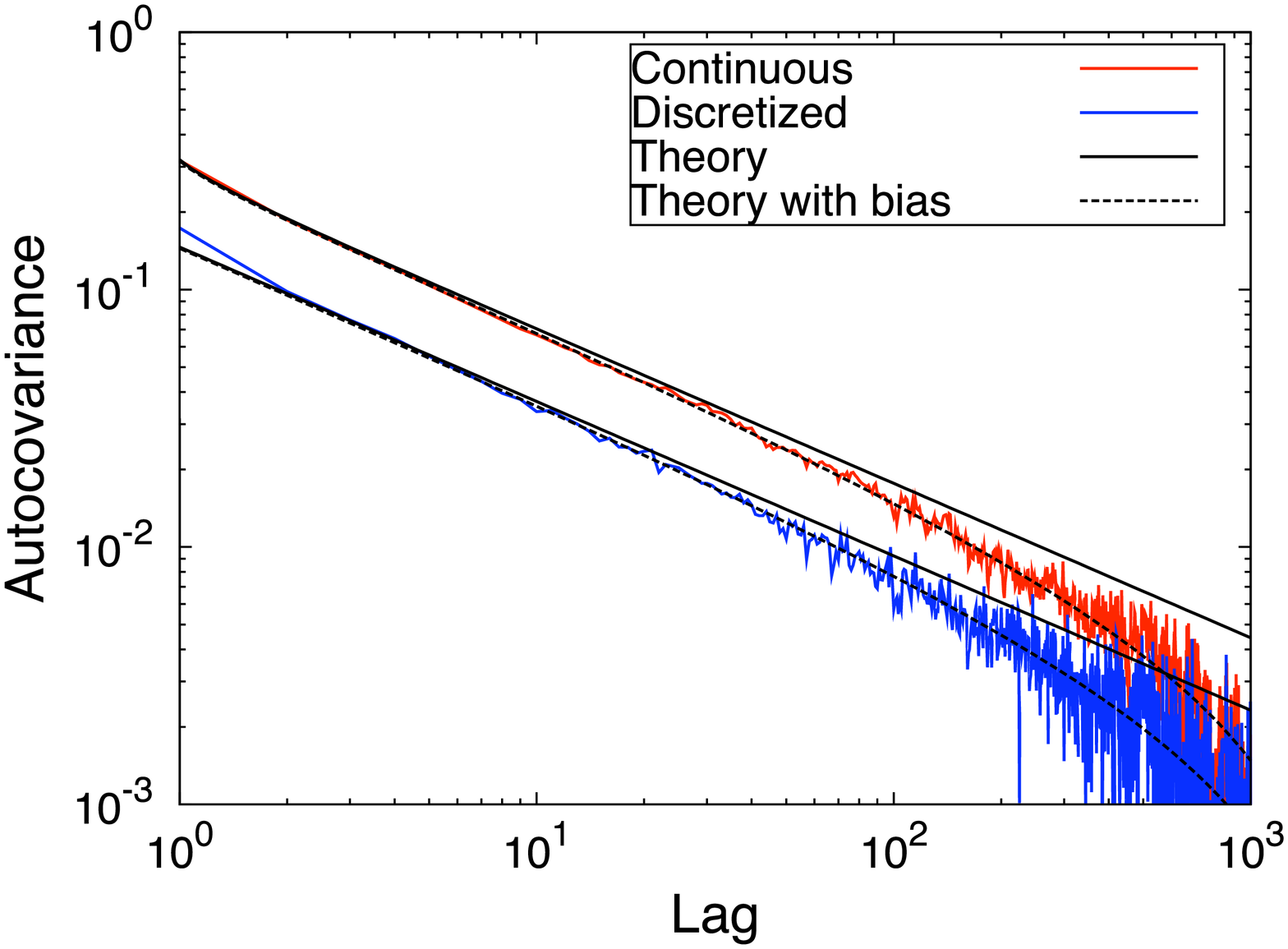}
\includegraphics[scale=0.28]{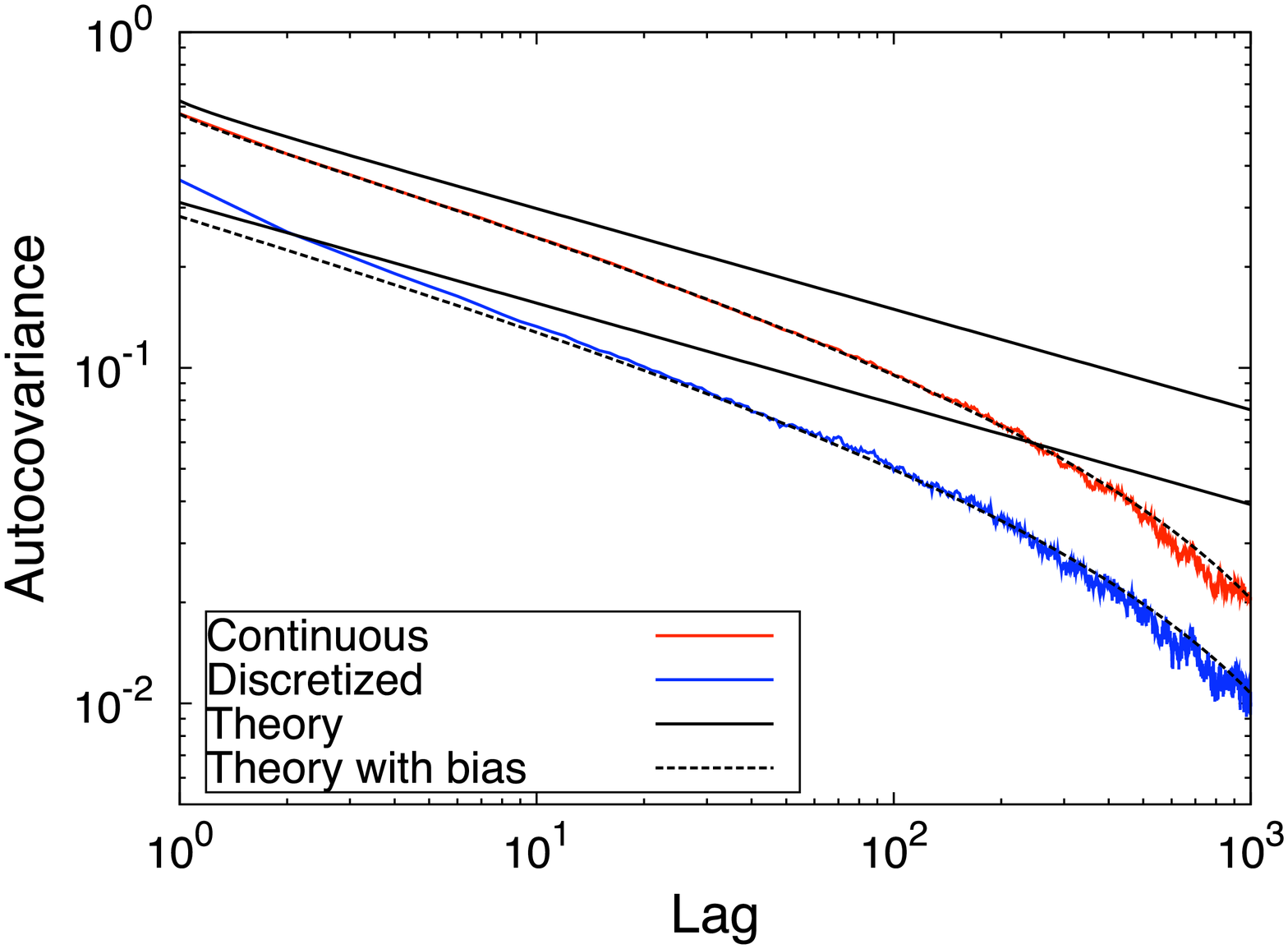}
\caption{Sample autocovariance of a numerical simulation of a fGn and its discretization with a scaling parameter $\chi=0.1$. The time series has length $n=2^{10}$ (top) and $n=2^{14}$ (bottom). The Hurst exponent of the fGn is $H=0.7$ (left) and $H=0.85$ (right). The figure also shows the asymptotic theoretical autocovariance and the autocovariance corrected for the finite sample bias of  (\ref{biasEq}) both for the fGn and for its discretization.}
\label{acv}
\end{center}
\end{figure}

\subsection{Spectral density}




\begin{proposition}\label{Prop:GenSpecDens}
Let $\{X(t)\}_{t \in \mathbb{N}}$  be a stationary Gaussian process with autocovariance function given by Definition \ref{generalDef}. If $L$ belongs to the Zygmund class\footnote{A positive measurable function $f$ belongs to the Zygmund class if, for every $\alpha>0$, $x^{\alpha} f(x)$ is ultimately incrasing, and $x^{-\alpha}f(x)$ is ultimately decreasing (see Zygmund (1959)).}, 
 the spectral density $\phi_d(\omega)$ of the discretized process $\{X_d(t)\}_{t \in \mathbb{N}}$ satisfies
\beq
\lim_{\omega\rightarrow 0^+} \phi_d(\omega)/\left[ \left(\frac{ \vartheta_2(0,e^{-1/2\chi}) }{\sqrt{2 \pi \chi}}\right)^2 c_{\phi} |\omega|^{1-2H} L(\omega^{-1})\right]=1
\eeq
where $\vartheta_a(u, q)$ is the elliptic theta function, and $c_{\phi}=\pi^{-1} \Gamma(2H-1) \sin(\pi H)$.
\end{proposition}

Similarly to the case of the autocovariance, second-order corrections to the spectral density are hidden inside the slowly varying function $L(k)$ and the terms we neglected. To compute these corrections explicitly in terms of powers of $|\omega|$ we need to specify a functional form for the slowly varying function. We therefore consider long memory processes whose slowly varying function in the autocovariance is well behaved according to Definition \ref{wellbehavedDef}. Moreover, we introduce the following assumption
%
%
%
\begin{assumption}\label{Ass:AbsSum}
In Definition \ref{wellbehavedDef}, either $I<\infty$, or, if $I=\infty$, then 
\bi
\item[(i)] $K=1$;
\item[(ii)] $\sup\{\beta_i\}>2H-1$;
\item[(iii)] $\sup_i \left\{\Gamma(2H-1-\beta_i)\right\}<\infty$.
\ei
\end{assumption}
Assumption~\ref{Ass:AbsSum} (i) provides that $L(k)=\sum_{i=0}^{\infty} b_i k^{-\beta_i}$ converges absolutely $\forall \, k\geq1$. Under this assumption the Fourier transform of $k^{2H-2}L(k)$ can be written as the series of Fourier transforms of the terms $k^{2H-2-\beta_i}$. Assumption~\ref{Ass:AbsSum} (ii) implies that there is only a finite number of terms in $L(k)$ that are not summable over $k$. Assumption~\ref{Ass:AbsSum} (iii) ensures that we can rearrange an infinite number of terms in a series expansion of polylogarithms that we use below to represent  Fourier series (see the proof of Proposition~\ref{Prop:SpecDens} and Lemma~\ref{Lemma:SpecDensGaussian}).

Then, we can prove the following
\begin{proposition}\label{Prop:SpecDens}
Let $\{X(t)\}_{t \in \mathbb{N}}$  be a stationary Gaussian process with autocovariance function given by Definition \ref{generalDef} and  $L \in \mathcal{L}$. 
Then, the spectral density $\phi_d(\omega)$ of the discretized process $\{X_d(t)\}_{t \in \mathbb{N}}$ satisfies, as $\omega \to 0^+$,
\beqnr\label{eq:SpecDensRound1stOrd}
\phi_d(\omega) =  \left(\frac{\vartheta_2(0,e^{-1/2\chi}) }{\sqrt{2 \pi \chi}}\right)^2 c_{\phi} b_0 \, |\omega|^{1-2H} + g(\omega) + o\left(g(\omega)\right)
\eeqnr
where $\vartheta_a(u, q)$ is the elliptic theta function, $c_{\phi}=\pi^{-1}\Gamma(2H-1)\sin(\pi H)$, and\footnote{Note that the third equation in (\ref{eq:SpecDensRound2ndOrd_1}) may hold only if $\beta_1<\frac{2}{3}$, so that $\frac{1+\beta_1}{2}<1-\frac{\beta_1}{4}$ and the set for $H$ is non-empty.}
\beqnr\label{eq:SpecDensRound2ndOrd_1}
g(\omega)=
\begin{cases}
c_0 \in \mathbb{R} & if \;\; H<\min\left(\frac{5}{6}, \frac{1 + \beta_1}{2}\right),\\
 g_3^2\left(\frac{b_0}{D}\right)^3 \pi^{-1} \Gamma(6H-5)\sin(3 H \pi)|\omega|^{5-6H} & if  \;\; H > \max \left(\frac{5}{6}, 1-\frac{\beta_1}{4}\right), \\
 \left(\frac{\vartheta_2(0,e^{-1/2\chi}) }{\sqrt{2 \pi \chi}}\right)^2 b_1 \pi^{-1} \Gamma(2H-1-\beta_1)\sin\left(\frac{2H-\beta_1}{2}\pi\right)|\omega|^{1-2H+\beta_1} & if \;\; H\in\left(\frac{1+\beta_1}{2}, 1-\frac{\beta_1}{4}\right), \\ 
 \left(\frac{\vartheta_2(0,e^{-1/2\chi}) }{\sqrt{2 \pi \chi}}\right)^2 b_1 \pi^{-1} \ln |\omega|^{-1} & if \;\; H=\frac{1+\beta_1}{2}<\frac{5}{6}, \\
 g_3^2\left(\frac{b_0}{D}\right)^3 \pi^{-1}\ln |\omega|^{-1} & if \;\; H=\frac{5}{6}<\frac{1+\beta_1}{2}, \\
\left (\left(\frac{\vartheta_2(0,e^{-1/2\chi}) }{\sqrt{2 \pi \chi}}\right)^2  b_1+ g_3^2\left(\frac{b_0}{D}\right)^3\right) \pi^{-1} \ln |\omega|^{-1} & if \;\; H=\frac{1+\beta_1}{2}=\frac{5}{6},
\end{cases}
\eeqnr
where $g_3$ is given by (\ref{eq:g3}). 

In addition, under Assumption~\ref{Ass:AbsSum}, $\forall\,j\geq0$ let $\{\widetilde{b}_{j, i}, \widetilde{\beta}_{j, i}\}_{0\leq i\leq I}$ be the real numbers given by the Cauchy product
\beq
\left(\sum_{i=0}^I b_i k^{-\beta_i}\right)^{2j+1} = \sum_{i=0}^{\widetilde{I}_j} \widetilde{b}_{j, i} k^{-\widetilde{\beta}_{j, i}} \qquad \forall \; k\geq 1\nonumber
\eeq
where $\widetilde{I}_j=\binom{2j+I}{I-1}$ if $I<\infty$, and $\widetilde{I}_j=\infty$ $\forall\, j\geq0$ if $I=\infty$. Then, if $H<\min\left(\frac{5}{6}, \frac{1 + \beta_1}{2}\right)$,
\beq\label{eq:SpecDensRound2ndOrd_2}
g(\omega) = \frac{D_d}{2\pi} + \pi^{-1} \sum_{j=0}^{\infty} \sum_{i=0}^{\widetilde{I}_j} \frac{g_{2j+1}^2}{D^{2j+1}}\, \widetilde{b}_{j, i}\, \zeta\left((2j+1) (2-2H)+\widetilde{\beta}_{j, i}\right)\,,
\eeq
where $D_d$ is the variance of the discretized process, $g_{2j+1}$ are the Hermite coefficients of the discretization, and $\zeta(\cdot)$ is the Riemann zeta function. 
\end{proposition}

Note that, if we just assume $L \in \mathcal{L}$, we are able to compute exactly only the second-order terms $O\left(|\omega|^{5-6H}\right)$, $O\left(|\omega|^{1-2H+\beta_1}\right)$ and $O\left(\ln |\omega|^{-1}\right)$. However, under Assumption~\ref{Ass:AbsSum} we can calculate exactly also the term $O\left(1\right)$ in (\ref{eq:SpecDensRound2ndOrd_1}), and therefore we have the exact second-order correction for all values of $H$.

As we have mentioned above, an important case is given by long memory processes whose slowly varying function is analytic at infinity. This class includes the fGn and the fARIMA process. If the slowly varying function is analytic at infinity, then $\beta_1 \in \mathbb{N}$ and therefore we have the following

\begin{corollary}\label{Cor:SpecDensAnalytic}
If $L\in \mathcal{R}_0$ and $L$ is analytic at infinity, then, as $\omega \to 0^+$,
\beq\label{eq:SpecDensAnalytic}
\phi_d(\omega) =  \left(\frac{\vartheta_2(0,e^{-1/2\chi}) }{\sqrt{2 \pi \chi}}\right)^2 c_{\phi} b_0 \, |\omega|^{1-2H} + 
\begin{cases}
c_0 + o\left(1\right)& \; for \;\; H\in (1/2,5/6)\\
c_1 \ln |\omega|^{-1} + o \left(\ln |\omega|^{-1}\right) & \; for \;\; H = 5/6\\
c_2 |\omega|^{5-6H} + o\left(|\omega|^{5-6H}\right) & \; for \;\; H\in (5/6, 1)
\end{cases}
\eeq
where $c_1>0$ and $c_2>0$ are given by the fifth and the second equation of (\ref{eq:SpecDensRound2ndOrd_1}), respectively. 

In addition, under Assumption~\ref{Ass:AbsSum}, $c_0$ is given by (\ref{eq:SpecDensRound2ndOrd_2}).
\end{corollary}

The sign of (\ref{eq:SpecDensRound2ndOrd_2}) depends on the specific autocovariance function (or, equivalently, on the specific spectral density) of the underlying Gaussian process. For example, the fGn satisfies Assumption~\ref{Ass:AbsSum} and in this case we can prove the following 

\begin{corollary}\label{Cor:SpecDensRoundfGn}
The spectral density of the discretization of a fGn satisfies (\ref{eq:SpecDensAnalytic}) and the second-order term is strictly positive for all $H \in (0.5, 1)$.
\end{corollary}

Unfortunately, the  fARIMA process does not satisfy Assumption~\ref{Ass:AbsSum} (i) and we are not able to prove that $c_0$ satisfies (\ref{eq:SpecDensRound2ndOrd_2}) and is strictly positive. However, from the first part of Corollary~\ref{Cor:SpecDensAnalytic} we know that the second-order term is strictly positive for $H \geq 5/6$ and our extensive numerical simulations (see below) indicate that it is strictly positive also for $H<5/6$.

\subsubsection{The spectral density of the discretization of fGn and fARIMA processes}

The fGn and the fARIMA process are often used in modeling Gaussian long memory processes. Starting from the explicit functional form of the spectral density of either of these two processes it is direct to show (see Beran (1994)) that 
the relative second-order term of the spectral density  is $O(\omega^2)$, which is a very small correction to the leading term even for relatively large frequencies. Hence, roughly speaking, on a log-log plot the spectral density of both the fGn and the fARIMA process is close to a straight line even for relatively large frequencies. 

On the other hand, the two corollaries above state that for the discretization of a Gaussian long memory process with analytic $L(k)$,  as the fGn and the fARIMA, the relative second order correction to the spectral density around $0^+$ is always larger than $O\left(|\omega|^{2/3}\right)$, which is a relatively large correction even for small frequencies. On a log-log plot this leads to a significant deviation from the straight-line behavior (leading term) even for relatively small frequencies. 

In Figures~\ref{perCoarse} we show examples of the sample periodogram of the discretized fGn ($\chi=0.1$) for different time series length and Hurst exponent. As expected, while the periodogram of the fGn is very straight in a log-log scale, the periodogram of the discretized process changes its slope for increasing frequencies. For the discretized process the absolute slope of the periodogram decreases for increasing frequencies for all values of $H$ due to the strictly positive second-order term. We ran the same simulations for various fARIMA processes and we obtained similar results. 

As we will show below, these observations have important consequences for the Hurst exponent estimators based on the periodogram.



\begin{figure}[t]
\begin{center}
\includegraphics[scale=0.28]{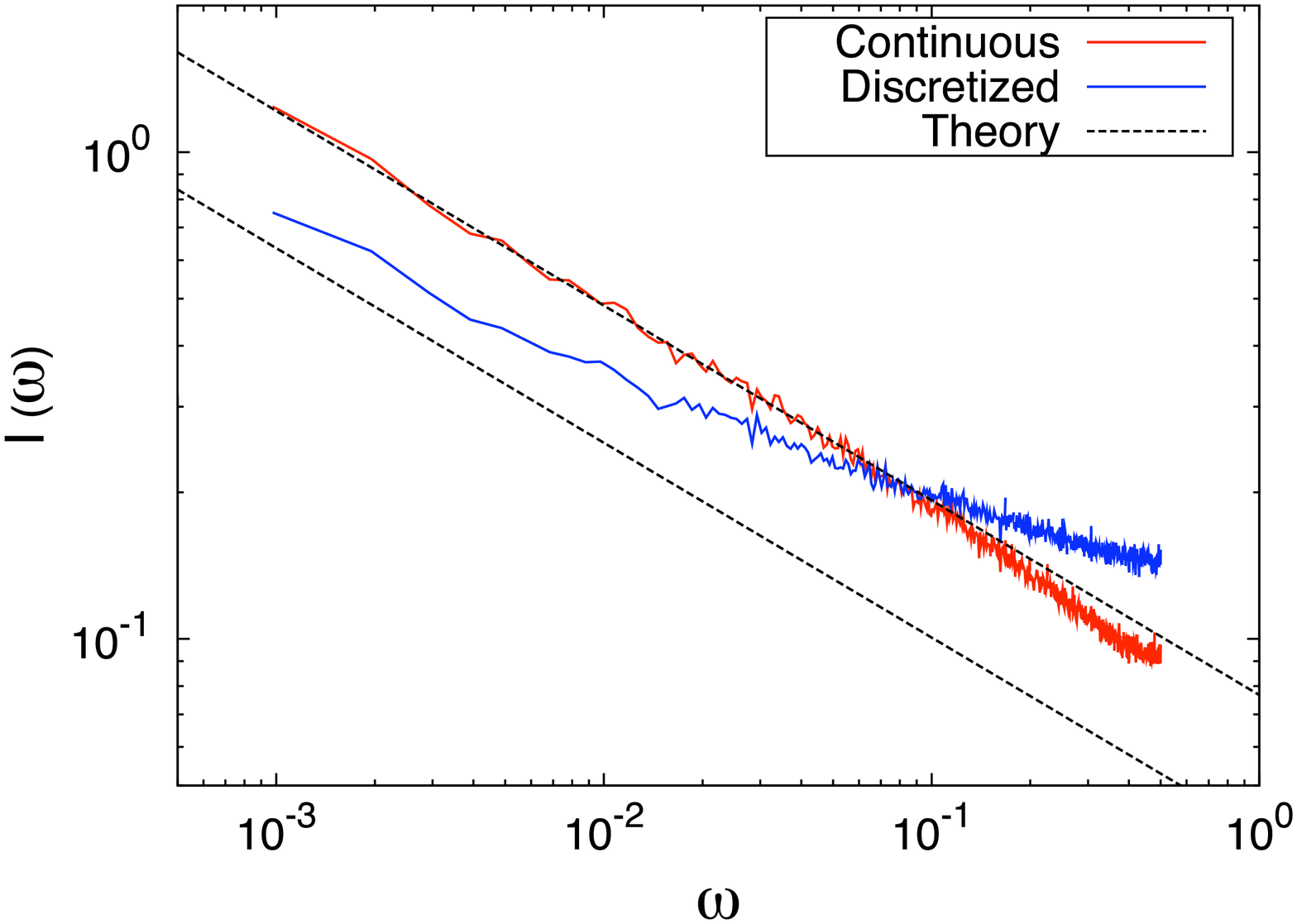}
\includegraphics[scale=0.28]{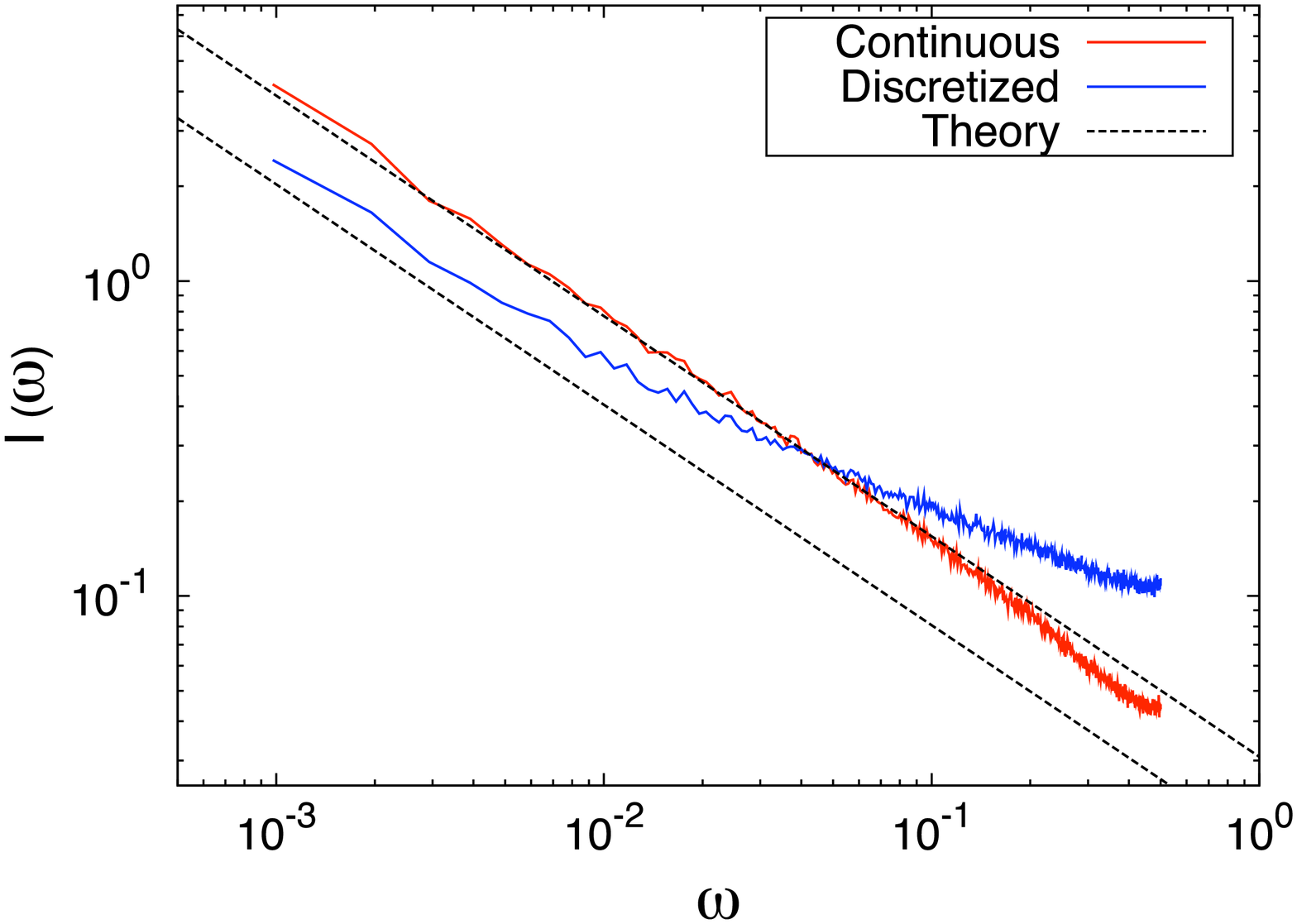}
\includegraphics[scale=0.28]{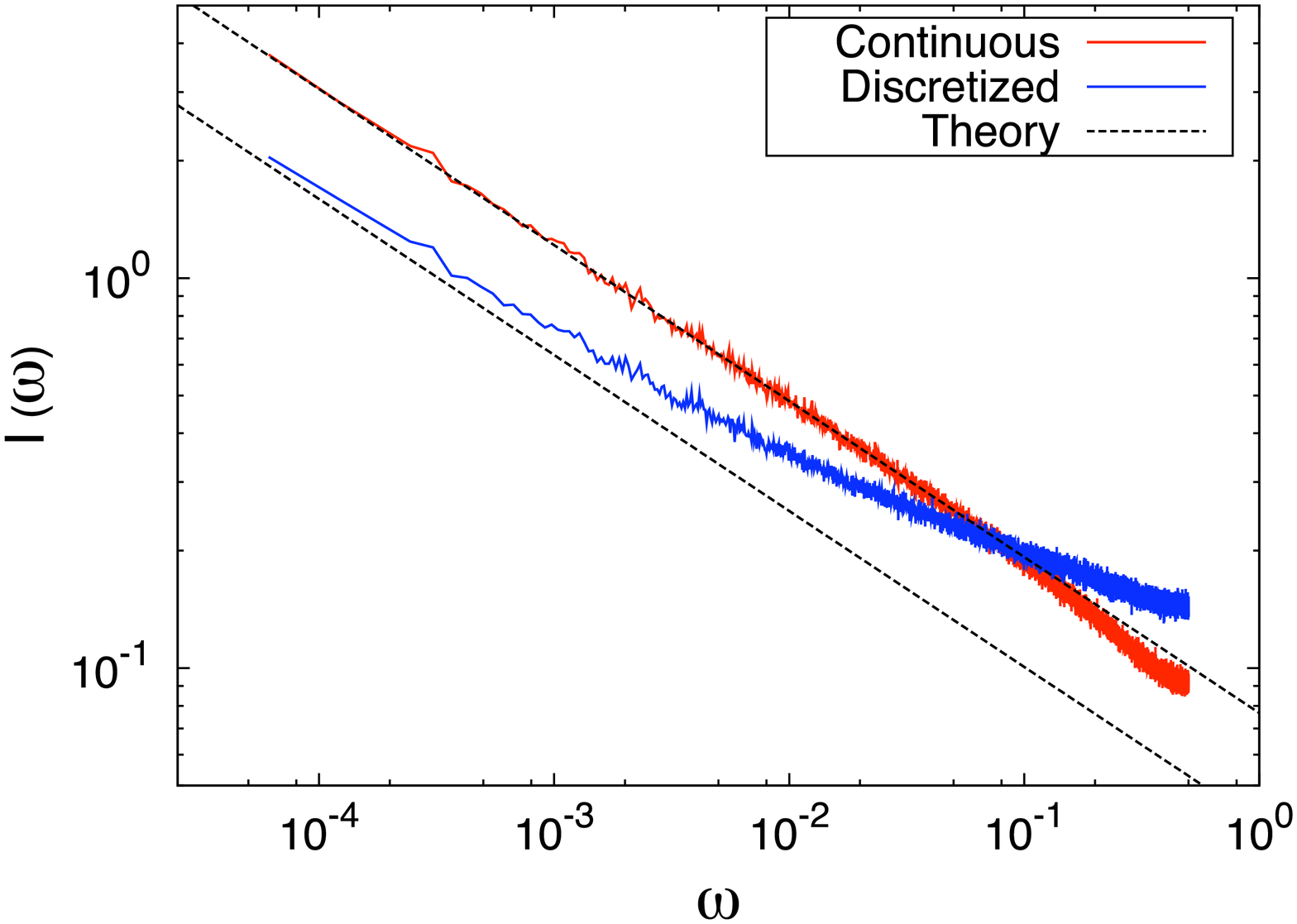}
\includegraphics[scale=0.28]{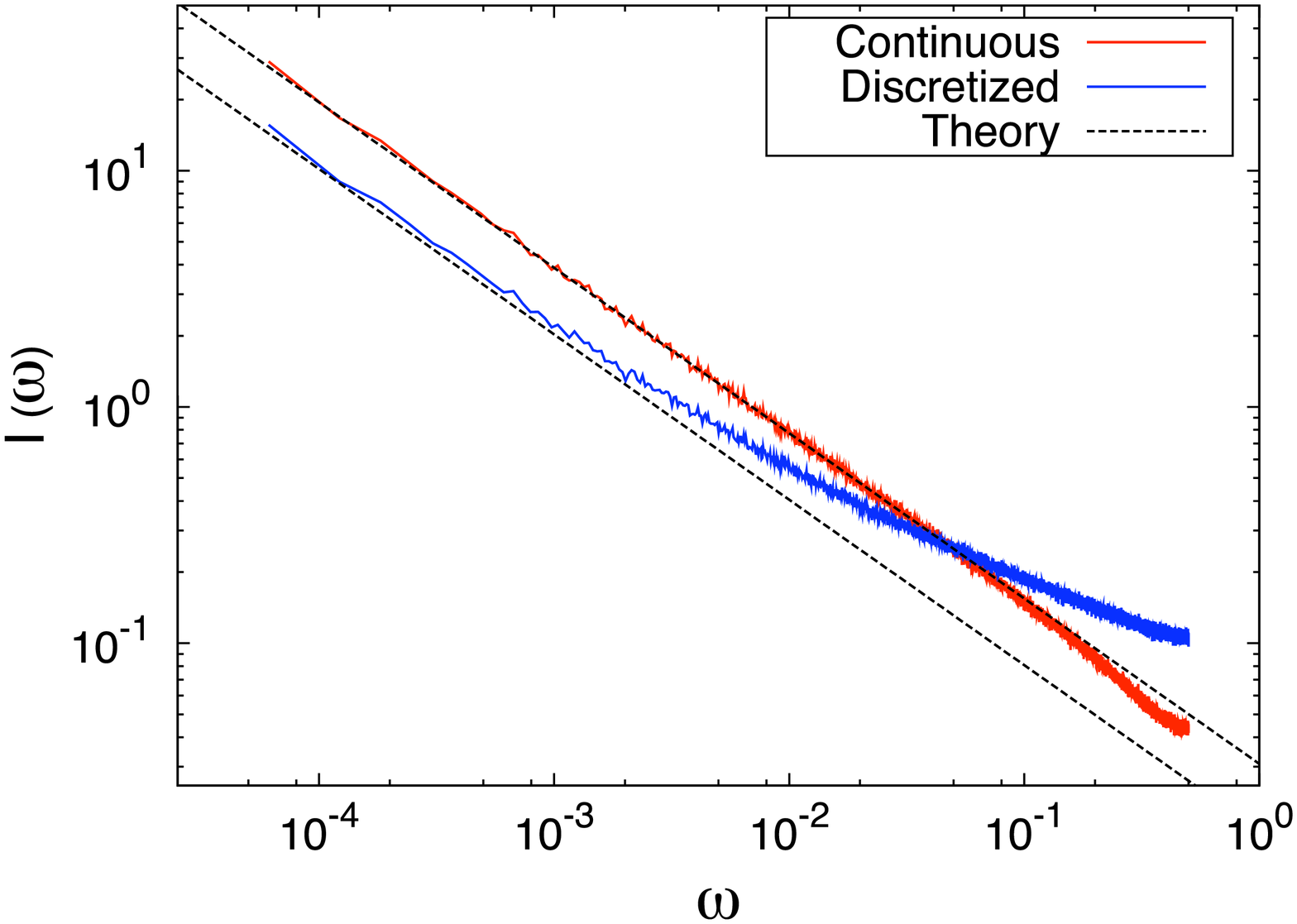}
\caption{Sample periodogram  of a numerical simulation of a fGn and its discretization with a scaling parameter $\chi=0.1$. The time series has length $n=2^{10}$ (top) and $n=2^{14}$ (bottom). The Hurst exponent of the fGn is $H=0.7$ (left) and $H=0.85$ (right). The figure also shows the leading term of the expansion of the spectral density for small $\omega$.}
\label{perCoarse}
\end{center}
\end{figure}





\section{Estimation of the Hurst exponent}
\label{estimation}


In this section we investigate the estimation of the Hurst exponent from finite time series of the discretized process. We numerically generate time series of discretized long memory processes and we compute the Hurst exponent by using  the local Whittle estimator and the Detrended Fluctuation Analysis. 
Despite the fact that, as we have demonstrated in the previous sections, the Hurst exponent of the discretized process is the same as the one of the original process, we will show here that both methods give estimates of the Hurst exponent which are systematically and significantly negatively biased. 
These results are consistent with the literature on long-memory signal plus noise processes (Arteche (2004), and Hurvich, Moulines and Soulier (2005)) and on general non-linear transformations of long-memory processes (Dalla, Giraitis and Hidalgo (2006)).

\subsection{Local Whittle estimator}\label{sec:Whittle}
The local Whittle (LW) estimator  is a Gaussian semiparametric estimator that works in the frequency domain. It has been suggested by K\"unsch (1987). Since its introduction it has been extensively studied in econometric theory (see Robinson (1995b), Velasco (1999), Phillips and Shimotsu (2004), Andrews and Sun (2004), Shimotsu and Phillips (2005), Shimotsu and Phillips (2006), Dalla, Giraitis and Hidalgo (2006), Shao and Wu (2007), Abadir, Distaso and Giraitis (2007)), and widely used in theoretical and applied works in financial econometrics to estimate long memory in volatility (see Hurvich and Ray (2003), Arteche (2004) and Hurvich, Moulines and Soulier (2005)).

Let $\{X_t\}_{t \in \mathbb{N}}$ be a weakly stationary long memory process with spectral density $\phi(\omega)$ satisfying
\beq
\phi(\omega) = c\,|\omega|^{1-2 H} \left(1 + o(1)\right) \qquad as \;\; \omega \to 0^+\,\label{eq:SpecWhittleGeneral}
\eeq
for some $c>0$.
For a time series $\{X_t\}$, $t=1, \ldots, n$, define the periodogram
\beq
I_n(\omega_j) = \left(2 \pi n\right)^{-1} \left|\sum_{t=1}^n X_t \exp^{i t \omega_j}\right|\nonumber
\eeq
Let $\omega_j = 2 \pi j/n$, $j=1, \ldots, n$, be the Fourier frequencies. The LW estimator is defined as the minimizer of the objective function $U_n(h; m)$, i.e.
\beq
\hat{H} \equiv \hat{H}_n = \argmin_{h \in [0, 1]} U_n(h; m)\label{LW:Def}
\eeq
where the local objective function is
\beq
U_n(h; m) = \log \left(\frac{1}{m} \sum_{j=1}^m \omega_j^{2h-1}I_{n}(\omega_j)\right) - \frac{2 h}{m} \sum_{j=1}^m \log \omega_j\label{LW:ObjFun}
\eeq
$m=m(n)$ is an integer-bandwidth parameter such that
\beq
m\to\infty, \;\; m=o(n), \;\; \text{as} \;n\to \infty \nonumber
\eeq
Note that, in semiparametric models, the spectral density function has property (\ref{eq:SpecWhittleGeneral}) and is only locally parameterized around $\omega=0$ by the parameters $H$ and $c$. Therefore, contrary to the parametric Whittle estimation, which employs the full spectrum of frequencies, the local Whittle estimator uses only the first $m$ Fourier frequencies.

Under Gaussianity assumption Fox and Taqqu (1986) proved consistency and asymptotic normality of the LW estimator, while Dahlhaus (1989) proved consistency and asymptotic normality of an exact likelihood approach.
Giraitis and Surgalis (1990) proved consistency and asymptotic normality of the LW estimator under the assumption of linearity for $H\in(0.5, 1)$.
Robinson (1995b) proved consistency and asymptotic normality of the LW estimator under the assumption that the innovations in the Wold representation are a martingale difference sequence with finite fourth moments, for $H \in (0, 1)$.
Under a similar set of assumptions Velasco (1999) extended RobinsonÕs results to the non-stationary region and showed that the estimator is consistent for $H\in (0, 1.5)$, and asymptotically normally distributed for $H\in (0, 1.25)$. Velasco also showed that, upon adequate tapering of the observations, the region of consistent estimation of $H$ may be extended
but with corresponding increases in the variance of the limit normal distribution. Hurvich and Chen (2000) also proved consistency and asymptotic normality of a tapered LW estimator for $H \in (0, 2)$.
Phillips and Shimotsu (2004) considered a special model of non-stationary fractional integration and extended the results of Velasco (1999) proving that the LW estimator: (i) has a non-normal limit distribution for $H \in [1.25, 1.5)$, (ii) has a mixed normal limit distribution for $H = 1.5$, and (iii) converges to unity in probability for $H >1.5$.  Shimotsu and Phillips (2005, 2006) proposed an exact\footnote{The word ``exact'' is used to distinguish the proposed estimator (which relies on an exact algebraic manipulation) from the conventional local Whittle estimator, which is based on the approximation $I_{X}(\omega_j ) \sim \omega_j^{-2d} I_u (\omega_j)$. Of course, the Whittle likelihood is itself an approximation of the exact likelihood.} form of the local Whittle estimator which is based on an exact algebraic manipulation of the Whittle likelihood and does not rely on differencing or tapering. They call this estimator the exact local Whittle (ELW) estimator and show that it is consistent and asymptotically normal when the optimization covers an interval of width less than $9/2$. 

From Proposition~\ref{Prop:SpecDens} we know that the discretized process satisfies equation (\ref{eq:SpecWhittleGeneral}). However, the discretized process is nonlinear, and  therefore we cannot apply the results cited above. Recently, several papers have studied the asymptotic properties of the LW estimator for nonlinear processes. For example, Hurvich and Ray (2003), Arteche (2004), and Hurvich, Moulines and Soulier (2005)  have studied the asymptotic properties of the LW estimator for signal-plus-noise processes in the presence of long-memory. However, even though the discretized process can be seen as a signal-plus-noise process (see Proof of Theorem~\ref{thm:Whittle} in Appendix~\ref{App:Proof2} for more details), we cannot apply their results because they assume that the noise is either a white noise or independent of the signal, which is not the case for the discretized process. Under regularity conditions Dalla, Giraitis and Hidalgo (2006) and Shao and Wu (2007) proved consistency and asymptotic normality of the LW for a very general class of nonlinear processes. Abadir, Distaso and Giraitis (2007) introduced a fully extended local Whittle (FELW) estimator, which is applicable not only for traditional cases but also for nonlinear and non-Gaussian processes, and they proved that it is consistent.


In this paper we follow the approach of Dalla, Giraitis and Hidalgo (2006) (henceforth DGH). The relevant result for our work is their Theorem 4; however, we cannot directly apply it for two reasons. First, DGH define long memory processes in the frequency domain, and therefore they impose assumptions directly on the spectral density. On the contrary, in the present paper we define long memory processes in the time domain (see Definitions~\ref{generalDef} and \ref{wellbehavedDef}) and therefore we want to impose assumptions on the autocovariance function. For this reason we introduce Lemma~\ref{Lemma:SpecDensGaussian} that makes the connection between the class of processes considered here and those studied in DGH.
Second, as we will show below, some assumptions in Theorem 4 of DGH is too restrictive for the processes considered in this paper. Instead, we can prove consistency and asymptotic normality of the LW estimator for the discretized process under weaker assumptions.

%
%

%
\begin{assumption}\label{Ass:CoeffGenDef2}
In Definition \ref{wellbehavedDef}, $\beta_1\neq 2H-1$, and either $2H\neq\beta_1\leq 2$ or $D/2 \neq -\sum\limits_{i:  \beta_i \neq 2H-1}^{I} b_i \zeta(2-2H+\beta_i)$,  where $\zeta(\cdot)$ is the Riemann zeta function.
\end{assumption}
Note that if $L(k)$ in Definition~\ref{wellbehavedDef} is analytic at infinity, then $\mathbb{N}  \ni \beta_1\neq 2H-1$, but not necessarily $\beta_1\leq 2$.
%
\begin{lemma}\label{Lemma:SpecDensGaussian}
Let $\{X(t)\}_{t\in\mathbb{N}}$ be a stationary Gaussian process with autocovariance given by Definition~\ref{generalDef} with  $L \in \mathcal{L}$. Let Assumptions~\ref{Ass:AbsSum} and \ref{Ass:CoeffGenDef2} hold. Then, as $\omega\to0^+$, the spectral density satisfies
\beq\label{eq:SpecDensGaussian}
\phi\left(\omega\right) = c_{\phi} |\omega|^{1-2H} \left(b_0 + c_{\beta} |\omega|^{\beta} + o\left(|\omega|^{\beta}\right)\right)\,,
\eeq
where $c_{\phi}>0$ is defined as in Proposition~\ref{Prop:GenSpecDens}, $c_\beta\neq 0$ and $\beta \in (0, 2]$.
\end{lemma}

In the above lemma $\beta$ is either $\beta_1$ or $2H-1$, depending on the specific functional form of the autocovariance of the underlying Gaussian process (see proof below). 
The fGn  is a special case of Lemma~\ref{Lemma:SpecDensGaussian} with $\beta=\beta_1=2$. The fARIMA process does not satisfy the conditions of Lemma~\ref{Lemma:SpecDensGaussian} (specifically, Assumption~\ref{Ass:AbsSum} $(i)$), but nonetheless its spectral density satisfies (\ref{eq:SpecDensGaussian}) with $\beta=\beta_1=2$

More important, Lemma~\ref{Lemma:SpecDensGaussian} provides conditions under which the underlying Gaussian process $\{X(t)\}_{t\in \mathbb{N}}$ satisfies assumption $T(\alpha_0, \beta)$ in DGH with $\alpha_0 = 2H-1$.

In order to prove asymptotic normality we need some smoothness assumption on the spectral representation of the underlying process. We assume that $\{X(t)\}_{t\in \mathbb{N}}$ is 
purely non-deterministic, so that it is also linear with finite fourth moments. Therefore, we can write
\beq
X_t = \sum_{j=0}^{\infty} a_j \varepsilon_{t-j}\nonumber
\eeq
where $\sum_{j=0}^{\infty} a_j^2 < \infty$ and $\varepsilon_t$ are i.i.d. Gaussian variables with zero mean and unit variance. Let 
\beq
\alpha(\omega) = \sum_{j=0}^{\infty} a_j e^{i j \omega}\nonumber
\eeq
We introduce the following assumption
\begin{assumption}\label{Ass:SmoothSpecAlpha}
\beq
\frac{d }{d \omega} \alpha(\omega) = O\left(\frac{|\alpha(\omega)|}{\omega}\right) \qquad as \;\; \omega \to 0^+
\eeq
\end{assumption}

Assumption~\ref{Ass:SmoothSpecAlpha} is the same as Assumption A2$^\prime$ in Robinson (1995b) and is the standard smoothness assumption used in the literature to prove asymptotic normality of the LW estimator.
Unfortunately, we are not able to relax this assumption for the class of long-memory processes considered in this paper.

%
%

%


Following Theorems 3 and 4 in DGH and under the assumptions above we derive the following theorem for a generic measurable transformation of an underlying Gaussian process.

\begin{theorem}\label{thm:Whittle}
Let $\{X(t)\}_{t \in \mathbb{N}}$  be a stationary purely non-deterministic Gaussian process with autocovariance given by Definition \ref{generalDef} and $L \in \mathcal{L}$. Let $Y(t) = g\left(X(t)\right)$ $\forall\, t \in \mathbb{N}$ for some measurable function $g$ with Hermite rank $j_0=1$. 
Then, as $n\to\infty$, the LW estimator $\widehat{H}^Y_{LW}$ of the process $\{Y(t)\}_{t\in \mathbb{N}}$ satisfies:
\begin{itemize}
\item[(i)] $\widehat{H}^Y_{LW} \overset{p}\rightarrow H$.
%
%
\item[(ii)] Let Assumptions~\ref{Ass:AbsSum}--\ref{Ass:SmoothSpecAlpha} hold. If  $m=o\left(n^{2\beta/(2\beta+1)}\right)$, then
\beqnr
\widehat{H}^Y_{LW} - H &=& O_{P} \left(m^{-1/2} + \left(\frac{m}{n}\right)^{\beta} + \left(\frac{m}{n}\right)^{r}\right)\label{eq:WhittleConvRate}
\eeqnr
where $\beta$ is defined as in Lemma~\ref{Lemma:SpecDensGaussian}, and 
\beqnr\nonumber
r =
\begin{cases}
 H-1/2 \qquad &if \;\; j_1 (2-2H) > 1\\
 (j_1-1) (1-H) \qquad &if \;\; j_1 (2-2H) < 1\\
 (j_1-1)/(2 j_1) - \varepsilon \qquad &if \;\; j_1(2-2H) = 1
\end{cases}
\eeqnr
for any $\varepsilon > 0$, where $j_1\geq 2$ is the order of the second non-vanishing Hermite coefficient.
%
%
%
%

\item[(iii)]In addition, if $m=o(n^{2r/(2r+1)})$, then
\beqnr\label{eq:WhittleAsympNorm}
\sqrt{m} \left(\widehat{H}^Y_{LW} - H\right) \overset{d}{\rightarrow} N\left(0, \frac{1}{4}\right).
\eeqnr
\end{itemize}
\end{theorem}

Note that, contrary to DGH in their Theorem 4 (see {\it Assumption LM} and equation (50) therein), we do not require $\beta=2$ in (\ref{eq:SpecDensGaussian}). As shown in the proof, we do not need this assumption to prove consistency and asymptotic normality of the LW estimator. Moreover, to establish an upper bound on the convergence rate of the LW estimate, i.e. (\ref{eq:WhittleConvRate}), we do not need to assume $m\geq n^{\gamma}$ for some $0<\gamma<1$ as in DGH (see equations (56) and (58) therein). Our trimming conditions on $m$ are weaker than theirs and our asymptotic results hold true even if $m(n)\to \infty$ more slowly than any power of $n$. 

Theorem~\ref{thm:Whittle} applies to the discretized process because the rounding is a measurable transformations of the underlying Gaussian process. So, we have the following

\begin{corollary}\label{Cor:Whittle}
Let $\{X(t)\}_{t \in \mathbb{N}}$  be a stationary purely non-deterministic Gaussian process with autocovariance given by Definition \ref{generalDef} and $L \in \mathcal{L}$. Then, the LW estimator $\widehat{H}^{d}_{LW}$ of the discretized process $\{X_d(t)\}_{t \in \mathbb{N}}$  is consistent. 

Moreover, if Assumptions~\ref{Ass:AbsSum}--\ref{Ass:SmoothSpecAlpha} hold and $m=o\left(n^{2\beta/(2\beta+1)}\right)$, where $\beta$ is defined as in Lemma~\ref{Lemma:SpecDensGaussian}, then $\widehat{H}^{d}_{LW}$ satisfies (\ref{eq:WhittleConvRate}) with
\beqnr\nonumber
r =
\begin{cases}
 H-1/2 \qquad &if \;\; H< 5/6\\
 2 (1-H) \qquad &if \;\; H>5/6\\
 1/3 - \varepsilon \qquad &if \;\; H= 5/6
\end{cases}
\eeqnr
for any $\varepsilon>0$.

In addition, if $m=o(n^{2r/(2r+1)})$, then $\widehat{H}^{d}_{LW}$ is asymptotically normal and satisfies (\ref{eq:WhittleAsympNorm}).
\end{corollary}

For fGn and fARIMA processes $\beta=2$, and therefore $\hat{H}_Y^{LW} - H = O_{P} \left(m^{-1/2} + \left(\frac{m}{n}\right)^{r}\right)$, with $r\in (0, 1/3)$. Moreover, if $m=n^{\gamma}$ with $\gamma > 0.4$, the term $O_P\left((m/n)^r\right)$ will dominate the term $O_P(m^{-1/2})$. As we discuss below, this consideration is important in order to understand the negative finite sample bias we observe in our numerical simulations.

\subsubsection{Numerical simulations} 

We numerically generated the Gaussian process as a fGn with unit variance and we considered two values of the Hurst exponent, namely  $H=0.7$ and $0.85$. This choice of  Hurst exponents is motivated by the need to have values below and above the critical value $H=5/6$. In order to test the dependence on the time series length we considered series of length $n=2^{10}$ and $n=2^{14}$. We then applied four discretization procedures corresponding to discretization parameter $\chi=0.1, 0.25, 0.5$ and the sign of the process. We performed $L=10^3$ simulations in order to obtain the statistical properties of the estimators.  


Our results are summarized in Tables~\ref{tab:LW07} and \ref{tab:LW085}. The tables show the mean,  standard error, and three quantiles (2.5\%, 50\%, and 97.5\%) of the estimator $\widehat{H}_{LW}$ defined by (\ref{LW:Def}) for four thresholds, namely $m=n^{0.5}$, $n^{0.6}$, $n^{0.7}$~and~$n^{0.8}$. We chose these thresholds to be consistent with DGH.

For the fGn the LW estimate increases with the threshold $m$, with a small negative bias for $m=n^{0.5}$ and a small positive bias for $m=n^{0.8}$. The bias is minimized for values of $m$ between $n^{0.6}$ and $n^{0.7}$, in agreement with the results for fARIMA processes reported in Robinson (1995b) and DGH (see Table I therein). As expected, the absolute bias decreases as we increase the sample size. For $n=2^{14}$, when $H=0.7$ the bias is not statistically significant at the 5\% level for $m=n^{0.5}$, $n^{0.6}$, $n^{0.7}$; when $H=0.85$ it is not statistically significant at the 5\% level for $m=n^{0.5}$. The standard error obviously decreases with $m$.

For both the discretized process and the sign process we observe a strong and statistically significant negative bias for all values of $n$, $m$, and $H$. The bias becomes more severe as the discretization becomes coarser, i.e. as $\delta$ increases. In general, the absolute bias increases with the threshold\footnote{The only exception is for $H=0.7$ in small samples ($n=2^{10}$) and for moderate discretizations ($\chi=0.25$, $0.5$). In this case the LW does change in a statistically significant way across different thresholds.}. 
This is because, as $m$ increases, the objective function includes higher-frequency ordinates of the periodogram. From (\ref{eq:LWAsympExpan}) in the proof of Theorem~\ref{thm:Whittle} (see Appendix~\ref{App:Proof2}) we know that the $O_P\left((m/n)^r\right)$ term in the asymptotic expansion of $\widehat{H}_{LW}$ represents the contribution of the sample estimate, taken with the negative sign, of the higher-order components of the spectral density. From Corollary~\ref{Cor:SpecDensRoundfGn} we know that for the discretization of a fGn the second-order term of the spectral density is strictly positive for all $H$. Therefore, intuitively, the term $O_P\left((m/n)^r\right)$ generates the negative finite-sample bias that we observe here. On the other hand, as expected, the standard error decreases as $m$ increases.  
The LW estimator with $m=n^{0.5}$ is less biased by the discretization than the estimator with threshold $m=n^{0.8}$, but it has a larger dispersion and is more noisy. These properties are expected from the bias-variance tradeoff.  Note that for $n=2^{10}$ and $m=n^{0.8}$ the 97.5\% quantile of $\widehat{H}_{LW}$ is smaller than  the true $H$ for $\chi=0.1$ if $H=0.7$, and for both $\chi=0.1$ and the sign process if $H=0.85$. Even more strikingly, for $n=2^{14}$ the 97.5\% quantile of $\widehat{H}_{LW}$ is smaller than the true $H$ for the sign process and for $\chi=0.1$, $0.25$, both for $H=0.7$ and for $H=0.85$.

We also note that the effect of the bias due to the discretization is increasing with the long memory parameter, i.e. the larger is $H$, the larger is the discretization bias. Finally, note that the bias is a finite size effect. 

We ran the same simulations for various fARIMA, and we got similar results.



\begin{table}[ht]
\begin{center}
{\fontsize{10}{8} \selectfont
\begin{tabular}{lrcrrcrrcrrcr}
  \hline
\multirow{3}{*}{} & & \multicolumn{2}{c}{$m=n^{0.5}$} & & \multicolumn{2}{c}{$m=n^{0.6}$} & & \multicolumn{2}{c}{$m=n^{0.7}$} & & \multicolumn{2}{c}{$m=n^{0.8}$} \\
& & & $2.5\%$ & & & $2.5\%$ & & & $2.5\%$ & & & $2.5\%$\\
& &  Mean($\widehat{H}$) & $50\%$ & & Mean($\widehat{H}$) & $50\%$ & & Mean($\widehat{H}$) & $50\%$ & & Mean($\widehat{H}$) & $50\%$\\
& & (SE($\widehat{H}$)) & $97.5\%$ & & (SE($\widehat{H}$)) & $97.5\%$ & & (SE($\widehat{H}$)) & $97.5\%$ & & (SE($\widehat{H}$)) & $97.5\%$\\
  \hline
  \hline
& \multicolumn{4}{c}{} & & \multicolumn{2}{c}{$n=2^{10}$} & & \multicolumn{4}{c}{}\\
  \hline
  &   &  & 0.4411 &  &  & 0.5097 &  &  & 0.5584 &  &  & 0.5917 \\ 
Sign   &  & 0.6610 & 0.6606 &  & 0.6583 & 0.6584 &  & 0.6578 & 0.6570 &  & 0.6561 & 0.6563 \\ 
   &  & (0.0035) & 0.8670 &  & (0.0023) & 0.7951 &  & (0.0016) & 0.7553 &  & (0.0011) & 0.7219 \\ 
   &  &  & 0.4109 &  &  & 0.4828 &  &  & 0.5323 &  &  & 0.5507 \\ 
  $\chi=0.1$ &  & 0.6350 & 0.6393 &  & 0.6308 & 0.6341 &  & 0.6272 & 0.6287 &  & 0.6223 & 0.6227 \\ 
   &  & (0.0035) & 0.8339 &  & (0.0023) & 0.7646 &  & (0.0015) & 0.7166 &  & (0.0011) & 0.6884 \\ 
   &  &  & 0.4375 &  &  & 0.5069 &  &  & 0.5708 &  &  & 0.6073 \\ 
  $\chi=0.25$ &  & 0.6710 & 0.6757 &  & 0.6709 & 0.6726 &  & 0.6717 & 0.6739 &  & 0.6726 & 0.6730 \\ 
   &  & (0.0036) & 0.8899 &  & (0.0024) & 0.8053 &  & (0.0016) & 0.7644 &  & (0.0010) & 0.7350 \\ 
   &  &  & 0.4625 &  &  & 0.5287 &  &  & 0.5884 &  &  & 0.6266 \\ 
   $\chi=0.5$  &  & 0.6809 & 0.6855 &  & 0.6815 & 0.6842 &  & 0.6861 & 0.6883 &  & 0.6898 & 0.6907 \\ 
   &  & (0.0035) & 0.8894 &  & (0.0023) & 0.8138 &  & (0.0015) & 0.7745 &  & (0.0010) & 0.7555 \\ 
   &  &  & 0.4641 &  &  & 0.5406 &  &  & 0.6050 &  &  & 0.6483 \\ 
   Continuous &  & 0.6905 & 0.6922 &  & 0.6935 & 0.6955 &  & 0.7014 & 0.7034 &  & 0.7113 & 0.7118 \\ 
   &  & (0.0036) & 0.9113 &  & (0.0023) & 0.8300 &  & (0.0015) & 0.7955 &  & (0.0010) & 0.7761 \\ 
  \hline
  \hline
& \multicolumn{4}{c}{} & & \multicolumn{2}{c}{$n=2^{14}$} & & \multicolumn{4}{c}{}\\
  \hline
 &  &  & 0.5773 &  &  & 0.6211 &  &  & 0.6336 &  &  & 0.6409 \\ 
  Sign &  & 0.6814 & 0.6831 &  & 0.6768 & 0.6763 &  & 0.6688 & 0.6691 &  & 0.6606 & 0.6606 \\ 
   &  & (0.0016) & 0.7772 &  & (0.0009) & 0.7331 &  & (0.0005) & 0.7019 &  & (0.0003) & 0.6816 \\ 
   &  &  & 0.5594 &  &  & 0.6014 &  &  & 0.6102 &  &  & 0.6106 \\ 
   $\chi=0.1$ &  & 0.6659 & 0.6691 &  & 0.6567 & 0.6575 &  & 0.6447 & 0.6448 &  & 0.6323 & 0.6327 \\ 
   &  & (0.0016) & 0.7580 &  & (0.0009) & 0.7105 &  & (0.0006) & 0.6791 &  & (0.0003) & 0.6536 \\ 
   &  &  & 0.5811 &  &  & 0.6267 &  &  & 0.6469 &  &  & 0.6544 \\ 
   $\chi=0.25$ &  & 0.6880 & 0.6901 &  & 0.6849 & 0.6853 &  & 0.6802 & 0.6798 &  & 0.6752 & 0.6753 \\ 
   &  & (0.0016) & 0.7816 &  & (0.0009) & 0.7396 &  & (0.0006) & 0.7156 &  & (0.0003) & 0.6961 \\ 
   &  &  & 0.5874 &  &  & 0.6318 &  &  & 0.6561 &  &  & 0.6687 \\ 
   $\chi=0.5$ &  & 0.6928 & 0.6930 &  & 0.6923 & 0.6933 &  & 0.6900 & 0.6895 &  & 0.6887 & 0.6882 \\ 
   &  & (0.0016) & 0.7897 &  & (0.0009) & 0.7469 &  & (0.0005) & 0.7247 &  & (0.0003) & 0.7094 \\ 
   &  &  & 0.5941 &  &  & 0.6406 &  &  & 0.6675 &  &  & 0.6831 \\ 
   Continuous &  & 0.6988 & 0.6993 &  & 0.7000 & 0.7005 &  & 0.7008 & 0.7006 &  & 0.7043 & 0.7044 \\ 
   &  & (0.0016) & 0.7927 &  & (0.0009) & 0.7555 &  & (0.0006) & 0.7344 &  & (0.0003) & 0.7253 \\ 
   \hline
\end{tabular}
}
\caption{Estimation of the Hurst exponent of a fGn with $H=0.7$ and its discretization with the sign function and with $\chi=0.1$, $0.25$, and $0.5$. The time series has length $2^{10}$ (top) and $2^{14}$ (bottom). The estimation of the Hurst exponent is obtained with the Local Whittle estimator by maximizing the objective function (\ref{LW:Def}) (see text). The table reports the mean value of the estimator $\widehat{H}$ over $10^3$ numerical simulations, together with the standard error (SE) and the $2.5\%$, $50\%$, and $97.5\%$ percentile.}
\label{tab:LW07}
\end{center}
\end{table}


\begin{table}[ht]
\begin{center}
{\fontsize{10}{8} \selectfont
\begin{tabular}{lrcrrcrrcrrcr}
  \hline
\multirow{3}{*}{} & & \multicolumn{2}{c}{$m=n^{0.5}$} & & \multicolumn{2}{c}{$m=n^{0.6}$} & & \multicolumn{2}{c}{$m=n^{0.7}$} & & \multicolumn{2}{c}{$m=n^{0.8}$} \\
& & & $2.5\%$ & & & $2.5\%$ & & & $2.5\%$ & & & $2.5\%$\\
& &  Mean($\widehat{H}$) & $50\%$ & & Mean($\widehat{H}$) & $50\%$ & & Mean($\widehat{H}$) & $50\%$ & & Mean($\widehat{H}$) & $50\%$\\
& & (SE($\widehat{H}$)) & $97.5\%$ & & (SE($\widehat{H}$)) & $97.5\%$ & & (SE($\widehat{H}$)) & $97.5\%$ & & (SE($\widehat{H}$)) & $97.5\%$\\
  \hline
  \hline
& \multicolumn{4}{c}{} & & \multicolumn{2}{c}{$n=2^{10}$} & & \multicolumn{4}{c}{}\\
  \hline
 &  &  & 0.5706 &  &  & 0.6519 &  &  & 0.6872 &  &  & 0.7056 \\ 
  Sign  &  & 0.8109 & 0.8179 &  & 0.8018 & 0.8029 &  & 0.7903 & 0.7924 &  & 0.7776 & 0.7783 \\ 
    &  & (0.0036) & 1.0294 &  & (0.0023) & 0.9474 &  & (0.0015) & 0.8801 &  & (0.0011) & 0.8457 \\ 
    &  &  & 0.5228 &  &  & 0.5928 &  &  & 0.6281 &  &  & 0.6362 \\ 
   $\chi=0.1$ &  & 0.7671 & 0.7681 &  & 0.7525 & 0.7530 &  & 0.7397 & 0.7394 &  & 0.7268 & 0.7268 \\ 
    &  & (0.0038) & 1.0013 &  & (0.0026) & 0.9072 &  & (0.0018) & 0.8495 &  & (0.0015) & 0.8122 \\ 
    &  &  & 0.5797 &  &  & 0.6515 &  &  & 0.7118 &  &  & 0.7312 \\ 
   $\chi=0.25$ &  & 0.8271 & 0.8354 &  & 0.8183 & 0.8220 &  & 0.8092 & 0.8091 &  & 0.7984 & 0.7993 \\ 
    &  & (0.0036) & 1.0488 &  & (0.0024) & 0.9541 &  & (0.0015) & 0.9043 &  & (0.0011) & 0.8639 \\ 
    &  &  & 0.6008 &  &  & 0.6737 &  &  & 0.7365 &  &  & 0.7612 \\ 
    $\chi=0.5$ &  & 0.8350 & 0.8423 &  & 0.8338 & 0.8389 &  & 0.8324 & 0.8332 &  & 0.8288 & 0.8302 \\ 
    &  & (0.0037) & 1.0538 &  & (0.0024) & 0.9739 &  & (0.0015) & 0.9228 &  & (0.0011) & 0.8919 \\ 
    &  &  & 0.6041 &  &  & 0.6930 &  &  & 0.7612 &  &  & 0.8004 \\ 
    Continuous &  & 0.8468 & 0.8554 &  & 0.8507 & 0.8538 &  & 0.8576 & 0.8591 &  & 0.8694 & 0.8698 \\ 
    &  & (0.0036) & 1.0568 &  & (0.0024) & 0.9952 &  & (0.0015) & 0.9484 &  & (0.0011) & 0.9338 \\ 
  \hline
  \hline
& \multicolumn{4}{c}{} & & \multicolumn{2}{c}{$n=2^{14}$} & & \multicolumn{4}{c}{}\\
  \hline
   &  &  & 0.7358 &  &  & 0.7632 &  &  & 0.7720 &  &  & 0.7675 \\ 
  Sign &  & 0.8331 & 0.8353 &  & 0.8242 & 0.8242 &  & 0.8077 & 0.8086 &  & 0.7887 & 0.7886 \\ 
   &  & (0.0015) & 0.9206 &  & (0.0009) & 0.8783 &  & (0.0006) & 0.8398 &  & (0.0003) & 0.8098 \\ 
   &  &  & 0.7153 &  &  & 0.7350 &  &  & 0.7317 &  &  & 0.7193 \\ 
  $\chi=0.1$  &  & 0.8143 & 0.8131 &  & 0.7971 & 0.7965 &  & 0.7735 & 0.7727 &  & 0.7484 & 0.7484 \\ 
   &  & (0.0016) & 0.9137 &  & (0.0010) & 0.8563 &  & (0.0007) & 0.8156 &  & (0.0005) & 0.7786 \\ 
   &  &  & 0.7447 &  &  & 0.7797 &  &  & 0.7898 &  &  & 0.7878 \\ 
  $\chi=0.25$ &  & 0.8425 & 0.8430 &  & 0.8374 & 0.8379 &  & 0.8252 & 0.8258 &  & 0.8092 & 0.8094 \\ 
   &  & (0.0015) & 0.9338 &  & (0.0009) & 0.8900 &  & (0.0006) & 0.8584 &  & (0.0003) & 0.8302 \\ 
   &  &  & 0.7486 &  &  & 0.7840 &  &  & 0.8028 &  &  & 0.8117 \\ 
  $\chi=0.5$ &  & 0.8464 & 0.8478 &  & 0.8450 & 0.8449 &  & 0.8382 & 0.8389 &  & 0.8309 & 0.8308 \\ 
   &  & (0.0015) & 0.9414 &  & (0.0009) & 0.8994 &  & (0.0006) & 0.8703 &  & (0.0003) & 0.8514 \\ 
   &  &  & 0.7497 &  &  & 0.7916 &  &  & 0.8158 &  &  & 0.8368 \\ 
  Continuous  &  & 0.8503 & 0.8506 &  & 0.8518 & 0.8514 &  & 0.8516 & 0.8519 &  & 0.8565 & 0.8569 \\ 
   &  & (0.0015) & 0.9427 &  & (0.0009) & 0.9057 &  & (0.0006) & 0.8856 &  & (0.0003) & 0.8770 \\ 
   \hline
\end{tabular}
}
\caption{Estimation of the Hurst exponent of a fGn with $H=0.85$ and its discretization with the sign function and with $\chi=0.1$, $0.25$, and $0.5$. The time series has length $2^{10}$ (top) and $2^{14}$ (bottom). The estimation of the Hurst exponent is obtained with the Local Whittle estimator by maximizing the objective function (\ref{LW:Def}) (see text). The table reports the mean value of the estimator $\widehat{H}$ over $10^3$ numerical simulations, together with the standard error (SE) and the $2.5\%$, $50\%$, and $97.5\%$ percentile.}
\label{tab:LW085}
\end{center}
\end{table}



\subsection{Detrended Fluctuation Analysis}\label{DFA}

\subsubsection{Definition and notation}

We now consider the Detrended Fluctuation Analysis (DFA) (see Peng {\it et al.} (1994)), a method to investigate the properties of a long memory process and to estimate the Hurst exponent. This method was introduced more than fifteen years ago to investigate physiological data, in particular the heartbeat signal. Since its introduction it has been applied to a large variety of systems, including physical, biological, economic, and technological data. Some examples of its application to economics and financial time series can be found in Schmitt {\it et al.} (2000), Lillo and Farmer (2004), Di Matteo {\it et al.} (2005), Yamasaki {\it et al.} (2005), and Alfarano and Lux (2007). The idea is to consider the integrated process and detrend it locally. The scaling of the fluctuations of the residuals as a function of the box size in which the regression is performed gives the estimate of the Hurst exponent. 
More precisely, let $\{X(t)\}$, $t=1, \ldots, n$, be a finite sample from a process $\{X(t)\}_{t\in\mathbb{N}}$ and denote the discrete integration of this sample as 
\beq
Y(k) = \sum_{t=1}^{k} X(t)\nonumber
\eeq

The integrated time series is divided into $[n/m]$ boxes of equal length $m$, where $[z]$ is the integer part of $z$. 
%
%
In each box a least squares line is fit to the data (representing the trend in that box).  The $y$-coordinate of the straight line segments is denoted by $\widehat{Y}_m(k)$.  Next, one detrends the integrated time series, $Y(k)$, by subtracting the local trend, $\widehat{Y}_m(k)$, in each box.  For any given box size $m$, the root-mean-square fluctuation (or, simply, fluctuation function) of this integrated and detrended time series is calculated by
\begin{equation}\label{eq:DFAfunctionDef}
F(m)=\sqrt{\frac{1}{m \cdot [n/m]}\sum_{k=1}^{m \cdot [n/m]} [Y(k)-\widehat{Y}_m(k)]^2}.
\end{equation}
This computation is repeated over all time scales (box sizes) to characterize the relationship between $F(m)$ and the box size $m$.  Typically $F(m)$ increases with box size $m$. For a long memory process the Hurst exponent is given by the relation  $F(m)\propto m^H$. Therefore $H$ is estimated  by performing a log-regression of  $F(m)$ versus $m$. The proposers of this method claim that DFA is able to perform well also in the presence of non-stationarities, such as trends, (see Peng {\it et al.} (1994)), even if some recent results dispute this claim (see Bardet and Kammoun (2008)).

The partial DFA function computed in the $j$-th window of size $m$ is
\beq\nonumber
F_j^2(m)=\frac{1}{m} \sum_{k = m (j-1) + 1}^{m j} (Y(k) - \widehat{Y}_m(k))^2 
\eeq
for $j\in \{1, \ldots, [n/m]\}$. Then, we can write
\beq
F^2(m)=\frac{1}{[n/m]} \sum_{j = 1}^{[n/m]} F_j^2(m)
\eeq

Recently, Bardet and Kammoun (2008) (Lemma 2.2) showed that for a stationary process  the series $\{F_j(m)\}_{1\leq j \leq[n/m]}$ is a stationary process for all $m$.  This means that in order to study the asymptotic statistical properties of the DFA we can focus on $F_1^2(m)$ only.
%
%
They then used this result to provide an asymptotic expression (Theorem 4.2) for $\mathbb{E} \left[F_1^2(m)\right]$ for a general class of Gaussian long memory processes. Here we extend this result to non-Gaussian long memory processes (see Theorem \ref{ThmExpF} below). In doing that we show that the order of the correction of the asymptotic expansion given in Bardet and Kammoun (2008)  for generic Gaussian processes (Theorem 4.2) is incorrect and we provide the correct order. For the special case of fGn and fARIMA the result provided in Bardet and Kammoun (2008) is correct because the prefactor of the correction term observed for generic long memory processes cancels out exactly.    

\subsubsection{A theorem on the detrended fluctuation analysis of a general long memory process}

Here we generalize Theorem 4.2 of  Bardet and Kammoun (2008) to a general class of non-Gaussian long memory processes. The discretized process belongs to this class. 

%
\begin{theorem}\label{ThmExpF}
Let  $\{X(t)\}_{t \in \mathbb{N}}$ be a weakly stationary long memory process, with zero mean, finite variance, and the following autocovariance function 
\beq
\gamma(k)=A k^{2H-2}\left(1 + O\left(k^{-\beta}\right)\right) \quad as \;\; k\rightarrow\infty\,,
\eeq
with $A>0$, $H\in(0.5,1)$, and $\beta>0$. Then, for large $m$
\beqnr
\mathbb{E} \left[F_1^2(m)\right] = A' f(H) m^{2 H}
\begin{cases}
\left(1 + O\left(m^{-\min(2H-1, \beta)}\right)\right) & \; if \;\; \beta\neq2H-1\\
\left(1 + O\left(m^{1-2H} \ln m\right)\right) & \; if \;\; \beta=2H-1
\end{cases}
\eeqnr
where 
$A'=\frac{A}{H (2H-1)}$, and  
$f(H)=\frac{1 - H}{(1 + H) (2 + H) (1 + 2 H)}$.
\end{theorem}

 First of all, note that the class of processes for which one can apply this theorem is more general than that of Definition \ref{wellbehavedDef}, but clearly less general than that of Definition \ref{generalDef}. Second, our theorem gives a different order of the correction compared to the one given in Theorem 4.2 of Bardet and Kammoun (2008). In fact, Bardet and Kammoun give $O(m^{-\min(1,\beta)})$, while our theorem gives $O(m^{-\min(2H-1,\beta)})$. In many cases of interest (for example when $\beta>1$) the leading error term is $O(m^{1-2H})$ rather than $O(m^{-1})$ as stated in Bardet and Kammoun (2008).
 
For fGn or fARIMA$(0, d, 0)$ we observe a specific behavior. In fact, in these cases the error term is $O(m^{-1})$, instead of $O(m^{-\min(2H-1, \beta)})$, because the prefactor of the terms of order $O(m^{1-2H})$ cancels out exactly, $\beta=2$ and the next sub-leading term is indeed $O(m^{-1})$. This depends on the specific functional form of the autocovariance of these processes. Therefore, as stated in Bardet and Kammoun (2008) Property 3.1, we have: 
\begin{corollary}\label{DFACorFGN}
For a fractional Gaussian noise (fGn) with variance $D$,
\beq
\mathbb{E} \left[F_1^2(m)\right] = D f(H) m^{2 H} (1 + O(m^{-1}))\,.
\eeq
\end{corollary}

The general Theorem \ref{ThmExpF} gives the asymptotic properties of the fluctuation function for a time series obtained by discretizing or taking the sign of a Gaussian process. Specifically, for discretized processes we have the following

\begin{corollary}\label{DFACorDiscrete}
For the discretization of a Gaussian time series with the properties of Theorem \ref{ThmExpF}, for large $m$ 
\beqnr
\mathbb{E} \left[F_1^2(m)\right] = \left(\frac{\vartheta_2(0,e^{-1/2\chi})}{\sqrt{2\pi \chi}}\right)^2 A' f(H) m^{2 H}
\begin{cases}
\left (1 + O\left(m^{1-2 H}\right)\right) & \; if \;\; H\in (1/2, 5/6)\\
\left(1 + O\left(m^{-2/3}\ln m\right)\right) & \; if \;\; H=5/6\\
\left(1 + O\left(m^{4 H- 4}\right)\right) & \; if \;\; H\in (5/6, 1)
\end{cases}
\eeqnr
where $\vartheta_a(u, q)$ is the elliptic theta function, and $A'$ and $f(H)$ are defined as in Theorem~\ref{ThmExpF}.
\end{corollary}
%
%


%
%
%

\subsubsection{Numerical simulations}

We then studied the performance of the DFA on finite time series. As in the numerical analysis of the LW estimator, we considered a fGn with  $H=0.7$ and $0.85$ and we generated $L=10^3$ time series of length $n=2^{10}$ and $2^{14}$. We then applied four discretization procedures: $\chi=0.1, 0.25, 0.5$, and the sign of the process. Following a standard practice, we consider values of $m$ ranging from $m=4$ to $m=n/4$.

Figure \ref{dfa1} show an example of the root-mean-square fluctuation $F(m)$ as a function of the block size $m$.  As in the case of the LW estimator, a different behavior is observed for the Gaussian process and for its discretization. For the former  $F(m)$ is well described by a power law with exponent $H$ over the whole range of investigated block sizes $m$. On the contrary, for the discretized (or sign) process,  $F(m)$ changes significantly slope in the log-log scale, i.e. it has a varying local power law behavior. Only for large box sizes the function  $F(m)$ converges to the expected asymptotic behavior from above. This again suggest a significant negative bias in the estimation of $H$ on finite samples, bias that we quantitatively observe in estimation below. In fact, in order to estimate the Hurst exponent through DFA one needs to perform a best fit of $F$ with a power law function. The ambiguity is however the interval of values of $m$ where one performs the fit. To the best of our knowledge there is no rule for selecting optimally such an interval. One expects to obtain a less biased, but noisier, estimate of $H$ by performing the fit in a small region corresponding to large values of the block size $m$. To investigate this point we estimate the Hurst exponent by performing the fit over a fraction $q$ of the largest values of $\log_{10}[m]$. We consider four values of $q$, namely $q=1$ (the whole interval), $q=0.75$ (the largest three-quarter), $q=0.5$ (the largest half), and $q=0.25$ (the largest quartile).

The results of our analysis are summarized in Tables \ref{TableDFACoarse07} and \ref{TableDFACoarse085}. For all $q$ we observe an underestimation of $\widehat{H}$ for the discretized series. This underestimation is always significant in terms of standard errors. Specifically, for short series ($n=2^{10}$) there is a  severe underestimation for the sign and for $\chi=0.1, 0.25$,  both when $H=0.7$ and when $H=0.85$. For long series ($n=2^{14}$) we observe a severe underestimation for the sign and for $\chi=0.1$,  both when $H=0.7$ and when $H=0.85$. Strikingly, when $q=1$ and $n=2^{14}$, the 97.5\% quantile of $\widehat{H}$ is smaller than (or roughly equal to) the true value of $H$ for the sign and for $\chi=0.1$, both when $H=0.7$ and when $H=0.85$. Finally, note that for the discretized process is not always true that the less biased estimator is obtained for small values of $q$ (i.e. large $m$). In fact, in many cases we observe quite the contrary. For example, in short samples ($n=2^{10}$) the DFA estimate with $q=1$ is less biased than that with $q=0.25$, except when $H=0.85$ and the discretization is very coarse (sign and $\chi=0.1$). 
These results might be due to the shape of the root-mean-square fluctuation function $F$ that changes slowly slope when the block size $m$ varies  and converges to its asymptotic value from above (see Figure \ref{dfa1}). 
The DFA estimator with small $q$ works better (in terms of bias) in long time series for coarser  discretizations and larger $H$.



\begin{figure}[t]
\begin{center}
\includegraphics[scale=0.28]{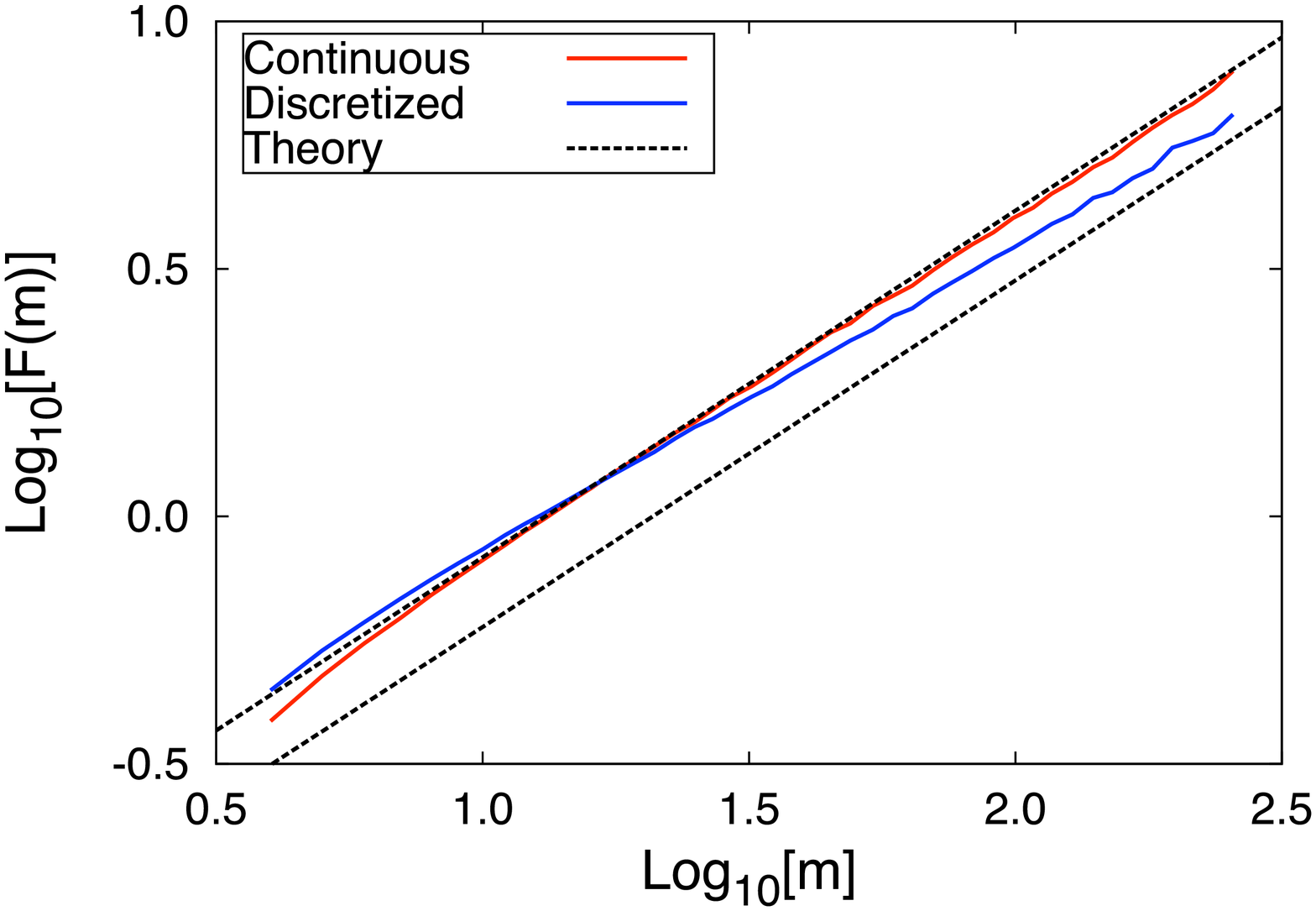}
\includegraphics[scale=0.28]{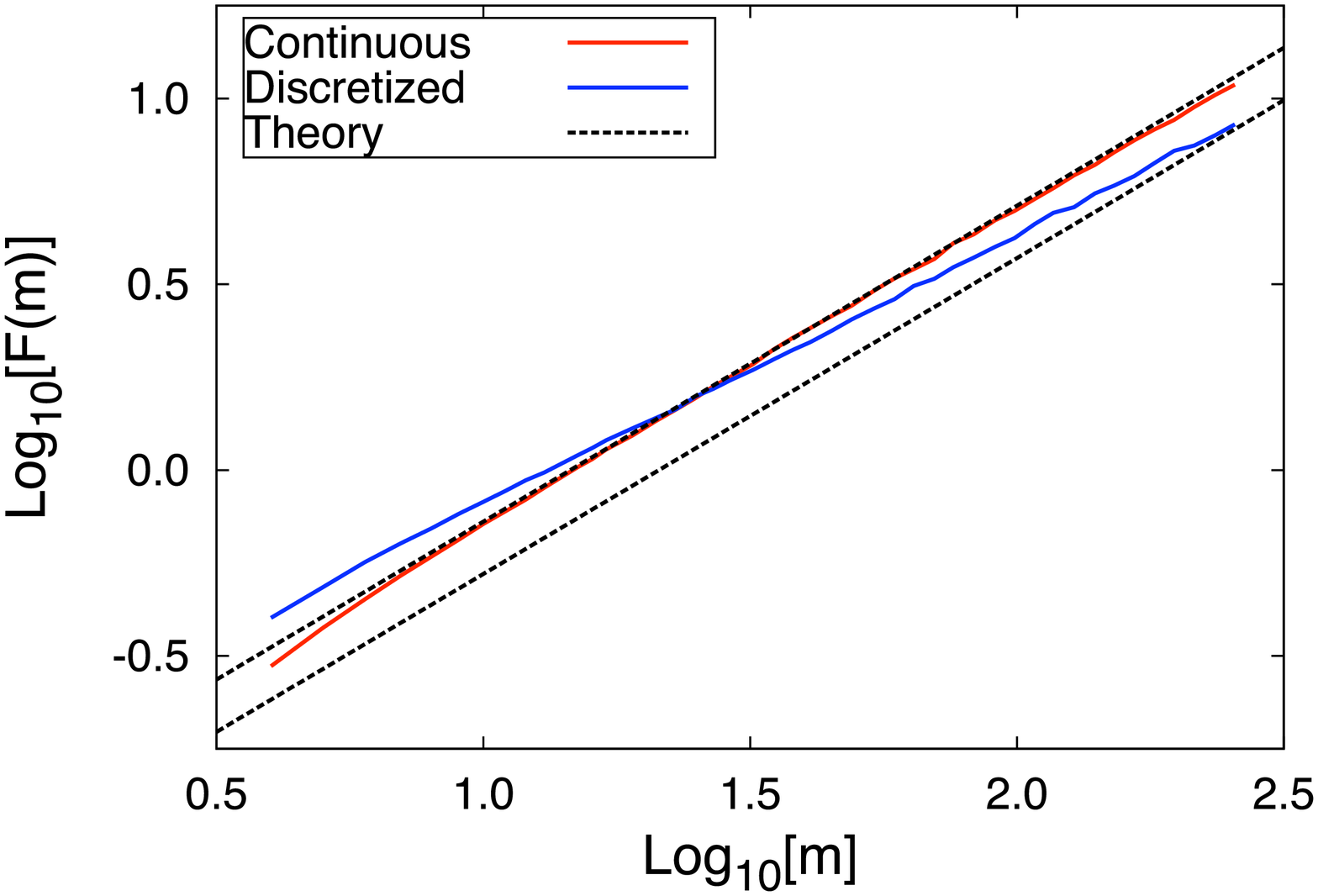}
\includegraphics[scale=0.28]{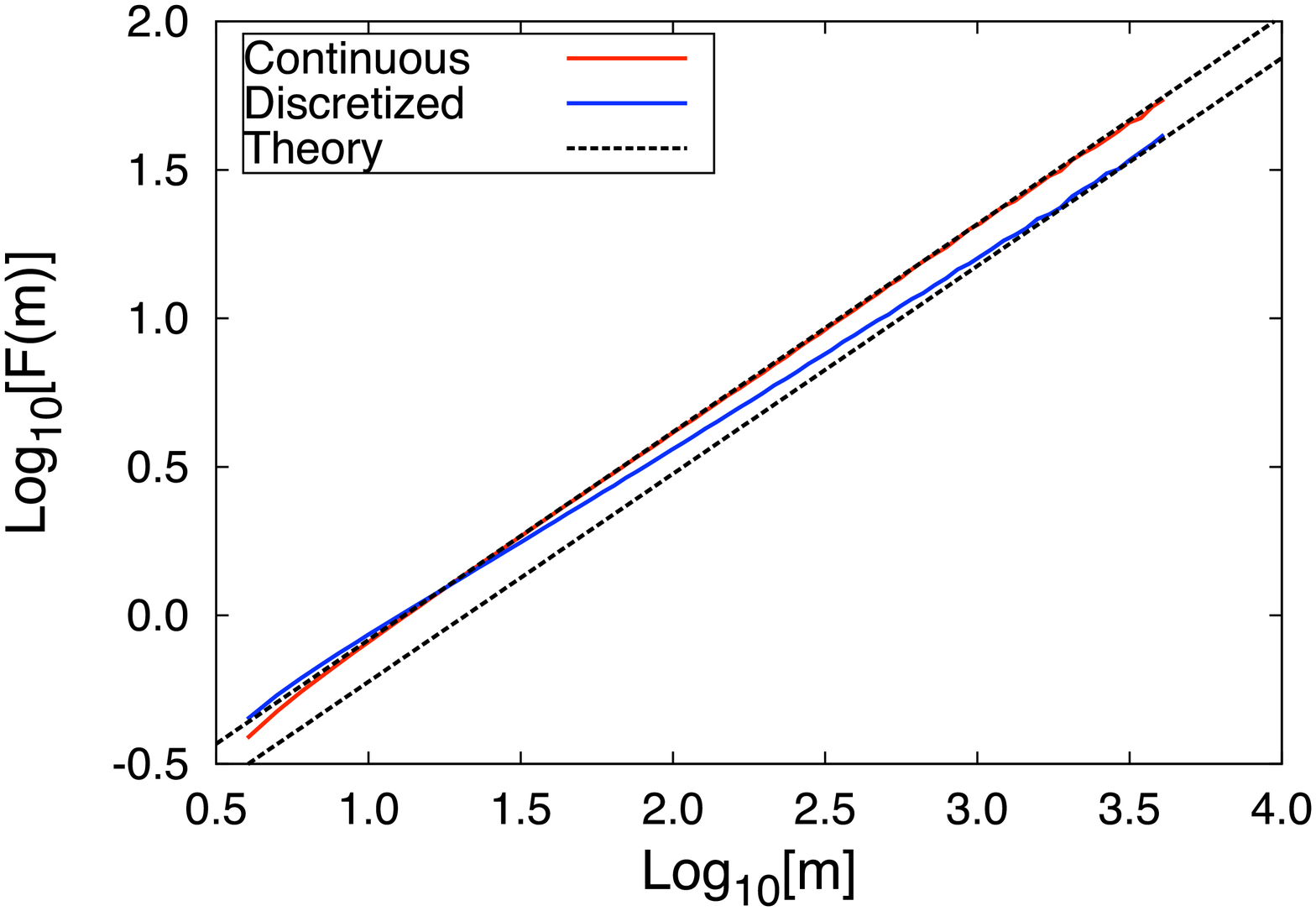}
\includegraphics[scale=0.28]{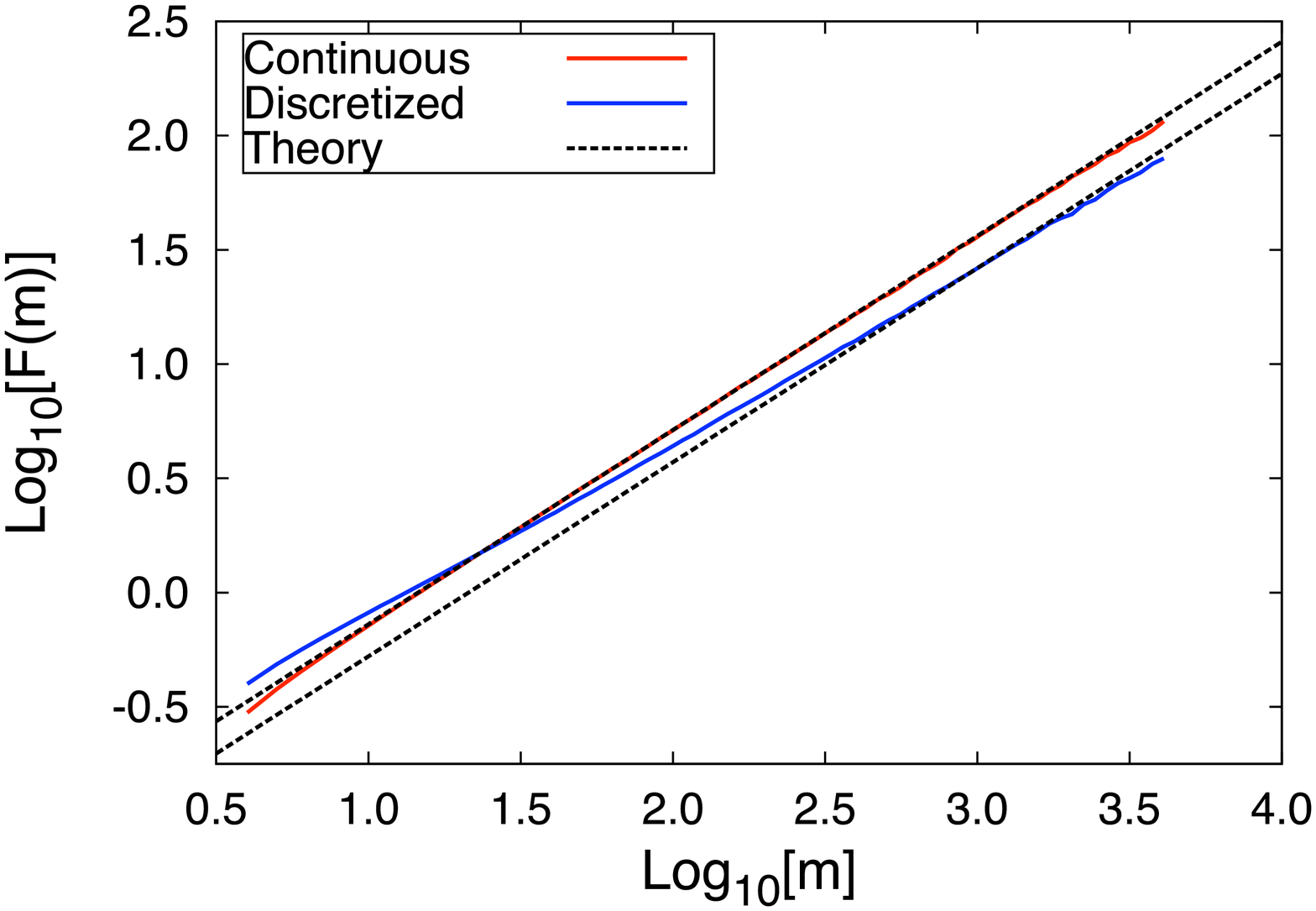}
\caption{Root mean square fluctuation $F(m)$ as a function of the block size $m$ of a numerical simulation of a fGn and its discretization with a scaling parameter $\chi=0.1$. The time series has length $n=2^{10}$ (top) and $n=2^{14}$ (bottom). The Hurst exponent of the fGn is $H=0.7$ (left) and $H=0.85$ (right). The figure also shows the leading term of the expansion of the $F(m)$ for large $m$ (dashed lines).}
\label{dfa1}
\end{center}
\end{figure}





\begin{table}[ht]
\begin{center}
{\fontsize{10}{8} \selectfont
\begin{tabular}{lrcrrcrrcrrcr}
  \hline
\multirow{3}{*}{} & & \multicolumn{2}{c}{$q=0.25$} & & \multicolumn{2}{c}{$q=0.5$} & & \multicolumn{2}{c}{$q=0.75$} & & \multicolumn{2}{c}{$q=1$} \\
& & & $2.5\%$ & & & $2.5\%$ & & & $2.5\%$ & & & $2.5\%$\\
& &  Mean($\widehat{H}$) & $50\%$ & & Mean($\widehat{H}$) & $50\%$ & & Mean($\widehat{H}$) & $50\%$ & & Mean($\widehat{H}$) & $50\%$\\
& & (SE($\widehat{H}$)) & $97.5\%$ & & (SE($\widehat{H}$)) & $97.5\%$ & & (SE($\widehat{H}$)) & $97.5\%$ & & (SE($\widehat{H}$)) & $97.5\%$\\
  \hline
  \hline
& \multicolumn{4}{c}{} & & \multicolumn{2}{c}{$n=2^{10}$} & & \multicolumn{4}{c}{}\\
  \hline
  &  &  & 0.3676 &  &  & 0.4771 &  &  & 0.5313 &  &  & 0.5807 \\ 
 Sign  &  & 0.6482 & 0.6430 &  & 0.6505 & 0.6469 &  & 0.6493 & 0.6462 &  & 0.6562 & 0.6547 \\ 
   &  & (0.0052) & 0.9836 &  & (0.0029) & 0.8415 &  & (0.0019) & 0.7635 &  & (0.0012) & 0.7299 \\ 
   &  &  & 0.3308 &  &  & 0.4416 &  &  & 0.5026 &  &  & 0.5471 \\ 
 $\chi=0.1$   &  & 0.6177 & 0.6133 &  & 0.6168 & 0.6147 &  & 0.6150 & 0.6147 &  & 0.6224 & 0.6240 \\ 
   &  & (0.0050) & 0.9541 &  & (0.0029) & 0.8093 &  & (0.0018) & 0.7316 &  & (0.0012) & 0.6957 \\ 
   &  &  & 0.3667 &  &  & 0.4992 &  &  & 0.5538 &  &  & 0.5982 \\ 
  $\chi=0.25$ &   & 0.6610 & 0.6590 &  & 0.6628 & 0.6611 &  & 0.6631 & 0.6628 &  & 0.6717 & 0.6729 \\ 
   &   & (0.0051) & 0.9773 &  & (0.0029) & 0.8537 &  & (0.0019) & 0.7770 &  & (0.0012) & 0.7448 \\ 
   &   &  & 0.3632 &  &  & 0.4979 &  &  & 0.5625 &  &  & 0.6145 \\ 
  $\chi=0.5$ &   & 0.6699 & 0.6605 &  & 0.6720 & 0.6720 &  & 0.6749 & 0.6753 &  & 0.6862 & 0.6861 \\ 
   &   & (0.0052) & 1.0018 &  & (0.0030) & 0.8624 &  & (0.0019) & 0.7863 &  & (0.0012) & 0.7623 \\ 
   &   &  & 0.3702 &  &  & 0.5098 &  &  & 0.5751 &  &  & 0.6272 \\ 
  Continuous &   & 0.6773 & 0.6719 &  & 0.6840 & 0.6803 &  & 0.6894 & 0.6880 &  & 0.7040 & 0.7035 \\ 
   &   & (0.0052) & 1.0020 &  & (0.0030) & 0.8812 &  & (0.0019) & 0.8071 &  & (0.0012) & 0.7807 \\ 
  \hline
  \hline
& \multicolumn{4}{c}{} & & \multicolumn{2}{c}{$n=2^{14}$} & & \multicolumn{4}{c}{}\\
  \hline
 &   &  & 0.4577 &  &  & 0.5712 &  &  & 0.6102 &  &  & 0.6295 \\ 
  Sign &   & 0.6643 & 0.6614 &  & 0.6718 & 0.6723 &  & 0.6691 & 0.6697 &  & 0.6664 & 0.6670 \\ 
   &   & (0.0035) & 0.8934 &  & (0.0016) & 0.7691 &  & (0.0009) & 0.7242 &  & (0.0006) & 0.7010 \\ 
   &   &  & 0.4584 &  &  & 0.5537 &  &  & 0.5898 &  &  & 0.6043 \\ 
  $\chi=0.1$ &   & 0.6554 & 0.6519 &  & 0.6552 & 0.6577 &  & 0.6480 & 0.6486 &  & 0.6418 & 0.6425 \\ 
   &   & (0.0034) & 0.8706 &  & (0.0016) & 0.7571 &  & (0.0009) & 0.7069 &  & (0.0006) & 0.6801 \\ 
   &   &  & 0.4643 &  &  & 0.5745 &  &  & 0.6189 &  &  & 0.6380 \\ 
  $\chi=0.25$ &   & 0.6682 & 0.6674 &  & 0.6777 & 0.6799 &  & 0.6779 & 0.6791 &  & 0.6779 & 0.6781 \\ 
   &   & (0.0035) & 0.8999 &  & (0.0017) & 0.7819 &  & (0.0010) & 0.7371 &  & (0.0006) & 0.7169 \\ 
   &   &  & 0.4653 &  &  & 0.5806 &  &  & 0.6250 &  &  & 0.6502 \\ 
  $\chi=0.5$ &   & 0.6714 & 0.6634 &  & 0.6836 & 0.6849 &  & 0.6859 & 0.6861 &  & 0.6885 & 0.6891 \\ 
   &   & (0.0035) & 0.9043 &  & (0.0017) & 0.7870 &  & (0.0010) & 0.7442 &  & (0.0006) & 0.7268 \\ 
   &   &  & 0.4663 &  &  & 0.5854 &  &  & 0.6326 &  &  & 0.6616 \\ 
  Continuous &   & 0.6766 & 0.6756 &  & 0.6893 & 0.6908 &  & 0.6941 & 0.6954 &  & 0.7004 & 0.7015 \\ 
   &   & (0.0035) & 0.9047 &  & (0.0017) & 0.7917 &  & (0.0010) & 0.7546 &  & (0.0006) & 0.7400 \\ 
   \hline
\end{tabular}
  \caption{Estimation of the Hurst exponent of a fGn with $H=0.7$ and its discretization with the sign function and with $\chi=0.1$, $0.25$, and $0.5$. The time series has length $2^{10}$ (top) and $2^{14}$ (bottom). The estimation of the Hurst exponent is obtained with the Detrended Fluctuation Analysis by performing the least square regression over a $q$ fraction of the largest values of $\log_{10} [m]$ (see text). The table reports the mean value of the estimator $\widehat{H}$ over $10^3$ numerical simulations, together with the standard error (SE) and the $2.5\%$, $50\%$, and $97.5\%$ percentile.}
  \label{TableDFACoarse07}
 }
 \end{center}
 \end{table}
%


%
\begin{table}[ht]
\begin{center}
{\fontsize{10}{8} \selectfont
\begin{tabular}{lrcrrcrrcrrcr}
  \hline
\multirow{3}{*}{} & & \multicolumn{2}{c}{$q=0.25$} & & \multicolumn{2}{c}{$q=0.5$} & & \multicolumn{2}{c}{$q=0.75$} & & \multicolumn{2}{c}{$q=1$} \\
& & & $2.5\%$ & & & $2.5\%$ & & & $2.5\%$ & & & $2.5\%$\\
& &  Mean($\widehat{H}$) & $50\%$ & & Mean($\widehat{H}$) & $50\%$ & & Mean($\widehat{H}$) & $50\%$ & & Mean($\widehat{H}$) & $50\%$\\
& & (SE($\widehat{H}$)) & $97.5\%$ & & (SE($\widehat{H}$)) & $97.5\%$ & & (SE($\widehat{H}$)) & $97.5\%$ & & (SE($\widehat{H}$)) & $97.5\%$\\
  \hline
  \hline
& \multicolumn{4}{c}{} & & \multicolumn{2}{c}{$n=2^{10}$} & & \multicolumn{4}{c}{}\\
  \hline
 &   &  & 0.4475 &  &  & 0.5755 &  &  & 0.6327 &  &  & 0.6798 \\ 
  Sign &   & 0.7796 & 0.7647 &  & 0.7782 & 0.7795 &  & 0.7737 & 0.7735 &  & 0.7724 & 0.7727 \\ 
   &   & (0.0061) & 1.1787 &  & (0.0034) & 0.9876 &  & (0.0022) & 0.9045 &  & (0.0015) & 0.8574 \\ 
   &   &  & 0.4278 &  &  & 0.5232 &  &  & 0.5866 &  &  & 0.6187 \\ 
  $\chi=0.1$ &   & 0.7345 & 0.7205 &  & 0.7306 & 0.7288 &  & 0.7237 & 0.7224 &  & 0.7220 & 0.7227 \\ 
   &   & (0.0059) & 1.1161 &  & (0.0033) & 0.9401 &  & (0.0023) & 0.8700 &  & (0.0016) & 0.8298 \\ 
   &   &  & 0.4557 &  &  & 0.5888 &  &  & 0.6583 &  &  & 0.7038 \\ 
  $\chi=0.25$ &   & 0.7923 & 0.7749 &  & 0.7954 & 0.7979 &  & 0.7929 & 0.7944 &  & 0.7923 & 0.7931 \\ 
   &   & (0.0061) & 1.1848 &  & (0.0033) & 1.0102 &  & (0.0021) & 0.9231 &  & (0.0014) & 0.8764 \\ 
   &   &  & 0.4739 &  &  & 0.6140 &  &  & 0.6785 &  &  & 0.7253 \\ 
  $\chi=0.5$ &   & 0.8058 & 0.7939 &  & 0.8106 & 0.8114 &  & 0.8127 & 0.8147 &  & 0.8176 & 0.8200 \\ 
   &   & (0.0061) & 1.2039 &  & (0.0033) & 1.0145 &  & (0.0021) & 0.9344 &  & (0.0014) & 0.8996 \\ 
   &   &  & 0.4815 &  &  & 0.6262 &  &  & 0.7008 &  &  & 0.7522 \\ 
  Continuous &   & 0.8140 & 0.7936 &  & 0.8241 & 0.8248 &  & 0.8330 & 0.8346 &  & 0.8485 & 0.8502 \\ 
   &   & (0.0062) & 1.2178 &  & (0.0033) & 1.0272 &  & (0.0021) & 0.9580 &  & (0.0014) & 0.9341 \\ 
  \hline
  \hline
& \multicolumn{4}{c}{} & & \multicolumn{2}{c}{$n=2^{14}$} & & \multicolumn{4}{c}{}\\
  \hline
 &   &  & 0.5712 &  &  & 0.7131 &  &  & 0.7557 &  &  & 0.7640 \\ 
  Sign &   & 0.8231 & 0.8201 &  & 0.8227 & 0.8222 &  & 0.8150 & 0.8143 &  & 0.8028 & 0.8025 \\ 
   &   & (0.0039) & 1.0578 &  & (0.0017) & 0.9289 &  & (0.0010) & 0.8758 &  & (0.0006) & 0.8410 \\ 
   &   &  & 0.5760 &  &  & 0.6924 &  &  & 0.7219 &  &  & 0.7264 \\ 
  $\chi=0.1$ &   & 0.8071 & 0.8076 &  & 0.8018 & 0.8019 &  & 0.7874 & 0.7869 &  & 0.7695 & 0.7697 \\ 
   &   & (0.0039) & 1.0472 &  & (0.0018) & 0.9184 &  & (0.0010) & 0.8533 &  & (0.0007) & 0.8124 \\ 
   &   &  & 0.5814 &  &  & 0.7192 &  &  & 0.7657 &  &  & 0.7787 \\ 
  $\chi=0.25$ &   & 0.8271 & 0.8249 &  & 0.8304 & 0.8299 &  & 0.8272 & 0.8269 &  & 0.8185 & 0.8183 \\ 
   &   & (0.0039) & 1.0749 &  & (0.0018) & 0.9410 &  & (0.0010) & 0.8891 &  & (0.0006) & 0.8565 \\ 
   &   &  & 0.5871 &  &  & 0.7265 &  &  & 0.7756 &  &  & 0.7929 \\ 
  $\chi=0.5$ &   & 0.8296 & 0.8301 &  & 0.8350 & 0.8348 &  & 0.8357 & 0.8346 &  & 0.8333 & 0.8329 \\ 
   &   & (0.0039) & 1.0671 &  & (0.0018) & 0.9472 &  & (0.0010) & 0.8984 &  & (0.0006) & 0.8714 \\ 
   &   &  & 0.5849 &  &  & 0.7303 &  &  & 0.7814 &  &  & 0.8103 \\ 
  Continuous &   & 0.8316 & 0.8318 &  & 0.8388 & 0.8390 &  & 0.8436 & 0.8428 &  & 0.8499 & 0.8497 \\ 
   &   & (0.0039) & 1.0732 &  & (0.0018) & 0.9516 &  & (0.0010) & 0.9048 &  & (0.0006) & 0.8884 \\ 
\hline
  \end{tabular}
  \caption{Estimation of the Hurst exponent of a fGn with $H=0.85$ and its discretization with the sign function and with $\chi=0.1$, $0.25$, and $0.5$. The time series has length $2^{10}$ (top) and $2^{14}$ (bottom). The estimation of the Hurst exponent is obtained with the Detrended Fluctuation Analysis by performing the least square regression over a $q$ fraction of the largest values of $\log_{10} [m]$ (see text). The table reports the mean value of the estimator $\widehat{H}$ over $10^3$ numerical simulations, together with the standard error (SE) and the $2.5\%$, $50\%$, and $97.5\%$ percentile.}
  \label{TableDFACoarse085}
}
 \end{center}
\end{table}


\subsection{Discussion on the estimation of the Hurst exponent of a discretized process}\label{discussionH}

In conclusion our analytical considerations and numerical simulations show that both the local Whittle estimator and the DFA consistently give negatively biased estimates of the Hurst exponent $H$ when they are applied to discretized processes. These results are consistent with the literature on long-memory signal plus noise processes (Arteche (2004); Hurvich, Moulines and Soulier (2005)) and on general non-linear transformations of long-memory processes (Dalla, Giraitis and Hidalgo (2006)).
Our results show also that the size of the negative bias of $\widehat{H}$ for a discretized process can be significant because the second order correction of the spectral density or of the DFA fluctuation function is large for these processes. The negative bias can be partly overcome by taking small values of $m$ (for the LW estimator) or small intervals for the fit of the root-mean-square fluctuations (for the DFA). However, in both cases the variance of the estimators become large and thus the estimator is not significantly reliable.

In short samples, the LW estimator with $m=n^{0.8}$ and the DFA estimator with $q=1$ give similar results for $H=0.7$, while the former is slightly less biased than the latter for $H=0.85$. In large samples, instead, the DFA estimator with $q=1$ is less biased than the LW estimator with $m=n^{0.8}$, both for $H=0.7$ and $H=0.85$.  On the other hand, the LW estimator with $m=n^{0.5}$ performs better (in terms of both variance and bias) than the DFA estimator with $q=0.25$, both in small and in large samples, and for all values of $H$.  However, it is worth noticing that number of points used in the DFA regression\footnote{In our numerical studies, as in the original code of Peng {\it et al.} (1994), the x-coordinates of the DFA regression are arranged in a geometric series such that the ratio between consecutive box sizes is approximately $2^{1/8}$.} is  much smaller than the number of points used for the LW minimization. Therefore, it is not surprising that the variance of the LW estimate is systematically lower than that of the DFA estimate.


Despite the fact that we have presented results only for the fGn and fARIMA$(0, d, 0)$ processes, we have also run extensive Monte Carlo simulations for stationary fARIMA$(1,d, 0)$ and fARIMA$(0, d, 1)$ processes. The results are similar to those obtained for fGn and fARIMA$(0, d, 0)$, with a statistically significant underestimation of the Hurst exponent for the sign process and for coarse discretizations, both for the LW estimation and for the DFA method. However, for stationary fARIMA$(1,d, 0)$ and fARIMA$(0, d, 1)$ the size of the underestimation depends on the high-frequency behavior of the spectral density of the original process. 
For the sake of brevity, we do not report these results here, but they are available from the authors by request.


\section{Discussion and concluding remarks }
\label{conclusion}

\subsection{Relation to the measurement error literature}

As mentioned in the introduction, round-off error can be considered a form of measurement error, even if the type of measurement error typically considered in the literature is different, being modeled as a 
white noise uncorrelated with the latent process. 
On the other hand, the round-off error is a deterministic function of the latent process itself, and it is neither uncorrelated with the latent process nor a white noise. 

Recently, in the context of short-memory processes, Hansen and Lunde (2010) have studied the effect of sampling errors on the dynamic properties of an underlying ARMA$(p, q)$ time series. They have proved that the estimates of both the persistence parameter and the autocorrelation function are negatively biased by the measurement error. 
Our paper contributes in this direction for the case of long-memory processes subject to round-off errors. In fact, we are able to compute exactly the asymptotic scaling factor between the autocorrelation function of the realized process and that of the latent process. This scaling factor is a function only of the adimensional parameter $\chi$, which in turn depends on the grid size ($\delta$) and the variance of the underlying series ($D$). In most cases of interest the researcher knows the level of discretization. In principle, if the researcher knew also the variance of the underlying process, then she could estimate the autocorrelation function of the latent variable exactly for large lags. The same holds true for the spectral density at small frequencies and the DFA function for large box sizes. In general, however, the variance of the underlying process is not known {\it ex ante}, and the researcher needs to make inference on it. 
Even if the problem of the inference on $\chi$ (or $D$) is outside the scope of this paper, it is worth mentioning that by considering the  fraction $q_0$ of points of the discretized process with value zero it is possible to infer the value of the variance of the underlying process. In fact, we showed that $q_0=erf[\delta/2\sqrt{2D}]$ and therefore the knowledge of the grid size $\delta$ and the measurement of $q_0$ allow to infer $D$. Finally, note that other methods could be used to estimate $D$, such as, for example, by comparing the autocovariance of the discretized process with two different values of $\chi$. 


In the context of long memory processes, Hurvich, Moulines, and Soulier (2005) considered a process that can be decomposed into the sum of a long memory signal (possibly non-stationary) and a white noise, possibly contemporaneously correlated with the signal. They proved that the observed series has the same Hurst exponent as the underlying signal, and from numerical simulations observed that standard estimators of the Hurst exponent, such as the GPH estimator, suffer from a negative bias. Hurvich, Moulines, and Soulier (2005)  proposed a corrected local Whittle estimator, which is unbiased, but has larger asymptotic standard errors. 
Recently Rossi and Santucci de Magistris (2011) analyzed the effect of measurement errors on the estimation of the Hurst exponent. Specifically, they studied the effect of measurement errors on the estimation of the dynamic properties of the realized variance, as an {\it ex-post} estimator of the integrated variance. 
Rossi and Santucci de Magistris (2011) find that discrete frequency sampling and market microstructure noise induce a finite-sample bias in the fractionally integration semiparametric estimates. Based on Hurvich, Moulines, and Soulier (2005)  they propose an unbiased local Whittle estimator (with larger asymptotic standard errors) that accounts for the presence of the measurement error. However, the corrected Whittle estimator proposed by Hurvich, Moulines, and Soulier (2005) and Rossi and Santucci de Magistris (2011) does not apply to the round-off error, because the round-off error does not satisfy the standard assumptions on the measurement error. More recently, in a working paper, Corsi and Ren\`o (2011) proposed an indirect inference approach to estimate the long memory parameter of a latent integrated volatility series. 

Our paper contributes to this literature by providing the second order term of the spectral density and of the root-mean-square fluctuation function for the discretized process. A possible extension is to identify analytically and numerically the asymptotic region over which the bias on the long memory estimates due to these second order corrections becomes negligible. From Corollary~\ref{Cor:Whittle} we know the order in probability of the bias of the LW estimator, which can help identifying a suitable bandwidth $m$ for the minimization of the Whittle likelihood function (\ref{LW:ObjFun}) in terms of bias reduction. Similarly, one could use our results on the second order correction to the root-mean-square fluctuation to estimate a suitable threshold $q$ for the DFA regression (see Section~\ref{DFA}) in terms of bias reduction. This is outside the scope of the present paper and is left as future work.


\subsection{Concluding remarks}

We have presented an extensive analytical and numerical investigation of the properties of a continuously valued long memory process subject to round-off error.  We have shown that
the discretized process is also long memory with the same Hurst exponent as the latent process. We have explicitly computed the leading term of the asymptotic expansion of the autocovariance function, of the  spectral density for small frequencies, and of the root-mean-square fluctuation for large box sizes.  We have shown that the autocovariance, the spectral density and the root-mean-square fluctuation are asymptotically rescaled by a factor smaller than one, and we have computed exactly this scaling factor. More important in all three cases we have computed the order of magnitude the second order correction. This term is important to quantify the bias of the Hurst exponent estimators based on these quantities.  An in depth analysis of Hurst exponent estimators, namely the periodogram and the Detrended Fluctuation Analysis, shows that both the estimators are significantly negatively biased as a consequence of the order of magnitude of the second order corrections.

These results can be considered the starting point for several future research directions. A first interesting question is whether the strong negative bias in the estimation of the Hurst exponent observed for the discretization transformation is also observed for other types of (non-linear) transformation of the Gaussian process. Since one often obtains non-Gaussian long memory processes by transforming a fGn or a fARIMA process (Dittman and Granger (2002)), the question of the bias of the estimator of the Hurst exponent is of great interest. 
A second extension is to study the effect of round-off error on non-stationary processes. The discretization of a non-stationary process is, for example, a more realistic description of the effect of tick size on the price process, and therefore it will have more direct applications to the microstructure literature. The discretization of the price process induces round-off errors in the observed returns, and therefore also in the realized volatility. From a financial econometrics perspective, it would be interesting to consider also processes with stochastic volatility, i.e. processes where the increments of the process are uncorrelated, but the volatility is significantly correlated or is a long memory process. A very small scale example of this type of analysis has been performed in La Spada {\it et al.} (2011), where discretization of simulated price process with volatility described by an ARCH process has been considered, finding qualitative results similar to those presented in this paper.
Finally, an interesting topic for further research is the extension of our results to continuous time process, and then the study of the combined effect of round-off error and discrete time sampling. We believe that these extensions are potentially of interest in the literature on realized volatility measures, such as the realized variance, that are imperfect estimates of actual volatility. 

Since most natural and socio-economic time series are observed on a grid of values, i.e. the measured process is naturally discretized, we believe that our results are potentially of interest in many contexts. The severe bias of the long memory property of a discretized process should warn the scholars investigating processes on the underestimation of the Hurst exponent due to the non-linear transformation induced by round-off errors.

\section*{Acknowledgments} We are grateful to Yacine A\"it-Sahalia, Fulvio Corsi, Gabriele Di Cerbo, Ulrich K. M\"uller, and Christopher Sims for their useful comments and discussions. We wish to thank J. D. Farmer for useful discussions inspiring the beginning of this work.

\newpage


\appendixtitleon

\begin{appendices}

\section{Distributional properties}\label{App:DistrProp}

In this appendix we consider the distributional properties of the discretization of a generic stationary Gaussian process. 

From (\ref{eq:pmf}) the $m$-th moment of the discretized process can be written as $E[X_d^m]= \sum_{n=-\infty}^{\infty} q_n (n\delta)^m$.
%
Since $X$ is Gaussian distributed, the variance $D_d$ of the discretized process can be calculated explicitly.  From the above expression with $m=2$ we obtain
\begin{equation}
D_d=D\sum_{n=-\infty}^{\infty} \frac{n^2}{2\chi}\left[erf\left(\frac{2n+1}{2\sqrt{2\chi}}\right)-erf\left(\frac{2n-1}{2\sqrt{2\chi}}\right)\right]
\label{Dd}
\end{equation}
Left panel of Figure \ref{varFigureIncr} shows the ratio $D_d/D$ as a function of the scaling parameter $\chi$. It is worth noting that  this ratio is not monotonic. For small $\chi$ the ratio goes to zero because $\delta$ is very large relatively to $D$ and essentially all the probability mass falls in the bin centered at zero. In this regime the variance ratio goes to zero as 
\begin{equation}\nonumber
\frac{D_d}{D}\simeq 2\sqrt{\frac{2}{\pi}} \frac{e^{-1/8\chi}}{\sqrt{\chi}}.
\end{equation}
When $\chi\gg 1$ the ratio tends to one because the effect of discretization becomes irrelevant. In this regime $D_d/D \simeq  1+ 1/(12\chi)$.

Analogously it is possible to calculate the kurtosis $\kappa_d=E[X_d^4]/(E[X_d^2])^2$ of the discretized process. It is direct to show that the kurtosis is 
\begin{equation}
\kappa_d=\frac{\sum_{n=1}^\infty n^4 \left[erf\left(\frac{2n+1}{2\sqrt{2\chi}}\right)-erf\left(\frac{2n-1}{2\sqrt{2\chi}}\right)\right]}{\left(\sum_{n=1}^\infty n^2 \left[erf\left(\frac{2n+1}{2\sqrt{2\chi}}\right)-erf\left(\frac{2n-1}{2\sqrt{2\chi}}\right)\right]\right)^2}
\label{kurt}
\end{equation}
For small $\chi$ the kurtosis diverges as 
\begin{equation}\nonumber
\kappa_d\simeq \sqrt{\frac{\pi}{8\chi}} e^{1/8\chi}
\end{equation}
because the fourth moment goes to zero slower than the squared second moment. For large $\chi$ the kurtosis converges as expected to the Gaussian value 3 as $\kappa_d\simeq 3-1/(120 \, \chi^{2})$. Note that the kurtosis reaches its asymptotic value 3 from below, since it reaches a minimum of  roughly $2.982$ at $\chi \simeq 0.53$ and then converges to three from below.



%
\section{Proofs  for Section~\ref{nullmodel}}\label{App:Proof1}
\begin{proof}[{\bf Proof of Proposition~\ref{Prop:ACVGen}}]
The discretized process, $X_d(t)$, is a non-linear transformation of the underlying real-valued process, $X(t)$. Let us denote the discretization transformation with $g(\cdot)$, so that $X_d(t)=g(X(t))$. From (\ref{hexpansion}) we know that
\beq
\gamma_d(k)=\sum_{j=1}^{\infty} g_j^2 \rho^j(k) \qquad \forall \; k=0, 1, \ldots \nonumber
\eeq
 where $\rho$ is the autocorrelation function of the underlying continuous process, and $g_j$ are defined in (\ref{gCoeff}).

Since the function $g(x)$ is an odd function $g_j=0$ for $j$ even, while  all the odd coefficients are non-vanishing.  Therefore the discretization function has Hermite rank 1 and can be written as an infinite sum of Hermite (odd) polynomials. The generic $g_j$ coefficient is
\begin{equation}
g_j=\frac{1}{\sqrt{2\pi D}}\sum_{n=-\infty}^\infty n\delta \int_{(n-1/2)\delta}^{(n+1/2)\delta} H_j\left(\frac{x}{\sqrt{D}}\right) e^{-x^2/2D} dx = \sqrt{\frac{D}{2\pi}}\sum_{n=-\infty}^\infty \frac{n}{\sqrt{\chi}} \int_{(n-1/2)/\sqrt{\chi}}^{(n+1/2)/\sqrt{\chi}} H_j(x) e^{-x^2/2} dx\nonumber
\end{equation} 

The first Hermite polynomial is $H_1(x)=x$ and the coefficient $g_1$ is 
\begin{equation}\label{eq:g1}
\frac{g_1}{\sqrt{D}}=\sqrt{\frac{2}{\pi \chi}} \sum_{k=0}^\infty \exp\left[-\frac{(2k+1)^2}{8\chi}\right]=\frac{1}{\sqrt{2\pi \chi}}\vartheta_2(0,e^{-1/2\chi})
\end{equation}
where $\vartheta_a(u,q)$ is the elliptic theta function. For large $\chi$, $g_1\simeq \sqrt{D}$, while for small $\chi$ the coefficient $g_1$ goes to zero as 
\begin{equation}
g_1\simeq \sqrt{D}\sqrt{\frac{2}{\pi}}\frac{e^{-1/8\chi}}{\sqrt{\chi}}\nonumber
\end{equation}

The second non-vanishing coefficient $g_3$ is given by
\begin{eqnarray}\label{eq:g3}
\frac{g_3}{\sqrt{D}}=-\sqrt{\frac{1}{3\pi \chi}} \sum_{k=0}^\infty \exp\left[-\frac{(2k+1)^2}{8\chi}\right]+\frac{1}{\sqrt{48\pi \chi^3}} \sum_{k=0}^\infty (2k+1)^2 \exp\left[-\frac{(2k+1)^2}{8\chi}\right]
\end{eqnarray}




In principle one could calculate all the coefficients $g_j$. Here we want to focus on the case when the correlation coefficient $\rho$ is small (i.e. $k$ is large), so that it suffices to consider the first two coefficients $g_1$ and $g_3$. 
%
Therefore, from (\ref{hexpansion}) we have
\beqnr
\gamma_d(k)
&=& \frac{g_1^2}{D} \gamma(k) \left(1 + O\left(\gamma(k)^2\right)\right) \qquad \text{as}\; k\to \infty\,.
\label{hexpansionProof_21}
\eeqnr
If we plug (\ref{ACVgeneralDef}) and (\ref{eq:g1}) into~(\ref{hexpansionProof_21}), we get the result.
\end{proof}
\begin{proof}[{\bf Proof of Proposition~\ref{Prop:ACV}}]
Under the assumption that $L \in \mathcal{L}$, we can write  $\sum_{i=0}^{\infty} b_i k^{-\beta_i} = b_0\, \left(1+ O\left(k^{-\beta_1}\right)\right)$; by substituting this expression into (\ref{eq:ACVGen}) we get the result.
\end{proof}
\begin{proof}[{\bf Proof of Corollary~\ref{CorACF}}]
It follows from the definition of autocorrelation function and Proposition~\ref{Prop:ACVGen}.
\end{proof}
\begin{proof}[{\bf Proof of Proposition \ref{Prop:GenSpecDens}}]
From Proposition~\ref{Prop:ACVGen} we can write
%
$\gamma_d(k) \sim \left(\frac{ \vartheta_2(0,e^{-1/2\chi}) }{\sqrt{2 \pi \chi}}\right)^2  k^{2H-2} L(k)$, as $k \to \infty\,.$ 
%
Then, the proposition follows from Theorem 3.3 (a) in Palma (2007).

\end{proof}
\begin{proof}[{\bf Proof of Proposition~\ref{Prop:SpecDens}}]
For the sake of simplicity, we consider the case of an underlying process with unit variance, namely $D=1$. In order to extend the proof to non unit variance we need simply to do the transformation $L(k) \to L(k)/D$. Obviously, $L(k)/D$ is still slowly varying at infinity. 
Moreover, we consider only the case $I=\infty$, the case $I<\infty$ being trivial.

This proof is divided in two parts: in the first part we prove (\ref{eq:SpecDensRound1stOrd}) and (\ref{eq:SpecDensRound2ndOrd_1}) under the assumption that $L\in \mathcal{L}$; in the second part we prove (\ref{eq:SpecDensRound2ndOrd_2}) under the additional Assumption~\ref{Ass:AbsSum}.

\paragraph{First part} 
From Proposition~\ref{Prop:ACV} and Theorem 2.1 in Beran (1994) we know that the spectral density $\phi_d(\omega)$ exists. 
%
%
%
Moreover, if $L \in \mathcal{L}$, then $\exists K > 0$ s.t. $\forall \, k\geq K,\; L(k)=\sum_{i=0}^{\infty} b_i k^{-\beta_i}$ and the series converges absolutely.
Then, from Herglotz's theorem and (\ref{hexpansion}) we can write
\beqnr
\phi_d(\omega)= \frac{D_d}{2 \pi} + \frac{1}{\pi} \sum_{k=1}^{K} \gamma_d(k) \cos(k \omega) 
 + \frac{1}{\pi} \sum_{k > K}^{\infty} \sum_{j=0}^{\infty} g_{2j+1}^2 \left(k^{2H-2}\sum_{i=0}^{\infty} b_i k^{-\beta_i} \right)^{2j+1}\cos(k \omega)\nonumber
\eeqnr
As $\omega \to 0^+$, the first two terms on the RHS converge to a constant, while the third term diverges.

Since $\sum_i b_i k^{-\beta_i}$ converges absolutely $\forall \, k \geq K$, we can write  
\beq
\sum_{i=0}^{\infty} b_i k^{-\beta_i} = b_0 + b_1 k^{-\beta_i} + b_2 k^{-\beta_2} (1+o(1))\nonumber
\eeq
Therefore,
\beqnr
&& \sum_{k > K}^{\infty} \sum_{j=0}^{\infty} g_{2j+1}^2 \left(k^{2H-2}\sum_{i=0}^{\infty} b_i k^{-\beta_i} \right)^{2j+1} \cos(k \omega) = \sum_{k>K}^{\infty} \left\{g_1^2 k^{2H-2} \left[b_0+b_1k^{-\beta_1} + b_2 k^{-\beta_2} (1+o(1))\right] + \right.\nonumber\\
&& \left. + g_3^2 k^{6H-6} \left[b_0^3+3 b_0^2 b_1 k^{-\beta_1}(1+o(1))\right] + g_5^2 k^{10H-10} b_0^5 (1+o(1))\right\} \cos(k\omega)\nonumber\\
&=&g_1^2 b_0 \sum_{k=1}^{\infty} k^{2H-2} \cos(k \omega) + g_1^2 b_1\sum_{k=1}^{\infty} k^{2H-2-\beta_1} \cos(k \omega) + O\left(\sum_{k=1}^{\infty} k^{2H-2-\beta_2} \cos(k \omega)\right) + \nonumber \\
&& + g_3^2 b_0^3 \sum_{k=1}^{\infty} k^{6H-6} \cos(k \omega) + O\left(\sum_{k=1}^{\infty} k^{6H-6-\beta_1} \cos(k\omega)\right) + O\left(\sum_{k=1}^{\infty} k^{10H-10} \cos(k \omega)\right) + O(1)\,,\label{phiProof_5}
\eeqnr
where the term $O(1)$ comes from substituting $\sum_{k>K}^{\infty}$ with $\sum_{k=1}^{\infty}$.

Then, we introduce the following representation of a trigonometric series
\beq
\sum_{k =1}^{\infty}  k^{-s} \cos(k \omega) = \frac{1}{2}\left(Li_{s}(e^{-i \omega})+Li_{s}(e^{i \omega})\right)\label{eq:PolyLogRep}
\eeq
where  $Li_{s}(z) = \sum_{k=1}^{\infty} z^k/k^s$ is the polylogarithm. 
%
%
The series expansion of $Li_{s}(e^z)$ for small $z$ is (see Gradshteyn and  Ryzhik (1980))
\beqnr
Li_s(e^z) &=& \Gamma(1 - s) \,(-z)^{s-1} + \sum_{k=0}^\infty {\zeta(s-k) \over k!} \,z^k\,, \qquad if \;\; s \notin \mathbb{N}\label{eq:ExpansionPolyLog1}\\
Li_s(e^z) &=& {z^{s-1} \over (s - 1)!} \left[H_{s-1} - \ln(-z) \right] + \sum_{k=0,\,k\ne s-1}^\infty {\zeta(s-k) \over k!} \, z^k\,, \qquad if \;\; s \in \mathbb{N}\label{eq:ExpansionPolyLog2}
\eeqnr
where $\zeta(\cdot)$ is the analytic continuation of the Riemann zeta function over the complex plane and $H_s$ is the $s$th harmonic number.
Finally, we plug (\ref{eq:ExpansionPolyLog1}) or (\ref{eq:ExpansionPolyLog2}) into (\ref{eq:PolyLogRep}), and then we apply (\ref{eq:PolyLogRep}) into (\ref{phiProof_5}). By substituting (\ref{eq:g1}) for $g_1^2$ and after some algebraic manipulation we get the result.

\paragraph{Second part} Under Assumption~\ref{Ass:AbsSum} (i) the series $L(k)=\sum_{i=0}^{\infty} b_i k^{-\beta_i}$ converges absolutely  $\forall \, k\geq 1$. Therefore, by Mertens' theorem\footnote{See G. H. Hardy, \textit{Divergent Series}, 2nd ed., (1991), Chapter X Theorem 160.}, for any $j\geq 0$  we can write
\beqnr
\left(\sum_{i=0}^{\infty} b_i k^{-\beta_i}\right)^{2j+1} = \sum_{i=0}^{\infty} \widetilde{b}_{j, i} k^{-\widetilde{\beta}_{j, i}} \qquad \forall\, k \geq 1\,,\nonumber
\eeqnr
where the series on the RHS is the Cauchy product.  
Note that $\widetilde{b}_{0, i}= b_i$ and $\widetilde{\beta}_{0, i}=\beta_i$ $\forall \, i$,  while $\widetilde{b}_{j, 0}=b_0^{2j+1}$ and $\widetilde{\beta}_{j, 0}=0$  $\forall \, j$. By absolute convergence of the original series the Cauchy product also converges absolutely $\forall \, k \geq 1$. 

From Herglotz's theorem and Proposition~\ref{Prop:ACV} we can write
\beqnr
\phi_d(\omega)
&=& \frac{D_d}{2\pi} + \frac{1}{\pi} \sum_{k=1}^{\infty} \sum_{j=0}^{\infty} \sum_{i=0}^{\infty} g_{2j+1}^2 \widetilde{b}_{j, i} k^{-\alpha_{j, i}} \cos(k \omega)\,, \nonumber
\eeqnr
where $\alpha_{j , i} \equiv (2j+1)(2-2H)+\widetilde{\beta}_{j, i}$.

Since $\left\{g^2_{2j+1}\right\}$ are bounded above (because $\sum_{j} g_{2j+1}^2=D_d<\infty$), Assumption~\ref{Ass:AbsSum} (i) implies that the double series $\sum_{j} \sum_{i} g_{2j+1}^2 \widetilde{b}_{j, i} k^{-\alpha_{j, i}}$ converges absolutely $\forall\, k\geq1$. Moreover, under Assumption~\ref{Ass:AbsSum} $(ii)$ there is only a finite number of terms in $L(k)$ (and therefore in $\sum_{j} \sum_{i} g_{2j+1}^2 \widetilde{b}_{j, i} k^{-\alpha_{j, i}}$) that are not summable over $k$. Thus, from Rudin (1976) (Chapter 8, Theorem 8.3) we can invert the order of summation and write

\beqnr
\phi_d(\omega) &=& \frac{D_d}{2\pi} + \frac{1}{\pi}  \sum_{j=0}^{\infty} g_{2j+1}^2 \sum_{i=0}^{\infty} \widetilde{b}_{j, i} \sum_{k=1}^{\infty} k^{-\alpha_{j, i}} \cos(k \omega) \nonumber
\eeqnr
In other words, the Fourier transform of the series becomes the series of the Fourier transforms. From this point on the proof is very similar to that of Lemma~\ref{Lemma:SpecDensGaussian} and therefore omitted for brevity.
\end{proof}
\begin{proof}[{\bf Proof of Corollary~\ref{Cor:SpecDensAnalytic}}]
If $L$ is analytic at infinity, then $L(k) = \sum_{n=0}^{\infty} b_n k^{-n}$ for large $k$  for some $\{b_n\}$, and the series converges absolutely within its radius of convergence. Then, $L \in \mathcal{L}$ with $\beta_1 \in \mathbb{N}$, and by applying Proposition~\ref{Prop:SpecDens} and (\ref{eq:g1}) we get (\ref{eq:SpecDensAnalytic}). It is straightforward to see that $c_2>0$ and  $c_1>0$ for $H>5/6$.
\end{proof}
\begin{proof}[{\bf Proof of Corollary~\ref{Cor:SpecDensRoundfGn}}]
For a fGn $L(k)=\frac{D}{2}\left(\left(1+\frac{1}{k}\right)^{2H}+\left(1-\frac{1}{k}\right)^{2H}-2\right)k^2$, which is an even analytic function at infinity and whose power series converges absolutely $\forall \, k\geq 1$. Because $L(k)$ is analytic and even, $\beta_i = 2 i$  $\forall \, i \geq 0$ and we can write 
\beq
\Gamma\left(2H-1-\beta_i\right)=(-1)^{2i-1}\frac{\Gamma\left(1-2H\right)\Gamma\left(2H\right)}{\Gamma\left(2(i+1-H)\right)}\longrightarrow 0 \qquad \text{as} \; i \to \infty\,.\nonumber
\eeq
Thus, the autocovariance of a fGn satisfies Assumption~\ref{Ass:AbsSum}, and therefore from Corollary~\ref{Cor:SpecDensAnalytic} it follows that the spectral density of a discretized fGn satisfies (\ref{eq:SpecDensAnalytic}) with $c_0$ given by (\ref{eq:SpecDensRound2ndOrd_2}).

From Corollary~\ref{Cor:SpecDensAnalytic} we already know that the second-order term is strictly positive if $H\geq 5/6$. Hence, we just need to prove that $c_0>0$ if $H<5/6$. Since $c_0$ is given by (\ref{eq:SpecDensRound2ndOrd_2}) we can write
\beqnr
c_0=\frac{D_d}{2\pi} +  \frac{g_1^2}{\pi} \sum_{i=0}^{\infty} \binom{2H}{2(i+1)}\, \zeta\left(2(i+1-H)\right) + \sum_{j=1}^{\infty}\sum_{i=0}^{\infty} \frac{g_{2j+1}^2}{\pi D^{2j+1}}\, \widetilde{b}_{j, i}\, \zeta\left((2j+1) (2-2H)+\widetilde{\beta}_{j, i}\right)\label{eq:cZero1}
\eeqnr
where we used the fact that for a fGn $b_i = \binom{2H}{2(i+1)} D$ and $\beta_i = 2 i$ $\forall \, i \geq 0$. The symbol $\binom{\cdot}{\cdot}$ denotes the generalized binomial coefficient.

First, since $\{b_i\}$ are strictly positive and $(2j+1)(2-2H)>1 \; \forall \, j\geq 1$ if $H<5/6$, the third term on the RHS of (\ref{eq:cZero1}) is strictly positive for $H<5/6$.

Second, from Sinai (1976) and Beran (1994) we know that the spectral density of a fGn satisfies
\beq
\phi(\omega) = 2 c^*_{\phi}  \left(1-\cos \omega \right) \sum_{j=-\infty}^{\infty} |2\pi j+ \omega|^{-2H-1} \,, \quad \omega \in [-\pi, \pi]\,.\nonumber
\eeq
where $c_{\phi}^*=\frac{D}{2\pi} \sin \left(\pi H \right) \Gamma \left(2H +1\right)$. 
For small $\omega$ the behavior of the above spectral density follows by Taylor expansion at zero:
\beq
\phi(\omega) = c_{\phi}^* |\omega|^{1-2H} - \frac{1}{12} c_{\phi}^* |\omega|^{3-2H} + o\left(|\omega|^{3-2H}\right)\,.\label{eq:fGnAsympBeran}
\eeq
%

On the other hand, because $L(k)$ is analytic with $\beta_1=2$, it satisfies Assumption~\ref{Ass:CoeffGenDef2}; therefore, the fGn satisfies the conditions of Lemma~\ref{Lemma:SpecDensGaussian}. Following the proof of that lemma, after some algebraic manipulations,  we get
\beq
\phi(\omega) = c_{\phi}^* |\omega|^{1-2H} + c_0^* -\frac{1}{12} c_{\phi}^* |\omega|^{3-2H} + o\left(|\omega|^{3-2H}\right)\,,\label{eq:fGnAsympLemma}
\eeq
where $c_0^*= D \pi^{-1}\left(\frac{1}{2} + \sum_{i=0}^{\infty} \binom{2H}{2(i+1)} \zeta\left(2(i+1-H)\right)\right)$.
By comparing (\ref{eq:fGnAsympBeran}) and (\ref{eq:fGnAsympLemma}) it follows that $c_0^*$ must be equal to zero, and therefore
\beq
\sum_{i=0}^{\infty} \binom{2H}{2(i+1)} \zeta\left(2(i+1-H)\right)=-\frac{1}{2}\,.\label{eq:SumBinomZeta}
\eeq

Finally, note that from (\ref{hexpansion}) we know that $D_d=\sum_{j=0}^{\infty} g_{2j+1}^2$. Hence, by plugging (\ref{eq:SumBinomZeta}) into (\ref{eq:cZero1}) we obtain
\beqnr
c_0&=&\frac{\sum_{j=0}^{\infty} g^2_{2j+1}}{2\pi} -  \frac{g_1^2}{2 \pi} +  \sum_{j=1}^{\infty} \sum_{i=0}^{\infty} \frac{g_{2j+1}^2}{\pi D^{2j+1}}\, \widetilde{b}_{j, i}\, \zeta\left((2j+1) (2-2H)+\widetilde{\beta}_{j, i}\right)\nonumber\\
&=& \frac{\sum_{j=1}^{\infty}g_{2j+1}^2}{2 \pi} +  \sum_{j=1}^{\infty} \sum_{i=0}^{\infty} \frac{g_{2j+1}^2}{\pi D^{2j+1}}\, \widetilde{b}_{j, i}\, \zeta\left((2j+1) (2-2H)+\widetilde{\beta}_{j, i}\right)\nonumber > 0
\eeqnr
\end{proof}

%
%

%
%

%
%


\section{Proofs for Section~\ref{estimation}}\label{App:Proof2}

\begin{proof}[{\bf Proof of Lemma~\ref{Lemma:SpecDensGaussian}}]
From Theorem 2.1 in Beran (1994) we know that the spectral density $\phi(\omega)$ exists and from Hergotz's theorem it is the discrete Fourier transform of the autocovariance $\gamma(k)=k^{2H-2} \sum_i b_i k^{-\beta_i}$.

Let $\alpha_i=2-2H+\beta_i \;\forall\,i\geq0$. Under Assumption~\ref{Ass:AbsSum} (ii) there is only a finite number of terms in the autocovariance of $X$ that are not summable, and therefore there will be only a finite number of divergent terms in the spectral density.
Moreover, under Assumption~\ref{Ass:AbsSum} (i) the series $\sum_{i=0}^{I} b_i k^{-\beta_i}$ converges absolutely $\forall\,k\geq 1$. Therefore, by Rudin (1976) (Theorem 8.3) we can write
\beqnr
\phi(\omega) 
&=& \frac{D}{2\pi} + \frac{1}{\pi} \sum_{i=0}^{I} b_i \sum_{k=1}^{\infty} k^{-\alpha_i} \cos(k \omega) \label{eq:Series1}
\eeqnr

By using the polylogarithm representation (\ref{eq:PolyLogRep}) introduced above, for small $\omega$ we can plug~(\ref{eq:ExpansionPolyLog1}) and (\ref{eq:ExpansionPolyLog2}) into (\ref{eq:Series1}). Let $I_1=\{ i \in \mathbb{N}: \alpha_i \notin \mathbb{N}\}$ and $I_2=\{ i \in \mathbb{N}: \alpha_i \in \mathbb{N}\}$. 
Then, under Assumption~\ref{Ass:AbsSum} (iii) we can rearrange the double series in (\ref{eq:Series1}) in the following way
\beqnr\label{eq:Series2}
\phi(\omega) = \frac{D}{2\pi} &+& \frac{1}{\pi} \, b_0 \left(\Gamma(2H-1) \sin\left((1-H)\pi\right) |\omega|^{1-2H} + \sum_{j=0}^{\infty} (-1)^j  \frac{\zeta(2-2H-2 j)}{ (2j)!} \omega^{2j}\right) \nonumber\\
&+&  \frac{1}{\pi} \sum_{i \in I_1} b_i \left(\Gamma(1-\alpha_i) \sin\left(\pi \alpha_i \over{2}\right) |\omega|^{\alpha_i-1} + \sum_{j=0}^\infty  (-1)^j {\zeta(\alpha_i-2j) \over (2j)!} \omega^{2j}\right)\nonumber\\
&+&  \frac{1}{\pi} \sum_{i \in I_2} b_i \left(\frac{|\omega|^{\alpha_i-1}}{(\alpha_i - 1)!}  \left[\sin\left(\pi \alpha_i \over{2}\right) H_{\alpha_i-1} + \frac{\pi}{2} \cos\left(\pi\alpha_i \over{2}\right) - \sin\left(\pi \alpha_i \over{2}\right) \ln|\omega| \right] + \right.\nonumber\\
&& \left. \qquad + \sum_{j=0,\,j\ne \alpha_i-1}^\infty (-1)^j {\zeta(\alpha_i-2j) \over (2j)!} \omega^{2j}\right)\,, \qquad as \;\; \omega\to 0^+\,,
\eeqnr
where $\zeta(\cdot)$ is the analytic continuation of the Riemann zeta function over the complex plane and $H_s$ is the $s$th harmonic number.

Under Assumption~\ref{Ass:AbsSum} we can collect all the terms of the same oder and rearrange (\ref{eq:Series2}) in powers of $\omega$. 
Let $c_0^*$ be the term of order $O(1)$ in (\ref{eq:Series2}); then, 
\beqnr\label{eq:TermOrder1}
c_0^*=\frac{D}{2\pi} + \frac{1}{\pi} \sum_{i: \,\alpha_i \neq 1}^I b_i \zeta\left(\alpha_i\right)\,.\nonumber
\eeqnr
%
%
Under Assumption~\ref{Ass:CoeffGenDef2} $\alpha_1 \neq 1$, and therefore, if also $\alpha_1\neq 2$, we can write 
\beqnr\label{eq:Case1}
\phi(\omega) = c_{\phi} b_0 |\omega|^{1-2H} + c_0^* + c^*_1 |\omega|^{1-2H+\beta_1} + o(|\omega|^{\min\left(0, 1-2H+\beta_1\right)})  \qquad \text{as} \;\; \omega\to 0^+\,,
\eeqnr
where $c_{\phi}$ is defined as in Proposition~\ref{Prop:GenSpecDens} and $c^*_1= \pi^{-1}b_1\Gamma(2H-1-\beta_1) \sin\left(\frac{2H-\beta_1}{2}\pi\right)$.

If $\alpha_1 = 2$, by Assumption~\ref{Ass:CoeffGenDef2} $c_0^*\neq 0$ and we can write
\beqnr\label{eq:Case2}
\phi(\omega) = c_{\phi} b_0 |\omega|^{1-2H} + c^*_0 + o(1) \,, \quad as \;\; \omega\to 0^+.
\eeqnr

Putting together (\ref{eq:Case1}) and (\ref{eq:Case2}), and noting that if $c_0^*=0$ in (\ref{eq:Case1}), then by Assumption~\ref{Ass:CoeffGenDef2} $\beta_1 \leq 2$, we get the result
\beqnr
\phi(\omega) = c_{\phi}  |\omega|^{1-2H} \left(b_0 + c_\beta |\omega|^{\beta} + o(|\omega|^{\beta})\right) \,, \quad as \;\; \omega\to 0^+\,,\nonumber
\eeqnr
where $c_{\beta}\neq 0$ and $\beta \in (0, 2]$.
\end{proof}

\begin{proof}[{\bf Proof of Theorem~\ref{thm:Whittle}}]
Following the proof of Theorem 4 in DGH, because $j_0=1$ we can write $Y_t$ as a signal plus noise process $Y_t=W_t + Z_t$, where
\beq
W_t = g_1 H_1\left(X_t\right) = g_1 X_t \qquad Z_t = \sum_{j=j_1}^{\infty} g_j H_j\left(X_t\right)\nonumber
\eeq
where $j_1$ is the second non-vanishing term in the Hermite expansion and $H_j(\cdot)$ is the $j$th Hermite polynomial.

\paragraph{Part (i)} If $L \in \mathcal{L}$, the spectral density of $W_t$ satisfies Assumption A in DGH, i.e.,
\beq
\phi_w(\omega)= c_w |\omega|^{1-2H} \left(1+o(1)\right) \qquad as\;\; \omega \to 0^+\,,\nonumber
\eeq
where $c_w=(g_1^2/D) c_{\phi}  b_0$. Moreover, since $X_t$ is a stationary purely non-deterministic Gaussian process, it is also linear with finite fourth moments. Consequently, $W_t = g_1X_t$ is also linear with finite fourth moments. Then, we can write
\beq
W_t = \sum_{j=0}^{\infty} a_j \varepsilon_{t-j}\nonumber
\eeq
where $\sum_{j=0}^{\infty} a_j^2 < \infty$ and $\varepsilon_t$ are i.i.d. Gaussian variables with zero mean and unit variance. Let $\alpha(\omega) = \sum_{j=0}^{\infty} a_j e^{i j \omega}$.
%
%
From Proposition 4 in DGH it follows that $W_t$ satisfies also Assumption B therein.

We show below that the spectral density of $Z_t$ satisfies $\phi_{z}(\omega)\leq C |\omega|^{1-2H_z}$, as $\omega \to 0^+$, for some $C>0$ and $H_z \geq 0.5$ such that
\beqnr\label{eq:d_zCases}
H>H_z = 
\begin{cases}
0.5 & \; if \;\; j_1(2-2H)>1\\
0.5-j_1 (1-H) & \; if \;\; j_1(2-2H)<1\\
\varepsilon & \; if \;\; j_1(2-2H)=1\\
\end{cases}
\eeqnr
for any $\varepsilon  \in (0.5, H)$.

Indeed, if $j_1(2-2H)>1$ from (\ref{hexpansion})
\beq
\sum_{k=1}^{\infty} \gamma_z(k) = \sum_{k=1}^{\infty} \sum_{j=j_1}^{\infty} g_j^2 \left(\rho(k)\right)^{j} \leq C \sum_{k=1}^{\infty} k^{-j_1 (2-2H)} <\infty \nonumber
\eeq
for some $C > 0$. Therefore, $H_z=0.5 < H$.

If $j_1(2-2H) < 1$, we can prove that 
\beq
\phi_{z}(\omega) = \frac{g_{j_1}^2}{D^{j_1}} b_0 c_{\phi} |\omega|^{1-2 H_z} (1 + o(1)) \leq C |\omega|^{1-2 H_z} \qquad as \; \omega \to 0^+\,,\nonumber
\eeq
for some $C>0$ and $H_z = H-(j_1-1) (1-H) < H \in (0.5, 1)$. Similarly, if $j_1(2-2H) = 1$, we can prove that
\beq
\phi_{z}(\omega) = C \ln |\omega|^{-1} (1 + o(1)) \leq C |\omega|^{-\varepsilon} \qquad as \; \omega \to 0^+\,,\nonumber
\eeq
for some $C>0$ and for any $\varepsilon>0$. The proof of the above results is a special case of the proof of Proposition~\ref{Prop:SpecDens}, and thus omitted. The results above prove (\ref{eq:d_zCases}).

Since $W_t$ satisfies Assumptions A and B in DGH and the spectral density $\phi_{z}$ satisfies the asymptotic conditions above, consistency of $\widehat{H}^Y_{LW}$ follows from Theorem 3 (i) in DGH. 

Moreover, if we write the periodogram of $Y_t$ as $I_{Y}(\omega_j)=I_{W}(\omega_j) + v_j$, where $I_{W}$ is the periodogram of the ``signal'' $W_t$ and $v_j$ is the contribution of the ``noise'' $Z_t$ at the $j$th Fourier frequency, it is straightforward to show (see DGH pg. 225--226) that
\beqnr
\widehat{H}^Y_{LW}-H&=&\left(\widehat{H}^X_{LW}-H\right)(1+o_P(1)) - \left(m^{-1}\sum_{j=1}^m \left(\log\left(\frac{j}{m}\right)+1\right) \frac{ v_j}{c_w\omega_{j}^{1-2H}}\right)(1+o_P(1))+O_P\left(\frac{\log m}{m}\right)\nonumber\\
&=& \left(\widehat{H}^X_{LW}-H\right)(1+o_P(1)) + O_P\left(\left(\frac{m}{n}\right)^{H-H_z} + \frac{\log m}{m}\right)\label{eq:LWAsympExpan}\,,
\eeqnr
where $\widehat{H}^X_{LW}$ denotes the LW estimator of $\{X_t\}$ if the sequence $\{X_t\}$ were observed.
 
 Note that, roughly speaking, $v_j$ represents the sample estimate of the higher-order terms of the spectral density of $Y_t$ at the $j$th Fourier frequency. For the discretization of a fGn we know from Corollary~\ref{Cor:SpecDensRoundfGn} that  the second-order term of the spectral density is strictly positive for all $H$; therefore, in that case, we expect that the second term on the RHS of the first line of (\ref{eq:LWAsympExpan}) will induce a negative finite sample bias on $\widehat{H}^Y_{LW}$. The order of magnitude of this finite sample bias is $O_{P}\left(\left(m/n\right)^{H-H_z}\right)$.

\paragraph{Part (ii)} Under Assumptions~\ref{Ass:AbsSum} and \ref{Ass:CoeffGenDef2}, from Lemma~\ref{Lemma:SpecDensGaussian} it follows that $X_t$ and therefore $W_t$ satisfy Assumption $T(\alpha_0, \beta)$ in DGH, with $\alpha_0=2H-1$ and $\beta$ defined as in Lemma~\ref{Lemma:SpecDensGaussian}. Moreover, under Assumption~\ref{Ass:SmoothSpecAlpha} we can combine the second part of Proposition 5 in DGH with Proposition 3 and Theorem 2 therein, and under the assumption that $m=o\left(m^{2\beta/(2\beta+1)}\right)$ we can write
\beq
\widehat{H}^X_{LW} - H =  - \left(\frac{m}{n}\right)^{\beta} \left(\frac{c_{\beta}}{b_{0}}\right) \frac{B_{\beta}}{2} -\frac{V_m}{2} (1+o_P(1)) + o_{p}\left(m^{-1/2} + \left(\frac{m}{n}\right)^{\beta}\right)\label{eq:LWLinearAsympExpan}
\eeq
where $c_{\beta}$ is defined as in Lemma~\ref{Lemma:SpecDensGaussian}, $B_{\beta} = (2\pi)^\beta \beta/(\beta+1)^2$, and 
\beq
V_m = m^{-1} \sum_{j=1}^m\left(\log \left(\frac{j}{m}\right)+1\right) \left(\eta_j-\mathbb{E}\eta_j\right)\nonumber
\eeq
with $\eta_j=I_X(\omega_j)/\phi_x(\omega_j)$. 

Let $r = H-H_z$. By plugging (\ref{eq:LWLinearAsympExpan}) into (\ref{eq:LWAsympExpan}) we obtain
\beq
\widehat{H}^Y_{LW}-H =  -\frac{V_m}{2}  - \left(\frac{m}{n}\right)^{\beta} \left(\frac{c_{\beta}}{b_{0}}\right) \frac{B_{\beta}}{2} + O_P\left(\left(\frac{m}{n}\right)^{r} + \frac{\log m}{m}\right) + o_{p}\left(m^{-1/2} + \left(\frac{m}{n}\right)^{\beta} + V_m\right)\,.\label{eq:LWAsympExpan2}
\eeq
Moreover, under Assumption~\ref{Ass:SmoothSpecAlpha} and $m=o\left(m^{2\beta/(2\beta+1)}\right)$, by Robinson's (1995) Theorem 2
\beq\label{eq:RobinsonThm2}
\sqrt{m} V_m \overset{d}{\rightarrow} N(0, 1)\,, \qquad as \;\; n \to \infty\,. 
\eeq
Therefore, $V_m = O_P\left(m^{-1/2}\right)$ and from (\ref{eq:LWAsympExpan2}) follows (\ref{eq:WhittleConvRate}).

\paragraph{Part (iii)} If $m=o(n^{2r/(2r+1)})$, equation (\ref{eq:WhittleAsympNorm}) follows from applying (\ref{eq:RobinsonThm2}) in (\ref{eq:LWAsympExpan2}).
\end{proof}

\begin{proof}[{\bf Proof of Corollary~\ref{Cor:Whittle}}]
The result of the corollary follows directly from Theorem~\ref{thm:Whittle} and from noticing that the second non-vaninshing Hermite coefficient for the discretized process is $g_3\neq 0$, so that $j_1=3$. 
\end{proof}

For the proof of Theorem~\ref{ThmExpF} we need the following lemma. Note that the proofs of the lemmas are at the end of this Appendix. 

%
%
\begin{lemma} \label{Lemma:remainderEM}
Let $m\in \mathbb{N}$, $\alpha<1$, and 
\beqnr
R(\alpha) &=& - \alpha(\alpha-1) \int_{1}^{m} \frac{B_2(\{1-t\})}{2!} t^{\alpha-2}\mathrm{d}t\label{remainderEM:alpha}\\
\widetilde{R}(\alpha) &=& -  \int_{1}^{m} \frac{B_2(\{1-t\})}{2!} t^{\alpha-2}\left(\alpha(\alpha-1) \log t -1 + 2\alpha\right)\mathrm{d}t\label{remainderEM:alphaLog}
\eeqnr
where $B_2(x)=1/6-x+x^2$ is the third Bernoulli polynomial and $\{x\}$ represents the fractional part of the real number $x$. Then, $R(\alpha)$ and $\widetilde{R}(\alpha)$ converge as $m\rightarrow \infty$.
\end{lemma}
Note that  $R(\alpha)$ and $\widetilde{R}(\alpha)$ are the remainders of a first order Euler-Maclaurin expansion of the sums $\sum_{k=1}^m k^{\alpha}$  and $\sum_{k=1}^m k^{\alpha}\log k$, respectively.
%


%
For the proof of Theorem~\ref{ThmExpF} we need the following lemmas.
%
 %
 \begin{lemma} \label{Lemma:DFASum}
 Let $i, k \in \mathbb{N}$, $\alpha>0$, and $\alpha\neq 1, 2$. Then, as $i\rightarrow\infty$,
 \beq
\sum_{k=1}^{i} (i-k) k^{-\alpha} = A_0(\alpha)  i^{2-\alpha} + A_1(\alpha) \, i + O(1)\nonumber
\eeq
where $A_0(\alpha)=\frac{1}{(1-\alpha) (2-\alpha)}$ and $A_1(\alpha)=-\frac{(\alpha+2)(\alpha+3)}{12 (1-\alpha)}+R(\alpha)$
%
and $R(\cdot)$ is defined as in~(\ref{remainderEM:alpha}).
\end{lemma}
 \begin{lemma}\label{Lemma:DFASum2}
 Let $i, k \in \mathbb{N}$. Then, as $i\rightarrow\infty$,
 \beqnr
 \sum_{k=1}^{i} (i-k) k^{-1} &=& i \ln i + \left(-\frac{5}{12}+R_1\right) i +  O(1)\nonumber\\
 \sum_{k=1}^{i} (i-k) k^{-2} &=& \left(\frac{5}{3}+R_2\right) i - \ln i +  O(1)\nonumber
 \eeqnr
 where $R_1 \equiv R(\alpha=-1)$ and $R_2 \equiv R(\alpha=-2)$.
 \end{lemma}
 %

  %
Before proving Theorem~\ref{ThmExpF} we prove the following proposition.
\begin{proposition}\label{PropSigma}
Under the assumptions of Theorem~\ref{ThmExpF}, let us define $\Sigma_m=Cov(Y(i), Y(j))$, 
i.e., the covariance matrix of the integrated process $(Y(1), \ldots, Y(m))$.
Then, 
\begin{itemize}
\item[(i)] if $\beta\neq 2H-1$, then
\beqnr
\Sigma_m=\frac{A}{2H (2H-1)}&&\left(i^{2 H} \left(1+O(i^{-\min(2H-1, \beta)})\right)+j^{2 H} \left(1+O(j^{-\min(2H-1, \beta)})\right)+\right.\nonumber\\
&&\left.-|i-j|^{2 H} \left(1+O(|i-j|^{-\min(2H-1, \beta)})\right)\right]\nonumber
\eeqnr
\item[(ii)] if $\beta=2H-1$, then
\beqnr
\Sigma_m=\frac{A}{2H (2H-1)}&&\left[i^{2 H} \left(1+O(i^{1-2H}\ln i)\right)+j^{2 H} \left(1+O(j^{1-2H}\ln j)\right)+\right.\nonumber\\
&&-\left.|i-j|^{2 H} \left(1+O(|i-j|^{1-2H} \ln |i-j|) \right)\right]\nonumber
\eeqnr
\end{itemize}
\end{proposition}
\begin{proof}[{\bf Proof}]
Under the assumptions on $X(t)$, for $1 \leq i, j, \leq m$ we can write
\beqnr
\Sigma_m &=& Cov(Y(i), Y(j))=\sum_{k=1}^{i} \sum_{l=1}^{j} Cov(X(k), X(l)) = \sum_{k=1}^{i} (i-k) \, \gamma(k) + \nonumber\\
&&+ \sum_{k=1}^{j} (j-k)\, \gamma(k) - \sum_{k=1}^{|i-j|} (|i-j|-k)\, \gamma(k) + \min(i, j) \,D\label{Sigma}
\eeqnr
where $D$ is the variance of the process $X(t)$.
By substituting the explicit functional form for $\gamma(k)$ we get
\beq
\sum_{k=1}^{i} (i - k)\, \gamma(k) \leq A\sum_{k=1}^{i}(i - k)\, k^{2H-2} + M \sum_{k=1}^i (i-k) k^{2H - 2 -\beta}\nonumber
\eeq
for some $M>0$ sufficiently large. By Lemma~\ref{Lemma:DFASum} we have $A\sum_{k=1}^{i}(i - k)\, k^{2H-2} = A \left(\frac{i^{2H}}{2H(2H-1)}+O(i)\right)$.

Now, we consider the following cases:
\bi
\item[(i)]{$\beta\neq2H-1$.} In this case we have to distinguish two cases.

If $\beta\neq2H$, we can use Lemma~\ref{Lemma:DFASum} and obtain
\beqnr
\sum_{k=1}^{i} (i - k)\, k^{2H - 2 -\beta} = A_0(2H-2-\beta) \, i^{2H-\beta}+A_1(2H-2-\beta)\, i + O(1) 
= i^{2H} \,O\left(i^{-\min(2H-1, \beta)}\right)\nonumber
\eeqnr
%
If $\beta=2H$, then we can use Lemma~\ref{Lemma:DFASum2} 
\beqnr
\sum_{k=1}^{i} (i - k)\, k^{2H - 2 -\beta} = \left(\frac{5}{3}+R_2\right)i-\ln i + O(1) =  O\left(i\right)\nonumber
\eeqnr
Then, we repeat the same calculation for the second and third term in (\ref{Sigma}). 
By noting that  $\min(i, j)$ is either of order $O(i)$ or $O(j)$ and putting together all the terms, we obtain the result.

\item[(ii)]{$\beta=2H-1$.}
In this case we can use Lemma~\ref{Lemma:DFASum2}
\beqnr
\sum_{k=1}^{i} (i - k)\, k^{2H - 2 -\beta} = \left(i \ln i + \left(-\frac{5}{12}+R_1\right)\,i+O(1)\right) = i^{2H} O\left(i^{1-2H}\ln i\right)\nonumber
\eeqnr
Then, we repeat the same calculation for the second and third term in (\ref{Sigma}).  
By noting that $\min (i, j)$ is either of order $O(i)$ or $O(j)$ and putting together all the terms, we obtain the result.
\ei
\end{proof}
%

%
 \begin{proof}[{\bf Proof of Theorem~\ref{ThmExpF}}]
First, for $j \in \{ 1, \ldots, [n/m]\}$ let us define the vector:
\beq
Y^{(j)} = (Y(1 + m(j - 1)), \ldots,  Y(mj))^{\top}\,, \nonumber
\eeq
where $x^{\top}$ means the transpose of $x$.

Then, following Bardet and Kammoun (2008), 
\beqnr
F_1^2(m) = \frac{1}{m} (Y^{(1)} - P_{E_1} Y^{(1)})^{\top} (Y^{(1)} - P_{E_1} Y^{(1)})
= \frac{1}{m} (Y^{(1) \top}Y^{(1)} - Y^{(1) \top}P_{E_1}Y^{(1)})\nonumber\,,
\eeqnr
where $E_1$ is the vector subspace of $\mathbb{R}^{m}$ generated by the two vectors $e_1=(1, \ldots, 1)^{\top}$ and $e_2=(1, 2, \ldots, m)^{\top}$, $P_{E_1}$ is the matrix of the orthogonal projection on $E_1$, and the second equality holds because the projection operator is idempotent.
As a consequence,
\beq
\mathbb{E}\left[F_1^2(m)\right]= \frac{1}{m} \left(\mathrm{Tr}(\Sigma_m)- \mathrm{Tr}(P_{E_1} \Sigma_m)\right)\nonumber
\eeq
%
\paragraph{Case (i)} If $\beta\neq 2 H-1$, from Proposition~\ref{PropSigma} we get
\beqnr
\mathrm{Tr}(\Sigma_m)
&=&\frac{A}{H (2 H -1)} m^{2H+1}\left(\int_{0}^{1} x^{2 H} \mathrm{d}x + O\left(m^{-\min(2H-1, \beta)}\right)+O\left(m^{-1}\right)\right)\,,\nonumber
\eeqnr
where the error $O\left(m^{-1}\right)$ comes from approximating the sum with the integral; therefore,
\beq
\mathrm{Tr}(\Sigma_m)=\frac{A m^{2H+1}}{H (2 H -1)} \frac{1}{2H+1} \left(1+O\left(m^{-\min(2H-1, \beta)}\right)\right)\label{TrSigma}
\eeq

For the term $\mathrm{Tr}(P_{E_1} \Sigma_m)$, we use the following representation of the projection operator 
\beq\label{Projection}
P_{E_1} = \frac{2}{m(m - 1)}\left((2 m + 1)- 3(i+j)+6\frac{i\cdot j}{1+m}\right)\,.
\eeq
Then, using (\ref{Projection}) and Proposition~\ref{PropSigma} we can write
\beqnr
\mathrm{Tr}(P_{E_1}\Sigma_m)=&&\frac{2 A m^{2 H +1} m^2 }{m(m-1)2 H(2 H -1)}
\left[\sum_{i=1}^{m}\sum_{j=1}^m \frac{1}{m^2}\left(\left(2+\frac{1}{m}\right)-3\frac{i+j}{m}+6\frac{i j}{m (m+1)}\right)\cdot\right.\nonumber\\
&&\left(\left(\frac{i}{m}\right)^{2 H} \left(1+O\left(i^{-\min(2H-1, \beta)}\right)\right)+\left(\frac{j}{m}\right)^{2 H} \left(1+O\left(j^{-\min(2H-1, \beta)}\right)\right)+\right.\nonumber\\
&&\left.-\left(\frac{|i-j|}{m}\right)^{2 H} \left(1+O\left(|i-j|^{-\min(2H-1, \beta)}\right)\right)\right)\Bigg]\,.\nonumber
\eeqnr
Approximating sums with integrals we get
\beqnr
\mathrm{Tr}(P_{E_1}\Sigma_m)&=&\frac{A m^{2 H +1}}{H(2 H -1)} \left(1+O\left(m^{-\min(2H-1, \beta)}\right)+O\left(m^{-1}\right)\right)\cdot \nonumber\\
&\quad& \cdot \int_0^1\int_0^1\left[\left(2-3(x+y)+6xy\right) \left(x^{2 H}+y^{2H}-|x-y|^{2H}\right)\right]\mathrm{d}x\mathrm{d}y \nonumber\\
&=&\frac{A m^{2 H +1}}{H(2 H -1)}\frac{1 + H (4 + H)}{(1 + H) (2 + H) (1 + 2 H)} \left(1+O\left(m^{-\min(2H-1, \beta)}\right)\right) \label{TrPSigma}
\eeqnr

Putting together (\ref{TrSigma}) and (\ref{TrPSigma}) we obtain
\beq
\frac{1}{m} \left(\mathrm{Tr}(\Sigma_m)- \mathrm{Tr}(P_{E_1} \Sigma_m)\right) = \frac{A}{H(2H-1)}f(H)m^{2H}\left(1+O\left(m^{-\min(2H-1, \beta)}\right)\right)\nonumber
\eeq
which is the formula of $\mathbb{E}\left[F_1^2(m)\right]$.
\paragraph{Case (ii)} If $\beta=2H-1$, the proof is exactly the same, except for replacing all the terms $O\left(i^{-\min(2H-1, \beta)}\right)$ with the terms $O\left(i^{1-2H}\ln i\right)$.
\end{proof}
%

%
 \begin{proof}[{\bf Proof of Corollary~\ref{DFACorFGN}}]
 It follows from the autocovariance of a fractional Gaussian noise (see formula (\ref{eq:fgnACF})). The proof is very similar to the proof of Theorem~\ref{ThmExpF}, and thus omitted. However, a complete proof can be found in Bardet and Kammoun (2008) (see Proof of Property 3.1. therein).
\end{proof}

 \begin{proof}[{\bf Proof of Corollary~\ref{DFACorDiscrete}}]
 It follows from Proposition~\ref{Prop:ACV} and Theorem~\ref{ThmExpF}.
\end{proof}

%



\subsection*{Proofs of lemmas for Section~\ref{estimation}}
 \begin{proof}[{\bf Proof of Lemma~\ref{Lemma:remainderEM}}]
First we prove that $R(\alpha)$ converges. Because $|B_2(\{1-t\})|\leq B_2(0)=1/6$ for all $t$, we have $|R(\alpha)| \leq \frac{|\alpha (\alpha-1)|}{12} \int_{1}^{m} t^{-2+\alpha}\mathrm{d}t$, where the integral on the right-hand side converges as $m\to \infty$ because $2-\alpha>1$.

Now we prove that $\widetilde{R}(\alpha)$ converges. Pick $\varepsilon>0$ s.t. $2-\alpha-\varepsilon>1$. One can always find such $\varepsilon$ because $\alpha<1$ by assumption. Because $\log t$ is slowly varying, there exists $T>1$ s.t. $t^{-2+\alpha}\log t < t^{-2+\alpha+\varepsilon}$ for all $t > T$. Because $|B_2(\{1-t\})|\leq 1/6$ for all $t\geq1$, we can write
\beq
|\widetilde{R}(\alpha)| <   \frac{1}{12} \int_{1}^{T} t^{-2+\alpha}\left(\alpha(\alpha-1) \log t -1 + 2\alpha\right)\mathrm{d}t + \frac{\alpha(\alpha+1)}{12} \int_{T}^{m} t^{-2+\alpha+\varepsilon} \mathrm{d}t\nonumber
\eeq
where the second integral converges as $m\to \infty$ because $2-\alpha-\varepsilon>1$.
\end{proof}
\begin{proof}[{\bf Proof of Lemma~\ref{Lemma:DFASum}}]
By using Euler-Maclaurin formula up to the first order we obtain
\beqnr
i \sum_{k=1}^{i} k^{-\alpha} &=& i \left(\frac{i^{1-\alpha}-1}{1-\alpha} + \frac{i^{-\alpha}+1}{2}+\frac{B_2}{2!} (-\alpha)(i^{-\alpha-1}-1)+R(\alpha)\right)\nonumber
\eeqnr
where $B_2=1/6$ is the third Bernoulli number, and $R(\alpha)$ is the remainder of the Euler-Maclaurin expansion given by (\ref{remainderEM:alpha}). From Lemma \ref{Lemma:remainderEM} we know that $R(\alpha)$ converges, as $i\rightarrow \infty$. So, we can write
 \beq
i \sum_{k=1}^{i} k^{-\alpha} =   \frac{i^{2-\alpha}}{1-\alpha} + A_1(\alpha) i + \frac{i^{1-\alpha}}{2} - \frac{\alpha \, i^{-\alpha}}{12}\nonumber
\eeq
where $A_1(\cdot)$ is defined as in Lemma~\ref{Lemma:DFASum}. Note that the second-order term is $O(i)$.

Similarly, we can write
\beq
\sum_{k=1}^{i} k^{1-\alpha} = \frac{i^{2-\alpha}}{2-\alpha} + \frac{i^{1-\alpha} }{2} + A_1(\alpha-1) + \frac{(1-\alpha) \, i^{-\alpha}}{12}\nonumber
\eeq
Note that in this case the second-order term is $O(i^{1-\alpha})$.

Putting together these two terms we have $\sum_{k=1}^{i} (i-k) k^{-\alpha} = A_0(\alpha) i^{2-\alpha}+A_1(\alpha) i + O(1)$.
%
%
where $A_0(\cdot)$ is defined as Lemma~\ref{Lemma:DFASum}. 
Note that the terms of order $i^{1-\alpha}$ cancel out exactly.
\end{proof} 
 \begin{proof}[{\bf Proof of Lemma~\ref{Lemma:DFASum2}}]
The proof is similar to the proof of Lemma~\ref{Lemma:DFASum}, and thus omitted.
\end{proof} 
%


\section{Sign Process}\label{App:Sign}
Taking the sign of a stochastic process can be thought of as an extreme form of discretization. Hence, to study the asymptotic properties of the sign process we can use the same technique outlined in Section~\ref{SubSec:ACV}  for general nonlinear transformations of Gaussian processes. By decomposing the sign transformation on the basis of Hermite polynomials we get the following

\begin{proposition}\label{Prop:ACVSign}
Let $\{X(t)\}_{t \in \mathbb{N}}$  be a stationary Gaussian process with autocovariance function given by Definition \ref{generalDef}. Then, the autocovariance $\gamma_s(k)$ and the autocorrelation $\rho_s(k)$ of the sign process satisfy
\beq
\gamma_s(k)=\rho_s(k)=\frac{2}{\pi}\arcsin \rho(k)
\eeq
\end{proposition}
Therefore, also the sign transformation preserves the long memory property and the Hurst exponent. Moreover, if the autocorrelation $\rho$ is small (for example if the lag $k$ is large) we have 
\beq
\gamma_s(k)=\rho_s(k) = \frac{2}{\pi}\left(\rho(k)+\frac{\rho^3(k)}{6}\right)+ O(\rho^5(k)).
\eeq
This expression has been obtained several times, as, for example, in the context of binary time series (see Keenan (1982)). Note that, trivially, when the discretization is obtained by taking the sign function the variance of the discretized process is $D_s=1$.

All the results on the discretized process presented above hold true for the sign process as well, with $\left(\frac{\vartheta_2(0,e^{-1/2\chi}) }{\sqrt{2 \pi \chi}}\right)^2$ replaced by $\frac{2}{\pi D}$ and $g_3=-\frac{1}{\sqrt{3\pi}}$. Numerical results are available from the authors by request.


%
\begin{proof}[{\bf Proof of Proposition~\ref{Prop:ACVSign}}]
As the discretization, the sign transformation is an odd function, and therefore $g_j=0$ when $j$ is even. When $j$ is odd the coefficients of the sign function in Hermite polynomials are
\begin{equation}
g_j= 2\int_0^\infty H_j(x) \frac{e^{-x^2/2}}{\sqrt{2\pi}}dx=\frac{2^{j+1/2} \Gamma(j+1/2)}{\pi \sqrt{(2j+1)!}}.\nonumber
\end{equation} 
By inserting these value in (\ref{hexpansion}) we obtain the autocorrelation (and autocovariance) function of the sign of a Gaussian process 
\begin{equation}
\gamma_s(k)=\rho_s(k)=\sum_{j=1}^{\infty}g_j^2 \rho(k)^j=\frac{2}{\pi}\arcsin \rho(k)\nonumber
\end{equation} 
\end{proof}

\end{appendices}

\newpage



\end{document}